\documentclass[11pt,english,american]{article}

\usepackage[colorlinks=true,linkcolor=black,
citecolor=black]{hyperref}
\usepackage{amsmath,amsfonts,amssymb,amsthm,mathrsfs}
\usepackage{mathabx}
\usepackage{setspace}
\usepackage{soul}
\usepackage{comment}

\theoremstyle{plain}\newtheorem{theorem}{Theorem}[section]
\theoremstyle{plain}\newtheorem{lemma}[theorem]{Lemma}
\theoremstyle{plain}\newtheorem{corollary}[theorem]{Corollary}
\theoremstyle{plain}
\theoremstyle{plain}\newtheorem{proposition}[theorem]{Proposition}
\theoremstyle{definition}
\theoremstyle{remark}
\theoremstyle{definition}\newtheorem{def:and:lemma}[theorem]{Definition and Lemma}

\usepackage[colorinlistoftodos,prependcaption,textsize=tiny]{todonotes}

\counterwithin{remark}{section}

\renewcommand{\Re}{\textnormal{Re}}
\renewcommand{\Im}{\textnormal{Im}}
  
\renewcommand{\d}{\mathrm d k}

\newcommand{\lsp}{\big \langle }
\newcommand{\rsp}{\big \rangle }

\newcommand{\<}{\langle} 
\renewcommand{\>}{\rangle}

\newcommand{\ve}{\varepsilon}
\renewcommand{\d}{\textnormal{d}}
\newcommand{\T}{\mathrm T}

\numberwithin{equation}{section}

\begin{document}

\bibliographystyle{alpha}

\title{\huge {Adiabatic Klein--Gordon Dynamics for the Renormalized Nelson Model}}

\author{Morris Brooks\thanks{Institut für Mathematik,~Universit\"at Z\"urich,~Winterthurerstrasse~190, CH-8057 Z\"urich, Switzerland. E-mail: \texttt{morris.brooks@math.uzh.ch}} \qquad \and \qquad David Mitrouskas\thanks{Institute of Science and Technology Austria (ISTA), Am Campus 1, 3400 Klosterneuburg, Austria. E-mail: \texttt{mitrouskas@ist.ac.at}}}

\date{\today}

\maketitle


\frenchspacing
\begin{spacing}{1.15} 
\begin{abstract}We study the renormalized Nelson model in a semiclassical regime where the field becomes classical while the particle remains quantum. The degree of classicality is measured by a small parameter $\varepsilon \ll 1$. In this scaling the particle evolves on microscopic times, whereas the field exhibits nontrivial dynamics only on the macroscopic scale $t=\mathcal{O}(\varepsilon^{-2})$. The natural semiclassical model is the coupled Schrödinger--Klein--Gordon (SKG) system, which encodes the time-scale separation through an explicit $\varepsilon$-dependence. Based on this scale separation in SKG, we apply the adiabatic principle to derive a new PDE for the classical field, the $\varepsilon$-free  adiabatic Klein--Gordon (aKG) equation, where the field is driven by the instantaneous ground state of the particle.
Our main result is a norm approximation of the Nelson dynamics by the aKG solution corrected by a quasi-free fluctuation dynamics around the classical field, generated by a renormalized Bogoliubov--Nelson Hamiltonian. As a corollary, we obtain convergence of the reduced one-body densities for both subsystems, where the fluctuation correction vanishes, thereby justifying aKG as a semiclassical Born--Oppenheimer type approximation of the renormalized Nelson model.
\end{abstract}

\tableofcontents

\allowdisplaybreaks

\section{Introduction}

The Nelson model describes nonrelativistic particles interacting with a quantized scalar boson field. Introduced by E. Nelson in 1964, it has long served as a benchmark in mathematical physics for studying ultraviolet divergences and renormalization at the Hamiltonian level, with motivations ranging from quantum electrodynamics to solid-state models. Beyond these aspects, it also provides a natural setting to investigate semiclassical limits where quantum matter interacts with emergent classical fields.

We investigate the single-particle, massless, renormalized Nelson model in a semiclassical regime where the quantized field behaves effectively classically, while the particle retains its quantum nature. The degree of classicality of the field is measured by a small parameter $\varepsilon \ll 1$, which can be viewed as an effective Planck constant. In this regime, the particle evolves on its intrinsic quantum scale $t=\mathcal{O}(1)$, during which the field remains essentially frozen. In contrast, on macroscopic times $t=\mathcal{O}(\varepsilon^{-2})$, the dynamics of the classical field become visible. This scale separation between the fast particle and the slow field underlies our analysis  and highlights the relevance of a long-time description.

The canonical semiclassical model of a quantum particle coupled to a classical field is the Schrödinger--Klein--Gordon (SKG) system,
\begin{align}
 i \partial_t \psi_t (x) &\; =\; \ve^{-2} \left( - \Delta \;+\; 2 \Re \int_{\mathbb R^3} \frac{e^{-ikx}}{\sqrt{|k|}}\, \varphi_t(k)\,  \d k \right) \psi_t(x) ,\label{eq:SKG:1}\\[0mm]
i \partial_t \varphi_t(k) & \; =\;   |k|\, \varphi_t(k) \;+\;  \int_{\mathbb{R}^3} 
   \frac{e^{-ikx}}{\sqrt{|k|}}\,|\psi_t(x)|^2 \,\d x  .\label{eq:SKG:2}
\end{align}
The first equation describes the fast quantum particle in the external potential generated by the field, while the second is a Klein--Gordon equation for the slow classical field driven by the particle density. The SKG system still carries an explicit $\varepsilon$-dependence, which encodes the time-scale separation between the particle and the field. This dependence on the semiclassical parameter, however, also makes SKG somewhat inconvenient as a semiclassical approximation of the full quantum model.

For initial data where the particle is prepared in the gapped ground state of the potential generated by the classical field, one can exploit the adiabatic principle to obtain an $\varepsilon$-free semiclassical description. The idea is to eliminate the fast particle by letting it adiabatically follow the instantaneous ground state associated with the field. This yields an $\ve$-independent, closed evolution for the field alone, the \emph{adiabatic Klein--Gordon (aKG) equation},
\begin{align}
 i\partial_t\varphi_t(k) \;=\; |k|\,\varphi_t(k) \;+\; \int_{\mathbb{R}^3} 
   \frac{e^{-ikx}}{\sqrt{|k|}}\,|\psi_{\varphi_t}(x)|^2 \,\d x,
\end{align}
where $\psi_{\varphi_t}$ denotes the ground state of the Schrödinger operator with external potential generated by $\varphi_t$. The aKG equation governs the autonomous evolution of the field on the macroscopic time scale, and can be viewed as a field-theoretic Born--Oppenheimer type approximation of SKG. To our knowledge, this equation has not been studied before. We establish local well-posedness and show that  SKG converges to aKG as $\varepsilon \to 0$.

Our main result is a quantitative norm approximation of the renormalized Nelson dynamics, valid for suitable initial states up to times $t=\mathcal{O}(\varepsilon^{-2})$. The effective description is given by the aKG pair $(\psi_{\varphi_t}, \varphi_t)$, corrected by quantum fluctuations around the classical field. These fluctuations are governed by a \emph{quadratic} (Bogoliubov) Hamiltonian and hence generate a \emph{quasi-free} evolution on Fock space. The naive quadratic generator, however, inherits the logarithmic ultraviolet divergence of the Nelson Hamiltonian, and one of the key ingredients is the renormalization of this operator, which we call the \emph{Bogoliubov--Nelson Hamiltonian}.  

As a consequence of the norm approximation, we obtain convergence of the reduced one-body densities for both particle and field. At this level the fluctuation correction disappears, thereby justifying aKG as a semiclassical Born--Oppenheimer type approximation of the renormalized Nelson model. Moreover, through convergence of SKG to aKG as $\ve \to 0$, our results directly imply corresponding approximations in terms of the ($\varepsilon$-dependent) SKG solutions.

The remainder of this introduction is devoted to precise formulations of the relevant models: the renormalized Nelson Hamiltonian, the adiabatic Klein--Gordon equation, and the renormalized quantum fluctuation dynamics around the aKG field.

\subsection{The renormalized Nelson Hamiltonian}

We consider a single particle coupled to a quantized scalar field. The Hilbert space is 
\begin{align}
\mathscr H = L^2(\mathbb R^3_x) \otimes \mathcal F,
\end{align}
where \( L^2(\mathbb{R}^3_x)\) denotes the particle space with position variable $x$, and $\mathcal F := \Gamma_{\mathrm{sym}}(L^2(\mathbb{R}^3_k))$ is the bosonic Fock space constructed over the one-particle space with momentum variable $k$. We write \(a_k^\ast\) and \(a_k\) for the creation and annihilation operators satisfying the usual bosonic commutation relations \([a_k,a_\ell^\ast]=\delta(k-\ell)\) and \([a_k,a_\ell]=[a_k^\ast,a_\ell^\ast]=0\). The free particle and field Hamiltonians are
\begin{align}
p^2:= - \Delta , \qquad H_f:=\int_{\mathbb{R}^3}\omega(k)\,a_k^\ast a_k\,dk, \qquad \omega(k):=|k|.
\end{align}

Formally, the massless Nelson Hamiltonian with semiclassical parameter \(\varepsilon \ll 1\) is given by
the expression
\begin{align}\label{eq:HN}
 p^2 \;+\; \varepsilon\, \int_{\mathbb R^3} \Big( \bold G(x,k) a_k^* + \overline{ \bold G(x,k) } a_k \Big) \d k \;+ \; \varepsilon^2 H_f,
\end{align}
acting on $L^2(\mathbb R^3_x) \otimes \mathcal F$, with $x$-dependent kernel
\begin{align}\label{eq:def:G}
\mathbf{G}(x,k):=e^{-ikx} \, \omega(k)^{-1/2}.
\end{align}
Note that the powers of $\ve$ match the number of creation and annihiliation operators in each term, which shows that $\varepsilon\to 0$ corresponds to a semiclassical limit.

Since $\omega^{-1} G \notin L^2(\mathbb R^3_k)$, the interaction is not form-bounded with respect to free Hamiltonian $p^2 + H_f$, and the formal Hamiltonian \eqref{eq:HN} is UV singular, hence neither self-adjoint nor associated with a closed quadratic form.

Introducing a cutoff $\Lambda\in(0,\infty)$ and
\begin{align}
\mathbf{G}_\Lambda(x,k):=  \bold G(x,k)\, \mathbf{1}_{ |k|\le\Lambda} .
\end{align}
we define the regularized Nelson Hamiltonian
\begin{align}\label{eq:regularized:Nelson}
 p^2 \;+\; \varepsilon\, \int_{\mathbb R^3} \Big( \bold G_\Lambda(x,k) a_k^* + \overline{ \bold G_\Lambda(x,k) } a_k \Big) \d k \;+ \; \varepsilon^2 H_f,
\end{align}
which is self-adjoint and semibounded with quadratic form domain \( Q(H_\Lambda^\varepsilon) = Q(p^2+H_f)\) for all $\Lambda \in (0,\infty)$. 

The following proposition shows that, up to a diverging energy shift, the regularized operators converge as $\Lambda \to \infty$. We call the limit operator $H^\varepsilon$ the \textit{renormalized Nelson Hamiltonian} in the semiclassical scaling.

\begin{proposition} \label{prop:ren:Nelson}
For every $\ve >0$ there is a self-adjoint, semibounded operator \(H^\varepsilon\) such that
\begin{align}\label{eq:renormalizated:Nelson:evolution}
\exp\big( -i t H^\ve_\Lambda  \big) \exp\big( -i t  4\pi \ve^2  \ln \Lambda \big)    \xrightarrow[\Lambda \to \infty]{} \exp\big(-i t H^{\ve} \big)
\end{align}
for every \(t\in\mathbb{R}\), in the strong sense on $L^2(\mathbb R^3_x) \otimes \mathcal F$.
\end{proposition}

This was first established in \cite{Nelson64} and later refined in \cite{Griesemer18}. The method is based on a unitary dressing transformation, which will also play a central role in our analysis. For completeness, we include a proof of Proposition~\ref{prop:ren:Nelson} in Appendix~\ref{app:renormalization}, where we adapt the ideas of \cite{Griesemer18} to the semiclassical setting. Alternative approaches to renormalization of the Nelson model rely on functional integration \cite{GubinelliHJ2014} or interior boundary conditions \cite{LS2019}.  Finally, we note the drastic effect of renormalization on the domain: as shown in \cite{Griesemer18,LS2019}, the form domains intersect only trivially, that is \(Q(H^\ve)\cap Q(H^\ve_\Lambda)=\{0\}\) for all \(\Lambda\in(0,\infty)\).

\subsection{The adiabatic Klein--Gordon equation}

The semiclassical model associated with the Nelson Hamiltonian is obtained by replacing the creation and annihilation operators in \eqref{eq:HN} with a classical field $\ve^{-1}\varphi\in L^2(\mathbb{R}^3_k)$. Taking the expectation value in a normalized particle state $\psi\in H^1(\mathbb{R}^3_x)$ then yields the $\varepsilon$-independent energy functional
\begin{align}\label{eq:semiclassical:energy}
\mathcal E(\psi , \varphi ) = \int \d x \, |\nabla \psi(x) |^2 + \int \d k \, \omega(k) |\varphi(k)|^2 + \iint \d k \d x \, | \psi (x)|^2 \frac{2 \Re (e^{-ikx}  \varphi (k) ) }{\sqrt{\omega(k)}}.
\end{align}
Unlike the quantum model, the semiclassical energy is ultraviolet regular. Its Hamiltonian flow gives the Schrödinger--Klein--Gordon (SKG) system for $(\psi_t,\varphi_t)\in L^2(\mathbb{R}^3_x)\times L^2(\mathbb{R}^3_k)$,
\begin{align}\label{eq:SKG}\tag{SKG}
\begin{cases}
\begin{aligned}
\ i \partial_t \psi_t & = h_{\varphi_t} \psi_t \\[2mm]
\ i \ve^{-2} \partial_t \varphi_t & =  \omega \varphi_t + \sigma ( \psi_t ) 
\end{aligned}
\end{cases}
\end{align}
with Schrödinger operator
\begin{align}\label{eq:def:h}
h_{\varphi_t} = p^2 + V_{\varphi_t} \qquad \text{with} \qquad V_{\varphi_t}(x) = 2 \Re \langle \varphi_t , \bold G(x,\cdot)  \rsp 
\end{align}
and source term
\begin{align}\label{eq:def:source}
\sigma ( \psi_t ) (k) =   \langle \psi_t, \bold G (\cdot, k) \psi_t\rangle
\end{align}
where $\bold G$ is given by \eqref{eq:def:G}. 

Note that \eqref{eq:SKG} is equivalent to \eqref{eq:SKG:1}–\eqref{eq:SKG:2} after rescaling time by $\varepsilon^{-2}$. A key feature of the SKG system is its explicit $\varepsilon$-dependence: although the semiclassical energy is $\varepsilon$–free, the underlying Poisson/commutator structure forces the field to evolve on the slow (macroscopic) time scale $t = \mathcal O( \varepsilon^{-2})$, while the particle evolves on the fast (microscopic) scale $t =  \mathcal O( 1 ) $. Thus the SKG solutions depend explicitly on $\varepsilon$, which makes them less convenient as a semiclassical approximation of the Nelson dynamics in the limit $\varepsilon\to 0$. This motivates the formulation of an $\varepsilon$-free effective description.

To achieve this, we exploit the explicit time-scale separation in \eqref{eq:SKG}: inspired by the adiabatic principle, we eliminate the fast particle by letting it adiabatically follow the instantaneous ground state associated with the field. This leads to a closed, $\varepsilon$-independent evolution for the field alone, which we call the \emph{adiabatic Klein--Gordon} (aKG) equation,
\begin{equation}\label{eq:aKG}\tag{aKG}
i \partial_t \varphi_t \;=\; \omega \varphi_t \; +\; \sigma(\psi_{\varphi_t}),
\end{equation}
where $\psi_{\varphi_t}\ge 0$ denotes the normalized ground state of $h_{\varphi_t}=p^2+V_{\varphi_t}$. The structure of aKG mirrors the classical Born--Oppenheimer approximation: the fast quantum particle adiabatically follows its instantaneous ground state, while the slow field evolves autonomously. This is analogous to molecular Born--Oppenheimer theory, where the fast electrons follow the electronic ground state determined by the nuclear configuration, and the heavy nuclei evolve on their intrinsic slow time scale; see \cite{Hagedorn1,SpohnTeufel,PST}.

Before continuing, we introduce the energy spaces for the particle and the field
\begin{alignat}{2}
H^r(\mathbb{R}^3_x)    &  := \big\{ \psi \in L^2(\mathbb{R}^3_x) : \|(1-\Delta)^{r/2} \psi \|_{L^2} < \infty \big\}, \quad & r\in \mathbb R, \\[1mm]
\mathfrak{h}_s  &:= \big\{ \varphi \in L^2(\mathbb{R}^3_k) : \|\omega^s \varphi\|_{L^2} < \infty \big\}, \quad & s\in\mathbb{R}.
\end{alignat}
For brevity, we write $H^r := H^r(\mathbb{R}^3_x)$. These spaces will be used throughout.

For all applications that follow, we consider initial data where the particle is in the gapped ground state of the Schr\"odinger operator $h_\varphi=p^2+V_\varphi$. To this end, we assume that
\begin{equation}\label{ass:negative:ev}
e(\varphi)\;:=\;\inf \sigma(h_\varphi)\;<\;0,
\end{equation}
which guarantees that $h_\varphi$ has a unique ground state $\psi_\varphi\in H^2(\mathbb{R}^3)$. Without loss of generality, we may take $\psi_\varphi\ge0$. We also define the spectral gap
\begin{align}\label{eq:definition:spectral:gap}
\triangle(\varphi) \,=\, \inf\big\{\,|e(\varphi)-\lambda| : \lambda\in\sigma(h_\varphi)\setminus\{e(\varphi)\}\,\big\},
\end{align}
that is, the distance between the ground state energy $e(\varphi)$ and the rest of the spectrum. 

Let $I\subset \mathbb R $ be an interval with $0 \in I $. We say that \(\varphi \in C(I;\mathfrak{h}_{1/2})\) is a solution to \eqref{eq:aKG} on the interval $I$ with initial datum \(\varphi_0\in\mathfrak{h}_{1/2}\), if \(e(\varphi_t)<0\) for all $t\in I$ and
\begin{equation}\label{eq:aKG-mild}
\varphi_t \;=\; e^{-i\omega t}\varphi_0 \;-\; i\int_0^t e^{-i\omega (t-s)}\,\sigma(\psi_{\varphi_s})\, \d s.
\end{equation}
The existence of the ground state $\psi_{\varphi_t}$ for all $t \in I$ follows from the condition $e(\varphi_t)<0$.

The next lemma establishes local well-posedness of the aKG equation.

\begin{proposition}\label{prop:aKG-local}
Let $\varphi_0\in\mathfrak h_{1/2}$ with $e(\varphi_0)<0$. 
There exist maximal times $T_-, T_+ \in(0,\infty]$ such  that \eqref{eq:aKG}
admits a unique solution
\begin{align}
 \varphi \in C((-T_-,T_+ );\mathfrak h_{1/2})
 \ \cap\ 
 C^1((-T_-,T_+);\mathfrak h_{-1/2}).
\end{align}
In particular, we have  $e(\varphi_t)<0$ for all $t\in(-T_- ,T_+ )$, and therefore $h_{\varphi_t}$ has a unique ground state $\psi_{\varphi_t} $ for every such $t$. Moreover, for any compact interval $J \subset (-T_-,T_+)$ the spectral gap remains uniformly positive,
\begin{align}
\inf_{ t \in J }\triangle(\varphi_t) \;>\; 0.
\end{align}
If in addition $\varphi_0\in \mathfrak{h}_0\cap\mathfrak{h}_{1/2}$, then $\varphi \in C((-T_-, T_+);\mathfrak{h}_0)$.

Finally, we either have $T_{+} = \infty $, or else $T_{+}<\infty$ and $\liminf_{ t \to T_{+}} \triangle (\varphi_t) = 0$, and equivalently for $T_-$.
\end{proposition}
The proof is given in Section~\ref{sec:properties:adiabatic_SKG}.
In general, global well-posedness cannot be expected, since the spectral gap $\triangle(\varphi_t)$ may close at a finite time. Once the gap vanishes, the ground state map $\varphi_t \mapsto \psi_{\varphi_t}$ is typically no longer well defined; moreover, continuation of the local solution fails in the absence of a positive gap. 

Conversely, if the spectral gap remains uniformly positive for all times, then the solution extends globally. 
This is in particular the case for stationary solutions, which are the minimizers of the classical 
field functional
\begin{align}\label{eq:classical:field:energy}
\mathcal E_{\mathrm{field}}(u) := e(u) + \|u\|_{\mathfrak h_{1/2}}^2 .
\end{align}
These minimizers are known to be unique up to spatial translations \cite{Lieb1977}. 
We denote by $\varphi_*$ a fixed minimizer, and by
\begin{align}
\mathcal M := \{\,T_y \varphi_* : y\in\mathbb R^3\}, 
\qquad (T_y u)(k):=e^{-ik y}u(k),
\end{align}
the manifold of all minimizers. 
For $u\in\mathfrak h_{1/2}$ we define the distance to $\mathcal M$ by
\begin{align}
\mathrm{dist}_{\mathfrak h_{1/2}}(u,\mathcal M) := \inf_{y\in\mathbb R^3} 
   \|u-T_y\varphi_*\|_{\mathfrak h_{1/2}}.
\end{align}

In the next lemma, we prove that for initial data that is sufficiently close to $\mathcal M$, the spectral gap remains uniformly positive along the evolution. As a consequence, the corresponding aKG solution exists for all times. The proof of the lemma is postponed to Appendix \ref{app:persistence}.

\begin{lemma} \label{lem:global} There exist $\eta>0$ such that for all initial data $\varphi_0\in\mathfrak h_{1/2}$ with  $\mathrm{dist}_{\mathfrak h_{1/2}}(\varphi_0,\mathcal M) \;\le\; \eta$, the aKG solution $\varphi_t$ extends globally in time and satisfies
\begin{align}
\inf_{t\in \mathbb R} \triangle (\varphi_t) \,  > \, 0.
\end{align}
\end{lemma}

Even though aKG is an equation for the field alone, we will often refer to $(\psi_{\varphi_t},\varphi_t)$ as the aKG solution pair. The main goal of this work is to show that this pair provides an effective description of the renormalized Nelson dynamics in the limit $\varepsilon\to 0$. At the same time, it is instructive to make the connection between SKG and aKG explicit. In Appendix~\ref{app:conv:SKG:aKG}, we therefore show that the SKG solution $(\psi_t^\varepsilon,\varphi_t^\varepsilon)$, when observed on the rescaled slow time scale, converges to the aKG pair $(\psi_{\varphi_t},\varphi_t)$. More precisely, on compact intervals in $(-T_-,T_+)$,
\begin{align}
\exp \left( i \ve^{-2} \int_0^t e \left(\varphi^\ve_{\ve^{-2}s} \right) \d s  \right) \, \psi_{\ve^{-2} t }^\ve \xrightarrow[\ve \to 0]{}  \psi_{\varphi_t}  
\qquad\text{and}\qquad \varphi_{\ve^{-2} t }^\ve \xrightarrow[\ve \to 0]{} \varphi_t
\end{align}
both in $L^2$-norm, where $e(\varphi_s^\varepsilon):=\inf\sigma(h_{\varphi_s^\varepsilon})$. Although this result is not used in our proofs, it makes the adiabatic picture precise and may be of independent interest. In particular, it allows us to translate the main results of Section~\ref{sec:main:results} into SKG-based approximations.

\subsection{The renormalized quantum fluctuations}
\label{sec:ren:bog:ev}

Let $\varphi_t$ denote the solution to \eqref{eq:aKG} with initial condition $\varphi_0 \in \mathfrak h_{1/2}$ satisfying $e(\varphi_0)<0$, and let $\psi_{\varphi_t}\ge 0$ be the instantaneous ground state of $h_{\varphi_t}$. By Proposition~\ref{prop:aKG-local} there exists $T>0$ such that $h_{\varphi_t}-e(\varphi_t)$ has a uniform spectral gap above zero for all $|t|\le T$. In particular, we can define the reduced resolvent on $L^2(\mathbb{R}^3_x)$, given by
\begin{align}\label{eq:def:reduced:resolvent}
R_{\varphi_t} := Q_{\varphi_t}\,( h_{\varphi_t}-e(\varphi_t))^{-1}\,Q_{\varphi_t},
\qquad
Q_{\varphi_t}=\mathbf{1}-P_{\varphi_t},\quad P_{\varphi_t}=|\psi_{\varphi_t}\rangle\langle \psi_{\varphi_t}|,
\end{align}
which is a bounded operator for all $|t|\le T$ thanks to the spectral gap.

For the regularized Nelson Hamiltonian \eqref{eq:regularized:Nelson} with UV cutoff $\Lambda\in(0,\infty)$, the evolution of the quantum fluctuations is generated by the time-dependent quadratic (Bogoliubov) Hamiltonian acting on Fock space $\mathcal{F}$,
\begin{align}\label{eq:Bog:cutoff}
\mathbb{H}_\Lambda(t)\;=\;H_f\;- \iint \, \mathrm{d}k\, \d \ell \,
\< \psi_{\varphi_t},\bold G_\Lambda(\cdot, k)   R_{\varphi_t} \bold G_\Lambda(\cdot,\ell),\psi_{\varphi_t}\> \,
\big(a_k^\ast+a_{-k}\big)\big(a_\ell^\ast+a_{-\ell}\big) 
\end{align}
where $H_f=\mathrm{d}\Gamma(\omega)$ is the free-field energy and $\<\cdot, \cdot\>$ denotes the inner product in $L^2(\mathbb R^3_x)$.

We denote by $\mathbb{U}_\Lambda(t)$ the unitary propagator on $\mathcal{F}$ generated by $\mathbb{H}_\Lambda(t)$ with initial time $t=0$. The next result establishes the existence of a renormalized Bogoliubov evolution in the limit $\Lambda\to\infty$; in particular, it shows that $\mathbb{H}_\Lambda(t)$ exhibits the same logarithmic divergence as the original Nelson Hamiltonian.

\begin{proposition}\label{prop:Bogo:renormalization}
Let $\varphi_0\in \mathfrak h_{1/2}$ with $e(\varphi_0)<0$, and let $\varphi_t$ denote the solution to \eqref{eq:aKG} with initial datum $\varphi_0$. There exists $T>0$ such that the following statements hold:
\begin{itemize}
\item[(i)] For every $\Lambda\in[2,\infty)$ and $\Psi\in D(H_f^{1/2})$ there exists a unique solution $\Psi\in C([-T,T],\mathcal F)\cap L^\infty([-T,T],D(H_f^{1/2}))$
to the Cauchy problem
\begin{align}
\begin{cases}
\begin{aligned}
i\partial_t \Psi(t) &= \mathbb{H}_\Lambda(t)\,\Psi(t),\\[1.2mm]
\Psi(0) &= \Psi,
\end{aligned}
\end{cases}
\end{align}
with the equation understood in $D((1+H_f)^{-1/2})$. The map $\Psi\mapsto \Psi(t)$ defines a unitary $\mathbb{U}_\Lambda(t)$ on $\mathcal{F}$ that is strongly continuous in $t$. 
\item[(ii)] For all $|t|\le T$, the strong limit
\begin{align}\label{eq:renormalized:bog:evolution}
\mathbb{U}(t)\;:=\;s - \lim_{\Lambda\to\infty}\,\mathbb{U}_\Lambda(t)\,\exp\big(-it\,4\pi\ln\Lambda\big)
\end{align}
exists, is unitary, and the map $t\mapsto \mathbb{U}(t)$ is strongly continuous.
\end{itemize}
\end{proposition}

The proposition is proved in Section \ref{sec:bog:ren}. We call $\mathbb{U}(t)$ the \emph{renormalized Bogoliubov--Nelson evolution} on $\mathcal{F}$. 
While we do not explicitly use the following property in our analysis, one can show that $\mathbb{U}(t)$ acts as a Bogoliubov transformation on Fock space, i.e., it implements a linear rotation of creation and annihilation operators. In particular, $\mathbb U(t)$ preserves the Gaussian (Wick factorization) property of quasi-free states in $\mathcal F$. 

\section{Main results}
\label{sec:main:results}

We consider \emph{semiclassical product initial data} of the form
\begin{align}\label{eq:initial:states}
\Psi^\varepsilon(\varphi_0) \;=\; \psi_{\varphi_0}\ \otimes\ W(\varepsilon^{-1}\varphi_0)\,\Omega,
\end{align}
where $\psi_{\varphi_0}\ge0$ is the normalized ground state of $h_{\varphi_0}=p^2+V_{\varphi_0}$ and $\Omega = (1, \bold 0)$ denotes the Fock vacuum in $\mathcal F$. 
Here 
\begin{align}
W(f):=\exp\!\big(a^\ast(f)-a(f)\big)
\end{align} 
is the Weyl operator, and thus the coherent state $
W(\varepsilon^{-1}\varphi_0)\Omega$
represents the classical field $\ve^{-1}\varphi_0$ with $\mathcal O(\varepsilon^{-2})$ excitations.

From a physical perspective, such states arise naturally in experimental situations: one first \emph{imposes} a classical field  $\varphi_0$ (e.g. by external driving) and keeps it fixed until the particle is cooled into the ground state $\psi_{\varphi_0}$ of the frozen Hamiltonian $h_{\varphi_0}=p^2+V_{\varphi_0}$.  
At time $t=0$ the external control is removed and the system evolves autonomously under  the renormalized Nelson Hamiltonian, with the initial state prepared in $\Psi^\ve (\varphi_0) $.
Unless the stationary condition $\omega \varphi_0+\sigma(\psi_{\varphi_0})=0$ is satisfied, this preparation constitutes a quench from a frozen to a non-stationary system, initiating nontrivial coupled dynamics on the macroscopic time scale.

Our main theorem gives a norm approximation for the renormalized Nelson evolution for initial states of the form \eqref{eq:initial:states}. 
As corollaries, we obtain convergence of the one-body reduced densities for both subsystems.

\subsection{Norm approximation}

The following theorem is the main result of this work.

\begin{theorem}\label{thm:norm:approximation}
Let $\varphi_0\in \mathfrak h_0\cap \mathfrak h_{1/2}$ with $e(\varphi_0)<0$, and let $\varphi_t$ be the solution to \eqref{eq:aKG} with initial datum $\varphi_0$.  Let $\psi_{\varphi_t}\ge 0$ denote the normalized ground state of $h_{\varphi_t}=p^2+V_{\varphi_t}$, and set
\begin{align}
\mu(t)=\int_0^t\Big(\Im\langle \varphi_s,\partial_s \varphi_s\rangle +\|\omega^{1/2}\varphi_s\|_{L^2}^2+e(\varphi_s)\Big)\,\mathrm{d}s .
\end{align}
Moreover, let $H^\varepsilon$ be the renormalized Nelson Hamiltonian from Proposition~\ref{prop:ren:Nelson}, and let $\mathbb U(t)$ be the renormalized Bogoliubov--Nelson evolution from Proposition~\ref{prop:Bogo:renormalization} (defined in terms of $\varphi_t$). Then there exist constants $C,T,\varepsilon_0>0$, depending only on $\varphi_0$, such that
\begin{align*}
\Big\|
\exp\!\left(-i \varepsilon^{-2} t H^\varepsilon + i \varepsilon^{-2}\mu(t)\right)\,(\psi_{\varphi_0}\otimes W(\varepsilon^{-1}\varphi_0)\Omega)
-
\psi_{\varphi_t}\otimes W(\varepsilon^{-1}\varphi_t)\,\mathbb U(t)\Omega
\Big\|
\;\le\; C \,\varepsilon^{\tfrac{1}{10}} \ln \tfrac{1}{\varepsilon} 
\end{align*}
for all $|t|\le T$ and $\varepsilon\in(0,\varepsilon_0]$.
\end{theorem}

This result provides a norm approximation of the renormalized Nelson dynamics up to times 
$t=\mathcal{O}(\varepsilon^{-2})$, in terms of a semiclassical state built from the aKG pair 
$(\psi_{\varphi_t},\varphi_t)$ and corrected by the fluctuation dynamics $\mathbb U(t)$ around the coherent aKG field. 
The theorem is formulated for $|t|\leq T$ (corresponding to physical times $|t|\leq \varepsilon^{-2}T$), 
where $T$ is tied to the existence of the aKG solution and the persistence of the associated spectral gap; see Proposition \ref{prop:aKG-local}.
In general, $T$ may have to be chosen small. By Lemma~\ref{lem:spectral:gap}, however, if the initial datum 
is sufficiently close to the manifold $\mathcal M$ of field minimizers, the aKG solution is global and the 
gap stays uniformly positive; in this case one can choose $T$ arbitrarily large, so the approximation 
holds on any finite macroscopic interval.

\subsection{Convergence of reduced densities}

As corollaries of Theorem~\ref{thm:norm:approximation}, we derive convergence of the reduced one-body densities. 

We start with the particle, whose reduced density at time $t$ is defined by
\begin{align}\label{eq:def-gamma-particle}
\gamma_{\rm particle}^\varepsilon(t) 
   := \operatorname{Tr}_{\mathcal F}\Big( 
      e^{-i \varepsilon^{-2} t H^\varepsilon} 
      |\Psi^\varepsilon(\varphi_0)\rangle\langle \Psi^\varepsilon(\varphi_0)|
      e^{i \varepsilon^{-2} t H^\varepsilon}
   \Big),
\end{align}
where $\Psi^\varepsilon(\varphi_0)$ is the initial state from \eqref{eq:initial:states}.

\begin{corollary}\label{cor:particle}
Let $\varphi_0\in \mathfrak h_0 \cap \mathfrak h_{1/2}$ with $e(\varphi_0)<0$, and let $\varphi_t$ be the solution to \eqref{eq:aKG} with initial datum $\varphi_0$.  Let $\psi_{\varphi_t}\ge 0$ denote the normalized ground state of $h_{\varphi_t}=p^2+V_{\varphi_t}$.  Then there exist constants $C,T,\varepsilon_0>0$, depending only on $\varphi_0$, such that
\begin{align}\label{eq:particle-aKG}
\operatorname{Tr}_{L^2}\Big|\,\gamma_{\rm particle}^\varepsilon(t)-|\psi_{\varphi_t}\rangle\langle\psi_{\varphi_t}|\,\Big|
\;\le\; C\,\varepsilon^{\tfrac{1}{10}}\ln \tfrac{1}{\varepsilon} 
\end{align}
for all $|t|\le T$ and $\varepsilon\in(0,\varepsilon_0]$.
\end{corollary}

To analyze the reduced one-body density of the field we start from a dressed initial state. In our proof this step is required, since we need an energy estimate for the initial state. The undressed product states $\Psi^\varepsilon(\varphi_0)$ are not in the form domain of $H^\varepsilon$ and therefore have infinite energy, whereas the dressed state introduced below has finite energy.

We therefore introduce the unitary dressing transformation
\begin{align}\label{eq:Gross:trafo}
U_K^\varepsilon \;:=\; \exp\!\left( \varepsilon \int \d k \Big( \bold B_K(x,k) a_k^*   -  \overline{ \boldsymbol{B}_K(x,k)} a_k \Big) \d k \right)
\end{align}
with 
\begin{align} 
\boldsymbol{B}_K(x,k) :=  \frac{\bold G(x,k)}{k^2}\,\mathbf{1}_{|k|\ge K}.
\end{align}
Accordingly, we define the dressed initial state as
\begin{align}\label{eq:dressed:initial:data}
\Psi_{K}^\varepsilon(\varphi_0) \;:=\; W(\varepsilon^{-1} \varphi_0)\, (U_K^\varepsilon)^* \,(\psi_{\varphi_0} \otimes \Omega).
\end{align}
The difference between the dressed state $\Psi_{K}^\varepsilon(\varphi_0)$ and the undressed state $\Psi^\varepsilon(\varphi_0)$ is negligible in norm as $\ve\to 0$ or $K\to \infty$:
\begin{align}
\| \Psi_{K}^\varepsilon(\varphi_0) - \Psi^\varepsilon(\varphi_0) \| \;\le\; C \varepsilon K^{-1},
\end{align}
so that Theorem~\ref{thm:norm:approximation} continues to hold with the dressed initial state. 
This bound follows from the standard estimate 
\begin{align} \label{eq:dressing:trafo:bound}
\| \big( U_K^\varepsilon - \mathbf{1} \big)\Psi \| \;\le\; C\varepsilon \| B_K\|_{L^2} \| (N+1)^{1/2} \Psi \|,
\end{align} 
valid for every $\Psi \in D(N^{1/2})$, together with $\| B_K \|_{L^2} \le C K^{-1}$. Here $N$ denotes the number operator on $\mathcal F$,  
\begin{align}
N := \int_{\mathbb R^3} a_k^\ast a_k \, \d k .
\end{align}

With this modification, we define the field’s reduced one-body density for the dressed initial state, through its kernel
\begin{align}\label{eq:def-gamma-field}
\gamma_{\rm field}^\varepsilon(t)(k,\ell)
   := \varepsilon^2\,
   \big\langle e^{-i\varepsilon^{-2}t H^\varepsilon}\Psi_{K}^\varepsilon(\varphi_0),\ 
   a_k^\ast a_\ell\, e^{-i\varepsilon^{-2}t H^\varepsilon}\Psi_{K}^\varepsilon(\varphi_0) \big\rangle.
\end{align}
The pre-factor ensures the normalization
\begin{align}
\lim_{\varepsilon\to 0} \operatorname{Tr}_{L^2}\big(\gamma_{\rm field}^\varepsilon(0)\big)
   = \|\varphi_0\|_{L^2}^2.
\end{align}

In principle, the parameter $K$ could be left free in the next corollary, leading to additional error terms that vanish as $K \to \infty$. To balance the errors and simplify the statements, we fix $K=\varepsilon^{-4/5}$.

\begin{corollary}\label{cor:field} 
Let $\varphi_0\in \mathfrak h_0 \cap \mathfrak h_{1/2}$ with $e(\varphi_0)<0$, and let $\varphi_t$ be the solution to \eqref{eq:aKG} with initial datum $\varphi_0$. For $K=\varepsilon^{-4/5}$, choose $\Psi_K^\varepsilon(\varphi_0)$ as in \eqref{eq:dressed:initial:data}. 
Then there exist constants $C,T,\varepsilon_0>0$, depending only on $\varphi_0$, such that 
\begin{align}\label{eq:expectation-N}
 \Big\|\,N^{1/2}\,W(\varepsilon^{-1}\varphi_t)^*\,e^{-i\,\varepsilon^{-2} t H^\varepsilon}\Psi_{K}^\varepsilon(\varphi_0) \Big\|^2 
& \;\le\; C \,  \varepsilon^{-29/15} \ln \tfrac{1}{\varepsilon} \\
\label{eq:field-aKG}
\operatorname{Tr}_{L^2}\Big|\, \gamma_{\rm field}^\varepsilon(t)-|\varphi_t\rangle\langle \varphi_t|\,\Big|   
& \;\le\; C\,\varepsilon^{1/15} \ln \tfrac{1}{\varepsilon}
\end{align}
for all $|t|\le T$ and $\varepsilon\in(0,\varepsilon_0]$.
\end{corollary}

The estimate \eqref{eq:expectation-N} shows that the field component of $e^{-i\varepsilon^{-2} t H^\varepsilon}\Psi_{K}^\varepsilon(\varphi_0)$ remains close to the coherent state $W(\varepsilon^{-1}\varphi_t)\Omega$, in the sense that the number of non-coherent excitations is negligible compared to the $\mathcal O(\varepsilon^{-2})$ coherent excitations as $\varepsilon \to 0$.
The trace-norm bound \eqref{eq:field-aKG} then yields convergence of the field’s one-body reduced density to the rank-one projector $|\varphi_t\rangle\langle \varphi_t|$. Taken together, \eqref{eq:particle-aKG}, \eqref{eq:expectation-N} and \eqref{eq:field-aKG} show that the Nelson dynamics is accurately described at the level of reduced densities by the aKG equation. While the quasi-free fluctuation dynamics $\mathbb U(t)$ is essential for the norm approximation, it does not affect the reduced observables.

\subsection{Comparison with the literature}

The present work lies at the intersection of partially classical limits for quantum field models
and strong-coupling polaron dynamics. We highlight three closely related works and sketch the relevant
similarities and differences.\medskip

\noindent \textit{Strong-coupling Fröhlich polaron dynamics.}
Our analysis is inspired by recent works on the Fröhlich polaron at large coupling \cite{LeopoldRSS2019,Mit20,
LeopoldMRSS2020}; see also \cite{FrankS2014,FrankG2017,Griesemer17} for earlier contributions. Setting $\alpha^{-1} = \varepsilon$ identifies the
strong-coupling limit $\alpha \to \infty$ with a semiclassical limit for the phonon field, and the macroscopic
time scale $t = \mathcal O( \alpha^2)$ corresponds to our time scale   $t = \mathcal O(  \varepsilon^{-2})$. The main difference from the Nelson model is that the Fröhlich Hamiltonian is UV-regular. In \cite{LeopoldMRSS2020}, a \emph{norm approximation} of the dynamics was
established up to times $t = \mathcal O( \alpha^2 ) $ by the Landau--Pekar (LP) system, corrected by a
Bogoliubov (quasi-free) fluctuation dynamics around the coherent state. The $\alpha$-dependent LP system is the analogue of \eqref{eq:SKG} for the polaron.

In the renormalized Nelson model, the problem is more delicate. Most importantly,
the fluctuation generator is UV-divergent and requires a renormalization procedure, a
complication that is absent in the Fröhlich case. In addition, we introduce and analyze the
aKG equation, which provides an $\varepsilon$-free effective dynamics, whereas both LP and SKG
depend explicitly on the semiclassical parameter. These two ingredients, renormalization of
the fluctuation generator and the use of aKG, constitute the conceptual novelties of this work.
At the same time, our analysis builds on and extends the adiabatic expansion scheme
developed for the polaron, by adapting it to the more singular setting of the Nelson model.\medskip

\noindent \textit{Partially classical limit for the renormalized Nelson model.}
Ginibre, Nironi and Velo \cite{Ginibre06} gave the first rigorous analysis of the renormalized Nelson
model in a partially classical regime. They worked on microscopic time scales $t= \mathcal O( 1 ) $ and used a
coherent-state/weak-coupling scaling, where the number of field excitations is of order
$\varepsilon^{-2}$ and the coupling strength is proportional to $\varepsilon$. In this regime the
quantum field converges to a solution of the free wave (Klein--Gordon) equation, and the
particles evolve by a Schrödinger equation with an external potential generated by the free
field. Since the back-reaction on the field is suppressed, the effective dynamics remains
linear. The renormalized Hamiltonian is constructed via a unitary dressing transformation,
and convergence of the propagators is obtained on a suitable subspace, with a quantitative error estimate.

In our scaling the situation is different: the Poisson/commutator structure places a factor
$\varepsilon^{-2}$ in front of $\partial_t \varphi$, so at microscopic times $t= \mathcal O( 1) $ the field remains
frozen rather than evolving freely. Nontrivial field evolution appears only on the slow scale
$t = \mathcal O(  \varepsilon^{-2} ) $. On that scale we establish a norm approximation in terms of the
{nonlinear} aKG equation, corrected by a renormalized fluctuation dynamics. The need to renormalize the fluctuation generator has no counterpart in \cite{Ginibre06}, since the field there remains free and unaffected by back-reaction and fluctuations..\medskip

\noindent \textit{Quasi-classical/Wigner-measure approach.}
In \cite{CFO22} the authors developed an abstract quasi-classical framework for particle-field systems
and applied it to several models, including the UV-{regularized} Nelson model. The quasi-classical parameter is introduced by scaling creation and annihilation operators with a factor
$\varepsilon,$\footnote{We write $\varepsilon$ in place of the $\varepsilon^{1/2}$ used in \cite{CFO22} to stay consistent with our notation.}
together with an $\varepsilon$-dependent weight $\nu(\varepsilon)$ in front of the field energy (our
setting corresponds to $\nu(\varepsilon)=1$). At microscopic times $t = \mathcal O( 1 ) $ they distinguish two
regimes: if $\varepsilon \nu(\varepsilon) \to 0$ the classical field remains frozen, while if
$\varepsilon \nu(\varepsilon) \to 1$ it evolves freely. Since there is no back-reaction on the field on this
time scale, the effective particle dynamics is linear, driven by an external classical field.
Convergence to the effective evolution is established via state-valued Wigner
measures. Their assumptions on initial states are very
general and allow for entangled particle-field configurations.

Our analysis addresses the renormalized Nelson model for more specific, semiclassical initial data and focuses on the macroscopic time scale. In this regime the field exhibits nontrivial back-reaction and the effective dynamics is nonlinear.\medskip

\noindent \textit{Other works.}
Within the broader context of deriving effective equations from particle-field models, let us
also mention the following works. In \cite{Falconi2,Falconi1,Ammari,
LeopoldPetrat,
Leopold3,Leopold2,Leopold1,Falconi3,
Leopold5}, different particle-field models were
studied in a many-body mean-field regime, where $N$ bosons couple weakly to the quantum
field with coupling strength of order $N^{-1/2}$. In particular, \cite{Ammari,Falconi3} treat the renormalized
Nelson model and derive the SKG system as an effective description. In this
setting there is no separation of time scales between the bosons and the field, and hence $\varepsilon=1$. Another difference to our
semiclassical setting is that the quantum fluctuations in the mean-field regime act simultaneously on the excitations of the
bosons and the field, which requires a different renormalization strategy \cite{Falconi3}. Finally, in the
mean-field limit, the norm approximation of the dynamics is obtained as a consequence of the
convergence of reduced one-body densities, whereas in our setting we first prove the norm
approximation and deduce the reduced densities as a corollary.

The works \cite{Davies,Hiroshima,Teufel,Tenuta,
GubinelliHJ2014,Cardenas} study the derivation of effective (direct) particle interactions
arising from the coupling to Nelson type quantum fields. These results are obtained in suitable
weak-coupling and adiabatic regimes, where the roles of time scales are reversed: the particles
evolve slowly while the field provides the fast subsystem.
\subsection{Outline of the paper}

We start in Section~\ref{sex:aux:estimates} with a collection of auxiliary estimates for creation and annihilation operators and for the classical potentials. These bounds are used throughout Sections~\ref{sec:dressed:Nelson}--\ref{sec:proof:reduced:densities}.

Section~\ref{sec:dressed:Nelson} introduces the dressed Nelson Hamiltonian, obtained by applying the unitary dressing transformation $U_K^\varepsilon$, see \eqref{eq:Gross:trafo}, to the renormalized Nelson Hamiltonian. We then perform a Weyl shift around the classical field, which separates the leading particle Hamiltonian from subleading field terms and expands the Hamiltonian in powers of $\varepsilon$. 
The purpose of this section is to obtain a suitable representation of the Nelson Hamiltonian that is needed in the proof of the norm approximation.

In Section~\ref{sec:properties:adiabatic_SKG} we analyze the \eqref{eq:aKG} equation: we prove local well-posedness (by a contraction argument), energy conservation and persistence of the spectral gap for suitable time intervals. Beyond well-posedness, this provides control on resolvents and regularity estimates for the aKG pair $(\psi_{\varphi_t},\varphi_t)$, which are repeatedly used later on.

In Section~\ref{sec:bog:ren} we construct the renormalized Bogoliubov--Nelson evolution 
$\mathbb U(t)$ as the strong limit of the regularized propagators generated by 
$\mathbb H_{\Lambda}(t)$ (defined in \eqref{eq:Bog:cutoff}), 
up to a logarithmically diverging energy counter term. This renormalization is based on the key identity
\begin{align}
\mathbb H_\Lambda(t) = \mathbb H_{\mathscr D,K,\Lambda}(t) - 4\pi\ln\Lambda
\end{align}
where the operator $\mathbb H_{\mathscr D,K,\Lambda}(t)$ is related to the dressed Nelson Hamiltonian and has a limit as $\Lambda \to \infty$. In contrast to the full quantum model, this renormalization does not rely on the application of a unitary dressing transformation. We also establish bounds for powers of the number operator propagated under $\mathbb U(t)$.

Section~\ref{sec:proof:norm:approximation} is devoted to the proof of 
Theorem~\ref{thm:norm:approximation}.
After applying the unitary dressing transformation and a coherent shift aligned with the aKG solution, 
we compare the renormalized Nelson dynamics with the semiclassical evolution associated with 
$(\psi_{\varphi_t},\varphi_t)$, corrected by the fluctuation dynamics $\mathbb U(t)$. 
The comparison relies on a first-order Duhamel expansion combined with a second- and third-order adiabatic expansion around the 
instantaneous gapped particle ground state. 
A crucial step in this expansion is a cancellation coming from the generator $\mathbb U(t)$, which removes a contribution of order one 
on the macroscopic time scale $t=\mathcal O(\varepsilon^{-2})$. 
After this cancellation, all remainder terms are of higher order in $\varepsilon$ and can be controlled using 
the estimates established in the preceding sections.

In Section~\ref{sec:proof:reduced:densities} we deduce convergence of the reduced densities. While the convergence of the reduced particle density follows immediately from the norm approximation, the analysis of the reduced one-body density of the field requires additional work. Here we combine the norm approximation with an improved energy estimate and an a~priori bound on the propagated number operator.

Finally, Appendix~\ref{app:renormalization} revisits the renormalization of the Nelson model 
in the semiclassical regime and establishes the dressing identity for the Nelson Hamiltonian. 
Appendix~\ref{app:persistence} extends the local well-posedness of aKG to global well-posedness 
for initial data close to $\mathcal M$, using energy conservation together with a coercive lower bound 
for the field energy. 
Appendix~\ref{app:conv:SKG:aKG} proves the convergence of SKG to aKG  as $\ve\to 0$.

\section{Auxiliary estimates}
\label{sex:aux:estimates}

\subsection{Definitions and notation}

Before turning to the technical estimates, we summarize notation that will be used repeatedly throughout this and the following sections.

First, we introduce $x$-dependent creation and annihilation operators acting on $L^2(\mathbb R^3_x)\otimes \mathcal F$ by
\begin{align}\label{eq:smeared:operators}
(a^\ast(\mathbf{f})\Psi)(x):=\left( \int_{\mathbb{R}^3} \mathbf{f}(x,k)\,a_k^\ast \, dk \right)\Psi(x), \quad
(a(\mathbf{f})\Psi)(x):=\left( \int_{\mathbb{R}^3} \overline{\mathbf{f}(x,k)}\,a_k \, dk \right)\Psi(x),
\end{align}
for kernels $\mathbf{f}:\mathbb{R}^3_x\times\mathbb{R}^3_k\to\mathbb{C}$ that are measurable in $x$ and satisfy $\omega^{-1/2}\mathbf{f}(x,\cdot)\in L^2(\mathbb{R}^3_k)$ for almost every $x\in\mathbb{R}^3$. Under this condition, the operators are well defined as quadratic forms on $Q(H_f)$. 
When $\mathbf{f}(x,k)=e^{-ik x}f(k)$ for suitable $f:\mathbb R^3_k\to\mathbb C$, we keep the boldface notation for the $x$-dependent kernel and write $f$ for the $k$-only kernel. We further set
\begin{align}
\phi (\mathbf{f}) := a^*(\mathbf{f}) + a(\mathbf{f}).
\end{align}

Next, we recall the kernels already defined in the previous sections,
\begin{alignat}{3}\label{eq:G:reminder}
\boldsymbol{G}_K(x,k) &= e^{-ikx}\, G_K(k) 
&\qquad &\text{with} \qquad 
&G_K(k) &= \omega(k)^{-1/2}\, \mathbf 1_{|k|\le K}, \\[1mm]
\boldsymbol{B}_K(x,k) &= e^{-ikx}\, B_K(k) 
&\qquad &\text{with} \qquad 
&B_K(k) &= \frac{G(k)}{|k|^2}\, \mathbf 1_{|k|\ge K}.\label{eq:B:reminder} 
\end{alignat}
With an additional ultraviolet cutoff $\Lambda\ge K$, we write
\begin{align}
\boldsymbol{B}_{K,\Lambda}(x,k) :=  e^{-ikx} B_{K,\Lambda}(k),
\qquad  B_{K,\Lambda}(k) := B_K(k)\, \mathbf 1_{|k|\le \Lambda}. \label{eq:BKLambda}
\end{align}

For $K\ge0$ we define the modified number and field energy operators
\begin{align}
N_K  :=  \int \mathbf 1_{|k|\ge K}\, a_k^* a_k \, \d k,
\qquad  
H_{f,K}  := \int \mathbf 1_{|k|\ge K}\, \omega(k)\, a_k^* a_k \, \d k,
\end{align}
so that $N=N_0$ and $H_f=H_{f,0}$.

For any densely defined operator $A$ we write 
\begin{align}
\langle A\rangle := (1+A^*A)^{1/2},
\end{align}
and we use 
\begin{align}
X \lesssim Y
\end{align}
to denote an estimate of the form $X \le C\,Y$ for some constant $C>0$ independent of the relevant parameters (e.g. $\varepsilon$, $K$, and $\Lambda$). The value of $C$ may change from line to line.

For ease of reference we recall the momentum operator  and the dispersion relation
\begin{align}
p:=-i\nabla_x, \qquad \omega(k):=|k|,
\end{align}
as well as the energy spaces for the particle and the classical field:
\begin{alignat}{2}
H^r  &:= \big\{ \psi \in L^2(\mathbb{R}^3_x) : \|(1+p^2)^{r/2} \psi \|_{L^2} < \infty \big\}, \quad & r\in\mathbb R, \\[1mm]
\mathfrak{h}_s  &:= \big\{ \varphi \in L^2(\mathbb{R}^3_k) : \|\omega^s \varphi\|_{L^2} < \infty \big\}, \quad & s\in\mathbb R.
\end{alignat}

Finally, for the Fourier transform we use the convention
\begin{align}
\widehat f(k) := \int_{\mathbb R^3} e^{-ikx} f(x)\,\d x,
\qquad 
f(x) = (2\pi)^{-3}\int_{\mathbb R^3} e^{ikx}\, \widehat f(k)\,\d k.
\end{align}

\subsection{Bounds for creation and annihilation operators}
\label{subsec:general-aa}

We begin with standard operator bounds for the creation and annihilation operators \eqref{eq:smeared:operators}. These are simple consequences of the CCR; see e.g. \cite[Lemma B.1]{Griesemer18} and \cite[Lemma 3.3]{Cardenas}. 

\begin{lemma} \label{lem:operator:bounds:a} 
Let $ \bold f(x,k) = e^{-ikx}  f(k)$ with $f:\mathbb R^3 \to \mathbb R$ measurable and $\mathrm{supp}(f) \subseteq \{ |k| \ge K \} $ for $K\ge 0$. Then
		\begin{alignat}{2}
			\label{af1}
		\|   a(\bold f ) \Psi \| 		 
		& \leq	
		\|   f\|_{L^2} 	 \| N_K ^{ \frac12}  \Psi \| , \qquad 
&		\|   a(\bold f ) \Psi \| 		 &
		\leq	
		\|  \omega^{ - \frac{1}{2}} f\|_{L^2}  	 \|  H_{f,K}^{ \frac12 }\Psi \| , \\[0mm]
				\|   a^*(\bold f) \Psi \| 		 
		& \leq	
		\|    f	 \|_{L^2}  \| \< N_K ^{ \frac12 }\>  \Psi \| , \qquad 
&		\|   a^*(\bold f ) \Psi \| 		 &
		\leq	
		\max\{ \|   f \|_{L^2}  , \|  \omega^{ - \frac{1}{2}} f\|_{L^2}  	\}  \|\<  H_{f,K}^{ \frac12 }\>  \Psi \|
	\end{alignat}
	for all $\Psi \in L^2(\mathbb R^3_x) \otimes \mathcal F$. Additionally, for $\bold g(x,k) = e^{-ikx}   g(k)$ with $g:\mathbb R^3 \to \mathbb R$ measurable and $\mathrm{supp}(g) \subseteq\{ |k| \ge \Lambda \}  $ for $\Lambda \ge K$, we have
\begin{align}
			\label{af3}
	 		\| a(\bold f)  a(\bold g) \Psi \| &   \leq	
	 \|  	 \omega^{ - \frac{1}{4}}   f\|_{L^2} 
	 \|  \omega^{ - \frac{1}{4}}   g\|_{L^2}  \,  \| \< N_K^{\frac12} \> H_{f,K} ^{\frac12} \Psi \| .
\end{align}
\end{lemma}

The next lemma is a generalized version of the Frank--Schlein estimate \cite[Lemma 7]{FrankS2014}.

\begin{lemma} \label{lem:Frank-Schlein} Let $\bold f (x,k) = e^{-ikx} f(k)$ with $f:\mathbb R^3 \to \mathbb R $ measurable. Then
\begin{align}
\| a( \bold f ) \Psi  \|^2 \le  \mathscr  C_{s,r}(f) \, \| \< |p|^r \> \textnormal{d}\Gamma(\omega^s)^{\frac12} \Psi \|^2 \quad \text{with}\quad  \mathscr  C_{s,r}(f)  =  \sup_{h \in \mathbb R^3} \int \d k  \frac{|f(k)|^2}{\omega(k)^s(1+ |h-k|^{2r}) }  \notag
\end{align}
for all $\Psi \in L^2(\mathbb R^3_x) \otimes \mathcal F$ and $s,r\ge 0$, where $\textnormal{d}\Gamma(\omega^s) = \int \omega(k)^s a_k^* a_k dk $. 

Moreover, if $ k \mapsto f(k)$ is radial and monotone decreasing, then 
\begin{align}\label{eq:symm:rearrangement}
\mathscr  C_{s,r}(f) \le \int \d k  \frac{|f(k)|^2}{ | k| ^s( 1 + |k|^{2r} ) } \ .
\end{align}
\end{lemma} 
\begin{proof}
Set $f_s = \omega^{-\frac{s}{2}}f$ and estimate via Cauchy--Schwarz:
\begin{align}\label{eq:proof:FrankSchlein:bound}
 \| a( \bold f ) \Psi \|^2  & = \int \d k \d \ell \Big\<    e^{i k x } \overline{f_s(k)} \omega^{\frac{s}{2}}(k)  a_k    \Psi  ,  \frac{  \< |p - k|^r \>}{ \< |p - \ell|^r \> }  \frac{  \< |p - \ell |^r \>}{ \< |p - k |^r \>  }  e^{i\ell x} \overline{f_s( \ell)}\omega^{\frac{s}{2}}(\ell)  a_\ell  \Psi \Big\>\notag  \\
& \le  \int \d k \d \ell  \Big\|   \frac{1}{\< |p - k|^r \>}   f_s(k ) \< | p - \ell |^r \> e^{i\ell x} \omega^{\frac{s}{2}}(\ell) a_\ell   \Psi  \Big\|^2  \notag\\
& = \int \d \ell \, \omega^s(\ell) \Big\<  \< |p - \ell|^r \> e^{i\ell x}  a_\ell \Psi , \int dk \frac{|f_s(k)|^2}{ \< |p - k|^r \>^2}  \< |p - \ell|^r \> e^{i\ell x}   a_\ell \Psi\Big\rangle \notag\\
& \le \bigg\| \int \d k  \frac{|f_s(k)|^2}{\< |p - k|^r \>^2} \bigg\| \,  \| \< |p|^r \> \textnormal{d}\Gamma (\omega^s )^\frac12 \Psi \|^2  \ ,
\end{align}
where we used the identity $\< | p-\ell |^r\> e^{i\ell x } = e^{i \ell x } \< |p|^r \>$ in the last step. It remains to estimate the operator norm, that is, for normalized $\psi \in L^2(\mathbb R^3)$
\begin{align}
\bigg\| \int  \d k \frac{ |f_s(k)|^2}{\< |p- k|^r \>^2 } \psi  \bigg\| & = \bigg[  \int \d p \bigg| \int \d k  \frac{ |f_s(k)|^2}{\< |p- k|^r \>^2 }   \widehat \psi(p) \bigg|^2 \bigg]^{\frac12} \le   \sup_{p\in \mathbb R^3} \int \d k \frac{ |f_s(k)|^2}{ 1+ | p - k| ^{2r}  } .
\end{align}
This completes the proof of the first statement. 

The second statement is a consequence of the symmetric rearrangement inequality.
\end{proof}

We now specialize these bounds to the kernels \eqref{eq:G:reminder} and \eqref{eq:B:reminder}. For reference, let us note the following bounds:
\begin{subequations}
\begin{align}\label{eq:bounds:G:1}
\| \omega^{- 	 s		} G_K \|_{L^2}  & \lesssim  K^{1-s}, \\[1mm]
\|  \mathbf 1_{|k| \ge 2 } |k|^{-1} G_K \|_{L^2} & \lesssim \sqrt {\ln K} , \label{eq:bounds:G:2} \\[1mm]
  \| \omega^{s-1}   k B_K \|_{ L^2} & \lesssim  K^{s-1} \ ,\label{eq:bounds:B}
\end{align}
\end{subequations}
for all $K \ge 2$ and $s\in [0,1)$. 

The next lemma collects useful estimates on the constant $\mathscr  C_{r,s}(f)$.
\begin{lemma}\label{lem:FS:constant:bounds}
Let $\mathscr  C_{r,s}(\cdot)$ be defined as in Lemma~\ref{lem:Frank-Schlein}. For all $K\ge2$,  $s\in\{0,1\}$ and $r\in(0,1]$,
\begin{alignat}{2}
\mathscr  C_{0,1}(G_K) &\lesssim \ln K   ,
&\qquad\qquad\qquad
\mathscr  C_{0,r}(kB_K) &\lesssim K^{-3r/2}  , \\
\mathscr  C_{1,1}(G_K) &\lesssim 1 ,
&
\mathscr  C_{1,1}(\omega k B_K) &\lesssim K^{-3/4} , \\
\mathscr  C_{s,1}(\omega G_K) &\lesssim K^{2-s} ,
&
\mathscr  C_{0,2}(\omega k B_K) &\lesssim K^{-1} .
\end{alignat}
\end{lemma}
\begin{proof} The bounds involving $G_K$ are readily verified via \eqref{eq:symm:rearrangement}. For the bounds involving $B_K$, where $K$ enters as a lower cut-off, we use  
\begin{align}
\mathscr  C_{0,r}(kB_K) & = \sup_{h\in \mathbb R^3} \int \frac{\mathbf 1(|k|\ge K) }{|k|^3(1+ |k-h|^{2r})} \d k \notag\\
&\le K^{- \frac32 r} \sup_{h\in \mathbb R^3}   \int \frac{\mathbf 1(|k|\ge K) }{|k|^{3-\frac32 r}(1+ |k-h|^{2r})}\d k  \lesssim K^{-\frac32 r },
\end{align}
where the last step follows by removing the restriction $|k|\ge K$ and then applying \eqref{eq:symm:rearrangement}. The other estimates follow analogously. 
\end{proof}

We combine Lemmas \ref{lem:operator:bounds:a} -- \ref{lem:FS:constant:bounds} to obtain the following estimates.

\begin{lemma}\label{lem:operator:bounds:b} 
For all $K\ge 2$ and $\Psi \in L^2(\mathbb R^3_x) \otimes \mathcal F$, we have
\begin{subequations}
	\begin{align}
		\label{op1}
	\| a ( \bold G_K) \Psi	\| 
		& \,  \lesssim	 \, 
		 \min \big\{ \, ( \ln K )^{\frac12}  \|	\< p\> N^{\frac12} \Psi \| \, , \,  \|\< p \>	H_f^\frac{1}{2}   \Psi	\| \, , \,  K^\frac12 \| H_f^\frac12 \Psi \| \big\} ,\\
\label{op2}
\|  a (\omega \bold G_K) \Psi	\| 
		& \,   \lesssim 		 \, 
  \min \big\{ \,  K \,  
		\|	\< p\> N^\frac12 \Psi	\| \,  , \,  	 K^{\frac12}
		\|	\< p\> H_f^\frac12 \Psi	\| \,  \big\}  , \\
		\label{op3}
		\|  a(k \bold B_K) \Psi \| 				& 		\,		 \lesssim		\,		
	  \min\big\{  \, K^{-\frac34}  \|	\< p \>  N^{\frac12} \Psi 	\| \,  , \,
	K^{-\frac12}  \|	H_f^{\frac12} \Psi	\| \,  \big\}, \\
		\label{op4}
	\|  a(\omega k \bold B_K) \Psi \|  		&	\,	\lesssim				\,		
		 \min \big\{ K^{-\frac38} \|	\< p \>  H_f^{\frac12} \Psi	\| \, ,\, K^{-\frac12} \| \< p^2 \> N^\frac12 \Psi \| \big\}  .
\end{align}
Additionally, for all $\Lambda \ge K  \ge 2$:
\begin{align} 		
		 |	 \langle 	\Psi,  a(k \bold B_K) a (k \bold B_\Lambda  ) \Phi \rangle | &	\, \lesssim  \, 	  (K \Lambda)^{-\frac14}   \, \|  \< H_f^{\frac{1}{2}}\> \Psi 	\|
		\| \<  H_f^{\frac{1}{2}}\> \Phi 	\|  \label{op5},\\[1mm]
		\|  a(k \bold B_K) a (k \bold B_\Lambda  ) \Psi \|  &	\, \lesssim  \,  (K \Lambda )^{-\frac38}    \, \| \< p \> N \Psi \| \,  \ . \label{op5:b}
	\end{align}
\end{subequations}
\end{lemma}

\begin{proof} 
The bounds \eqref{op1} and \eqref{op2} follow from Lemmas \ref{lem:Frank-Schlein} and \ref{lem:FS:constant:bounds}. The second bound in \eqref{op3} follows directly from Lemma \ref{lem:operator:bounds:a} and \eqref{eq:bounds:B}. The first bound in \eqref{op3} and both bounds in \eqref{op4} follow from Lemma \ref{lem:Frank-Schlein} together with \eqref{eq:bounds:B}.

For the  bound in \eqref{op5}, we apply Cauchy–Schwarz, Lemma \ref{lem:operator:bounds:a}, and \eqref{eq:bounds:B}:
\begin{align}
	 	|	\< \Psi, a( k \bold B_K )  a( k \bold B_\Lambda ) \Phi \> | & =  |  \langle  \Psi, a( k \bold B_K )  a( k \bold B _\Lambda ) \< N_K ^{\frac12}\> \< N_K^{\frac12}\>^{-1}  \Phi \> |
	 	 \notag\\
	 	& =  | \< (N_K + 3)^{\frac12} \Psi, a( k \bold B_K )  a( k \bold B _\Lambda ) \< N_K^{\frac12}\>^{-1}  \Phi \> |
	 	 \notag\\
	 	& \lesssim \| \< N_K^{\frac12} \> \Psi \|\,  \| a( k \bold B_K )  a( k\bold B_\Lambda ) \< N_K^{\frac12}\>^{-1}  \Phi \| \notag\\
	 	& \lesssim  \|  	 \omega^{ - \frac{1}{4}}  k B_K \|_{L^2} \, 
	 \|  \omega^{ - \frac{1}{4}}  k B_\Lambda \|_{L^2} \,  \| \< N_K^{\frac12} \> \Psi \|\, \| H_{f,K}^\frac12 \Phi \| 
	 	\notag\\
	 &\lesssim  (K\Lambda)^{-\frac14} \, \| \< H_f^{\frac12} \> \Phi \|\, \| H_f^\frac12 \Psi \| ,
\end{align}
where we used $N_K \le  H_{f,K}\le   H_f $ for $K\ge 2$.  
For the bound in \eqref{op5:b}, we apply Lemma \ref{lem:Frank-Schlein} twice, that is, we estimate
\begin{align}
\| a(k \bold B_K) a( k \bold B_\Lambda) \Psi  \|^2 & \le \mathscr  C_{0,\frac12}(k B_K) \, \|  \< |p|^\frac12 \> a (k \bold B_\Lambda) N^\frac12 \Psi \|^2 
\end{align}
and
\begin{align}
\| \< |p|^\frac12 \> a (k \bold B_\Lambda ) N^\frac12 \Psi  \|^2 & \le 2\,  \| a ( \< |k|^\frac12 \> kB_\Lambda ) N^\frac12 \Psi  \|^2  + \,  2 \, \|  a (k\bold B_\Lambda ) \< |p|^\frac12 \>  N^\frac12 \Psi  \|^2   \notag\\[1mm]
& \le 2\, \mathscr  C_{0,1} (|k|^\frac32  B_\Lambda  ) \, \| \< p\> N \Psi  \|^2 + 2\, \mathscr  C_{0,\frac12} (kB_\Lambda )  \, \| \<|p|^\frac12\>^2 N \Psi \|^2 .
\end{align}
Using $\langle |p|^{1/2}\rangle^2 \lesssim \langle p\rangle$ and Lemma \ref{lem:Frank-Schlein} concludes the proof of \eqref{op5}; note that $\mathscr C_{0,1}(|k|^{3/2} B_\Lambda ) = \mathscr  C_{1,1}(\omega k B_\Lambda )$, since $\omega(k) = |k|$.
\end{proof}

\subsection{Bounds for $A_K$ and $D_K$}

We now apply the previous estimates to control operators that will appear repeatedly in the upcoming sections. To this end, let us define
\begin{align}\label{eq:def:A}
A_{K} & :=  \phi( \bold G_K )  + 2 a^*( k \bold B_{K } ) \cdot p + 2  p \cdot a ( k \bold B_{K} ), \\[1.5mm]
D_{K} & :=   2  a^*(k \bold B_{K } ) a( k  \bold B_{K} ) + a^*( k \bold B_{K} )^2 + a( k \bold B_{K}) ^2 + 4\pi \ln K  \label{eq:def:D}
\end{align} 

The estimates summarized next play a central role in the construction of the renormalized fluctuation dynamics in Section~\ref{sec:bog:ren}. We provide bounds for $A_K$ and $D_K$ both in terms of the number operator $N$ and the field energy $H_f$.

\begin{lemma} \label{lem:T:bounds} 
Let $R\ge 0 $ be a bounded non-negative operator acting on $L^2(\mathbb R^3_x)$. For all $K\ge 2$ and $\Psi \in L^2(\mathbb R^3_x) \otimes \mathcal F$ we have
\begin{subequations}
\begin{align}
|\< \Psi, A_{K} R A_{K} \Psi \>| & \lesssim   K^2 \big(\| \<  p  \> R \<  p \> \|+ \| R \| \big) \| \< H_f^\frac12 \> \< p \> \Psi \|^2 , \label{eq:ARA:T:bound}\\[0mm]
|\< \Psi, [N,A_{K}] R A_{K} \Psi \>| & \lesssim K^2 \big(\| \<  p  \> R \<  p \> \|+ \| R \| \big) \| \< H_f^\frac12 \> \< p \> \Psi \|^2 ,\label{eq:comm:ARA:T:bound}\\[0mm]
|\< \Psi, D_{K}  \Psi \>| &  \lesssim K^{-\frac12} \| \< H_f^\frac12 \>  \Psi \|^2   +  \ln K \| \Psi \|^2 , \label{eq:D:T:bound} \\[0mm]
|\< \Psi, [N,D_{K}]  \Psi \>| & \lesssim K^{-\frac12} \| \< H_f^\frac12 \>  \Psi \|^2 . \label{eq:comm:D:T:bound}
\end{align}
\end{subequations}
Moreover, in terms of the number operator $N$ we have
\begin{subequations}
\begin{align}
|\< \Psi, A_K  R A_{K} \Psi \>| & \lesssim  \ln K  \big(\| \<  p  \> R \<  p \> \|+ \| R \| \big) \| \< N^\frac12 \> \< p \> \Psi \|^2 , \label{eq:ARA:N:bound} \\[0.5mm]
|\< \Psi, [N , A_K  ] R A_{K} \Psi \>| & \lesssim \ln K  \big(\| \<  p  \> R \<  p \> \|+ \| R \| \big) \| \< N^\frac12 \> \< p \> \Psi \|^2 , \label{eq:comm:ARA:N:bound} \\[0.5mm]
|\< \Psi, D_K  \Psi \>| & \lesssim  \ln K  \|  \< N^\frac12 \> \< p \> \Psi \|^2 , \label{eq:D:N:bound} \\[0.5mm]
|\< \Psi, [N , D_K  ]   \Psi \>| & \lesssim K^{-\frac12} \|  \< N^\frac12 \> \< p \> \Psi \|^2  . \label{eq:comm:D:N:bound} 
\end{align}
\end{subequations}
\end{lemma}

\begin{proof} 
Since $R\ge 0$, it suffices to bound the diagonal terms in $A_K R A_K$. With the aid of Lemma~\ref{lem:operator:bounds:a}   and $\| G_K\|_{L^2} \lesssim K$, we find
\begin{align}
|\< \Psi, \phi (\bold G_K) R  \phi (\bold G_K)  \Psi \>| &  \lesssim  K^2 \| R \| \,  \| \< H_f^\frac12 \> \Psi \|^2 , \label{eq:bound:phi:R:phi},
\end{align}
and invoking Lemma \ref{lem:operator:bounds:b}, we also get
\begin{align}
|\< \Psi, (a^*(k\bold B_K) \cdot p) R (p\cdot a (k \bold B_K))  \Psi \>| & \lesssim  K^{-1} \| \<  p  \> R \<  p \> \|\, \| \< H_f^\frac12 \> \Psi \|^2 .\label{eq:bound:a*:R:a}
\end{align}
For the third diagonal term in $A_K R A_K$, we apply Cauchy--Schwarz and then Lemma~\ref{lem:Frank-Schlein} (with $s=0$, $r=1$) and Lemma \ref{lem:FS:constant:bounds}:
\begin{align}
 \< \Psi, (p\cdot a(k \bold B_K)) R (a^* (k \bold B_K) \cdot p ) \Psi \>  
 & = | \< (N_K+2)^{\frac12} \Psi, (p \cdot a(k \bold B_K)) \< N_K^\frac12 \>^{-1} R (a^* (k \bold B_K) \cdot p) \Psi \> | \notag\\
& \lesssim \|p  \<N_K^\frac12 \> \Psi \| \,  \|  a(k \bold B_K) \< N_K^\frac12 \>^{-1} R (a^* (k \bold B_K) \cdot p) \Psi \|  \notag\\
 & \lesssim \mathscr C_{0,1}(kB_K)^\frac12 \, \|p  \<N_K^\frac12 \> \Psi \| \,  \|  R a^* (k\bold B_K)  \cdot p \Psi \| \notag\\
 & \lesssim K^{-\frac12} \, \| R^{\frac12} \| \, \| \< p\> \< N_K^\frac12\> \Psi \| \, \|  R^{\frac12} a^* (k\bold B_K)\cdot  p \Psi \| .
\end{align}
Dividing by the last factor, and using $N_K\le H_f$ for $K\ge 1$ gives
\begin{align}\label{eq:bound:a:R:a*}
 \< \Psi, (p\cdot a(k \bold B_K) ) R ( a^* (k \bold B_K)  \cdot p ) \Psi \> &  \lesssim  K^{-1} \|  R \| \, \| \< p\> \< H_f^\frac12\> \Psi \|^2.
\end{align}
Combining \eqref{eq:bound:phi:R:phi}–\eqref{eq:bound:a:R:a*} proves \eqref{eq:ARA:T:bound}; the commutator bound \eqref{eq:comm:ARA:T:bound} follows in the same way.  
The two bounds for $D_{K}$ are immediate consequences of Lemma~\ref{lem:operator:bounds:b}.  

To show \eqref{eq:ARA:N:bound}, we estimate again the diagonal terms in $A_K R A_K$. By Lemmas \ref{lem:Frank-Schlein} and \ref{lem:FS:constant:bounds},
\begin{align}
|\< \Psi, a^* (\bold G_K) R  a (\bold G_K)  \Psi \>| &  \lesssim  \mathscr C_{0,1}(G_K) \,  \| R \| \,  \| \< p \> \< N^\frac12 \> \Psi \|^2  \lesssim \ln K \, \| R \|\,  \| \< p \> \< N^\frac12 \> \Psi \|^2 
\end{align}
and similarly for
\begin{align}
|\< \Psi, a(\bold G_K) R  a^* (\bold G_K)  \Psi \>| & = |\< (N+2)^\frac12 \Psi, a(\bold G_K) \< N^\frac12 \>^{-1} R  \< N^\frac12\>^{-1} a^* (\bold G_K) (N+2)^\frac12 \Psi \>| \notag\\
& \lesssim \| a(\bold G_K) \, \< N^\frac12 \>^{-1} R^\frac12 \|^2 \, \| \< N^\frac12 \>^{-1} \Psi \|^2 \notag\\
& \lesssim  \ln K \,  \| \< p \> R^{\frac12} \|^2 \| \< N^\frac12 \> \Psi \|^2.
\end{align}
The remaining two diagonal terms in the estimate for $A_K R A_K$ are obtained in close analogy, as is the commutator bound \eqref{eq:comm:ARA:N:bound}. 

Finally, the two bounds for $D_K$ follow again directly from Lemma \ref{lem:operator:bounds:b}. 
\end{proof}

In the next lemma, we recast the previous bounds in a form more suitable for the proof of our norm-approximation theorem in Section~\ref{sec:proof:norm:approximation}. 

\begin{lemma} \label{lem:N:bounds:2} 
Let $m \in \{0,1,2\}$. For all $K\ge 2$, we have
\begin{subequations}
\begin{align}
\|\< N^\frac12\>^{m}   \<p\>^{-1} (A_{K} - \phi(\bold    G_K))  \<p\>^{-1}  \< N^\frac12\>^{-m-1} \| & \lesssim K^{-\frac34} , \label{eq:bound:A:N:2} \\
\| \< N^\frac12\>^{m}  \<p\>^{-1} A_{K}  \< p \>^{-1}  \< N^\frac12\>^{-m-1}\| & \lesssim  ( \ln K)^\frac12 , \label{eq:bound:A:N:1} \\
\|  \< N^\frac12\>^{m}   \<p\>^{-2}  [H_f,   A_K  ]\<p\>^{-2}   \< N^\frac12 \>^{-m-1} \| & \lesssim  K , \label{eq:bound:A:N:3}\\
\|  \< N^\frac12 \>^{m}    \< p \>^{-1} (D_{K} - 4\pi \ln K) \< p \>^{-1} \< N^\frac12 \>^{-m-2} \| & \lesssim K^{-\frac34}, \label{eq:bound:D:N:2}\\
\| \< N^\frac12 \>^{m}   \< p \>^{-1}   D_{K}\< p \>^{-1} \< N^\frac12 \>^{-m-2}  \| & \lesssim   \ln K , \label{eq:bound:D:N:1}
\end{align}
where $\| \cdot \|$ denotes the operator norm on $L^2(\mathbb R^3_x) \otimes \mathcal F$.
\end{subequations}
\end{lemma}

\begin{proof}
For the first bound we note that $A_K - \phi(\bold G_K) = 2 p\cdot a(k \bold B_K) + 2 a^*(k \bold B_K)  \cdot p$. Commuting through the operator $\langle N^{\frac12}\rangle^{m}$ with $a$ and $a^*$ (which shifts $N$ by $\pm 1$) and applying \eqref{op3}, we obtain
\begin{align}
\| \< N^\frac12\>^{m}   \<p\>^{-1} (p \cdot a (k\bold B_K) )  \<p\>^{-1}  \< N^\frac12\>^{-m-1} \| & \lesssim \|  a (k \bold  B_K)   \<p\>^{-1}  \< N^\frac12\>^{-1} \|   \lesssim K^{-3/4},\\
\| \< N^\frac12\>^{m}   \<p\>^{-m-1} (a^*(k \bold  B_K) \cdot p)  \<p\>^{-1}  \< N^\frac12\>^{-1} \| & \lesssim \|  \< N^\frac12\>^{-1}  \<p\>^{-1} (a^*(k \bold  B_K) \cdot p)  \<p\>^{-1}  \| \notag\\
& \lesssim \|  a (k \bold  B_K)   \<p\>^{-1}  \< N^\frac12\>^{-1} \| \notag\\
& \lesssim K^{-3/4}.
\end{align} 
The second bound follows analogously, invoking also \eqref{op1}. 

For the commutator
\begin{align} 
[H_f, A_K] =  a^*( \omega \bold  G_K )  - a( \omega \bold  G_K )   + 2 a^*( \omega k \bold  B_{K } )\cdot p - 2  p\cdot  a ( \omega k \bold  B_{K} ),
\end{align}
we apply \eqref{op2} and \eqref{op4}:
\begin{align}
\|  \< N^\frac12\>^{m}  \< p \>^{-2} a (\omega \bold  G_K) \< p\>^{-2} \< N^\frac12 \>^{-m-1} \| &  \lesssim K \\
\|  \< N^\frac12\>^{m}  \< p\>^{ -2} ( p \cdot a(\omega k \bold  B_K) )  \< p \>^{-2} \<N^{\frac12}\>^{-m-2} \| & \lesssim  \|  a(\omega k \bold  B_K)    \< p \>^{-2} \<N^{\frac12}\>^{-1} \|  \lesssim K^{-\frac12} ,
\end{align}
and analogously for the terms with $a^*(\omega G_K)$ and $ a^*(\omega k \bold B_K) \cdot p$.

To obtain the fourth bound, we apply \eqref{op3} and \eqref{op5:b}:
\begin{align}
\|  \< N^\frac12 \>^{m}  \< p \>^{-1} a^*(k \bold B_K) a(k\bold B_K) \< p \>^{-1} \< N^\frac12 \>^{-m-2} \| & = \| \< N^\frac12 \>^{-1}  \< p \>^{-1} a^*(k \bold B_K) a(k\bold B_K) \< p \>^{-1} \< N ^\frac12 \>^{-1} \| \notag\\
& \lesssim \| a(k\bold B_K) \< p \>^{-1} \< N^\frac12  \>^{-1} \|^2\notag\\
& \lesssim  K^{-3/2} , \\
\|  \< N^\frac12 \>^{m}  \< p \>^{-1} a( k \bold B_{K} )^2 \< p \>^{-1}  \< N^\frac12 \>^{-m-2} \| & \lesssim  \| a( k \bold B_{K} )^2 \< p \>^{-1}  \< N^\frac12  \>^{-2} \|  \lesssim  K^{-\frac34}  ,
\end{align}
and similarly for $a^*(k \bold B_K)^2$. The last bound is a direct consequence of the previous one.
\end{proof}

\subsection{Bounds for the classical potentials}
\label{sec:classical:potentials}

We conclude with estimates for the classical potential $V_\varphi$ and the source term $\sigma(\psi)$, defined in \eqref{eq:def:h} and \eqref{eq:def:source}, respectively. These bounds are used mainly for the aKG analysis in Section~\ref{sec:properties:adiabatic_SKG}.

\begin{lemma}\label{lemma:form:bound:V}
There exists a constant $C>0$ such that the following bounds hold:
\begin{subequations}
\begin{alignat}{3}
\| V_{\varphi} \psi \|_{L^2} 
& \le C \| \varphi \|_{\mathfrak h_{1/2}} \, \| \psi \|_{H^1}  
&\qquad\quad & \forall\, \varphi \in \mathfrak h_{1/2}, 
&\quad & \psi \in H^1,
\label{eq:V:bound:1}
\\[0.5mm]
 \| V_{\varphi} \psi \|_{L^2} 
& \le C \| \varphi \|_{\mathfrak h_{-s}} \, \| \psi \|_{H^2}  
&\qquad\quad & \forall\, \varphi \in \mathfrak h_{-s}, 
&\quad & \psi \in H^2, \quad s\in  [0,  \tfrac34 ],
\label{eq:V:bound:2}
\\[0.5mm]
| \< \psi , V_{\varphi} \psi \> | 
& \le    C  \| \varphi \|_{\mathfrak h_{-1/2}}   \| \psi \|_{H^1}^2 
&\qquad\quad & \forall\, \varphi \in \mathfrak h_{-1/2}, 
&\quad & \psi \in H^1, 
\label{eq:V:bound:3}
\\[0.5mm]
| \< \psi , V_{\varphi} \psi \> | 
& \le C  \| \varphi \|_{\mathfrak h_{1/2}}   (\delta^{-1} \| \psi \|_{L^2}^2 + \delta  \| \psi \|_{H^1}^2 ) 
&\qquad \quad & \forall\, \varphi \in \mathfrak h_{1/2}, 
&\quad & \psi \in H^1, \quad  \delta >0.
\label{eq:V:bound:4}
\end{alignat}
\end{subequations}
In particular, the last inequality implies for all $\varphi \in \mathfrak h_{1/2}$ and  $\psi \in H^1$:
\begin{align}\label{eq:H1:bound:Hamiltonian}
\| \psi \|_{H^1}^2 \le C (1+ \| \varphi \|_{\mathfrak h_{1/2}}^2 ) \| \psi \|_{L^2}^2 + C  \< \psi , (p^2 + V_\varphi)  \psi \>  .
\end{align}
\end{lemma}

\begin{proof} 
The first two bounds can be viewed as commutative versions of Lemma~\ref{lem:Frank-Schlein}.  
Set $f=\omega^{-1/2}$ and replace the operator $a_k$ by the function $\varphi(k)$.  
Proceeding as in Lemma~\ref{lem:Frank-Schlein}, we find for $s\in\mathbb R$ and $r\ge0$:
\begin{align}
 \Big\| \int \d k\, \omega(k)^{-1/2} e^{ikx} \varphi(k)\, \psi \Big\|^2 
 & = \int \d k\,\d \ell\, \omega^{-\frac{1}{2}}(k)\,{\varphi(k)}\,\omega^{-\frac{1}{2}}(\ell)\,{\varphi(\ell)}\,
   \Big\< e^{i k x } \psi ,  e^{i\ell x}  \psi \Big\> \notag\\
& \le \bigg\| \int \d k \frac{1}{\omega^{s+1}(k)\,\< |p - k|^r \>^2} \bigg\| \, \| \omega^{s/2} \varphi \|_2^2 \, \| \< |p|^r \> \psi \|^2 \notag\\[2mm]
& \le \mathscr C_{s,r} (\omega^{-1/2})\, \|   \varphi \|_{\mathfrak h_{s/2}}^2 \,  \| \psi \|^2_{H^r} 
\end{align}
with $\mathscr C_{s,r}(f)$ defined in Lemma \ref{lem:Frank-Schlein}. From \eqref{eq:symm:rearrangement}, we infer that $\mathscr C_{1,1} (\omega^{-1/2}) \lesssim 1$ and $\mathscr C_{-s,2} (\omega^{-1/2}) \lesssim 1$ for $s\in [ 0 , \frac34 ]$, which yields the first two statements.

For \eqref{eq:V:bound:3}, apply Cauchy–Schwarz:
\begin{align}
| \< \psi , V_{\varphi} \psi \> | 
& \le 2 \int \d k\, \omega^{-1/2}(k)\, |\varphi(k)|\, | \widehat {|\psi|^2}(k) | 
\le  C \| \varphi \|_{\mathfrak h_{-1/2}} \| \psi \|_{L^4}^2 .
\end{align} 
By Sobolev's inequality, $\| \psi \|_{L^4} \lesssim \| \psi \|_{H^1}$, this yields the claim.

To prove the last bound, decompose $\varphi=\varphi_K+\varphi_{>K}$ with $\varphi_K (k) := \mathbf 1_{|k|\le K} \varphi(k)$, and estimate
\begin{align}
| \< \psi , V_{\varphi_K} \psi \> | 
& \le \| \psi \|_{L^2}^2 \int_{|k|\le K} \d k\, \omega(k)^{-1/2} | \varphi (k) |
\le C K\, \| \varphi \|_{\mathfrak h_{1/2}} \| \psi \|_{L^2}^2  , \\
| \< \psi , (V_{\varphi} - V_{\varphi_K}) \psi \> | 
& \le  2 K^{-1} \int_{|k|\ge K} \d k\, \omega(k)^{1/2} |\varphi(k)|\, | \widehat {|\psi|^2}(k)|
\le 2K^{-1} \| \varphi\|_{\mathfrak h_{1/2}} \| \psi \|_{L^4}^2 ,
\end{align}
and hence
\begin{align}
|\< \psi , V_\varphi \psi \> | \le C (K \| \psi \|_{L^2}^2 + K^{-1} \| \psi \|_{H^1}^2 )  \| \varphi \|_{\mathfrak h_{1/2}} 
\end{align}
for all $K>0$. Setting $\delta := K^{-1}$ proves the statement.
\end{proof}
\begin{lemma}\label{lem:sigma:bounds} 
There is $C>0$ such that for all $L^2$-normalized $\psi \in H^1$ and $s\in [-\frac12,\frac12]$,
\begin{align}\label{eq:sigma:bound:s}
\|\sigma(\psi)\|_{\mathfrak h_{s}}
& \le C \| \psi\|_{H^1}^2
\end{align}
Moreover, for all $L^2$-normalized $\psi ,\phi \in H^1$,
\begin{align}
\| \sigma(\psi)-\sigma(\phi) \|_{\mathfrak h_{-s}}
& \le C\,\big(\|\psi\|_{H^1}+\|\phi\|_{H^1}\big)\, \|\psi-\phi\|_{L^2}, \qquad s \in [ 0,1/2 ],\label{eq:sigma:difference:bound:s}
\\[0.5mm]
\| \sigma(\psi) - \sigma(\phi) \|_{\mathfrak h_{1/2}} & \le C  \big( \|  \psi \|_{H^1}  + \| \phi \|_{H^1} \big) \| \psi - \phi \|_{H^1} .
\end{align}
\end{lemma}
\begin{proof}  Recall that $\sigma(\psi)(k) = \omega(k)^{-1/2} \widehat {|\psi|^2}(k)$. For the first bound, we split the integral into
\begin{align}
\| \sigma(\psi) \|_{\mathfrak h_{s}}^2  =\int\limits_{|k|\le 1} \frac{|\widehat{|\psi|^2}(k)|^2}{|k|^{1+2s}}  \d k  + \int\limits_{|k|>1} \frac{|\widehat{|\psi|^2}(k)|^2}{|k|^{1+2s}}  \d k   =:I_1 + I_2.
\end{align}
The low-momentum part satisfies $I_1 \le  \| \psi \|_{L^2}^4 \int_{|k|\le 1 } |k|^{-1-2s} \d k \lesssim 1$ for $s\in [-\frac12,\frac12]$. For the large-momentum part, we use Plancherel and Sobolev to estimate
\begin{align}
I_2
\;\le\; \int_{\mathbb{R}^3} |\widehat{|\psi|^2}(k)|^2 \, \d k
= \| \psi \|_{L^4}^4 \lesssim \| \psi \|_{H^1}^4.
\end{align}

For the second bound, set $\rho:=|\psi|^2-|\phi|^2$. 
By Plancherel, we have
\begin{align}
\|\sigma(\psi)-\sigma(\phi)\|_{\mathfrak h_{-s}}^2
=\int \frac{|\widehat{\rho}(k)|^2}{|k|^{1+2s}} \, \d k =  c_s \iint \frac{\rho(x)\,\rho(y)}{|x-y|^{\,2-2s}}\, \d x\, \d y,
\end{align}
for some constant $c_s$.  Now apply the Hardy--Littlewood--Sobolev inequality \cite[Theorem 4.3]{LiebLoss}  for $d=3$, $\lambda=2-2s$, $s\in [0,\tfrac12]$, and 
\begin{align}
p=q= p_s := \frac{3}{2+s} \in \big[\tfrac{6}{5},\tfrac{3}{2}\big],
\end{align}
to obtain $
\|\sigma(\psi)-\sigma(\phi)\|_{\mathfrak h_{-s}}^2
\le C_s\,\|\rho\|_{L^{p_s}}^2$. Write 
\begin{align}\label{eq:f:identity}
\rho=|\psi|^2-|\phi|^2 = (\psi-\phi)\,\overline{\psi}+(\overline{\psi-\phi})\,\phi
\end{align} 
and take $q_s$ with
$\frac{1}{p_s}=\frac{1}{2}+\frac{1}{q_s}$, i.e., $
q_s= \frac{6}{1+2s}\in [2,6]$. Then by Hölder and Sobolev $H^1\hookrightarrow L^{q_s}$, we find
\begin{align}
\|\rho\|_{L^{p_s}}
\le \|\psi-\phi\|_{L^2}\big(\|\psi\|_{L^{q_s}}+\|\phi\|_{L^{q_s}}\big)
\le C\,\|\psi-\phi\|_{L^2}\big(\|\psi\|_{H^1}+\|\phi\|_{H^1}\big)
\end{align}
for all $s\in [0,1/2]$ (Note that $C_s\lesssim 1$ uniformly on this interval).

For the last estimate, we have
\begin{align}
\|\sigma(\psi) - \sigma(\phi)\|_{\mathfrak h_{1/2}}
&= \||\psi|^2 - |\phi|^2\|_{L^2} .
\end{align}
Using \eqref{eq:f:identity} and Hölder's inequality, we find
\begin{align}
\||\psi|^2 - |\phi|^2\|_{L^2}
&\le \|\psi - \phi\|_{L^6} \, \big( \|\psi\|_{L^3} + \|\phi\|_{L^3} \big).
\end{align}
By Sobolev embedding $H^1(\mathbb R^3)\hookrightarrow L^6(\mathbb R^3)$,
\begin{align}
\|\psi-\phi\|_{L^6} \le C\,\|\nabla(\psi-\phi)\|_{L^2} \le C\,\|\psi-\phi\|_{H^1}.
\end{align}
To bound the $L^3$-norms, we use $\|\psi\|_{L^3}
\le C\, \|\psi\|_{H^1}^{1/2}$. Thus
\begin{align}
\|\sigma(\psi)-\sigma(\phi)\|_{\mathfrak h_{1/2}} = \||\psi|^2-|\phi|^2\|_{L^2}
&\le C\,\|\psi-\phi\|_{H^1}\,\big(\|\psi\|_{H^1}+\|\phi\|_{H^1}\big).
\end{align}
as claimed.
\end{proof}

\section{The dressed Nelson Hamiltonian}
\label{sec:dressed:Nelson}
In this section, we introduce the dressed Nelson Hamiltonian, which is related to the renormalized Nelson  Hamiltonian $H^\ve$ through the unitary dressing transformation \eqref{eq:Gross:trafo}. The dressing transformation allows us to isolate the UV singular term and work with the dressed Hamiltonian. While the dressed Hamiltonian has a more complicated (quadratic) structure, it involves more regular kernels. In a second step, we derive a suitable Weyl-shifted representation of the dressed Hamiltonian.

\begin{lemma}\label{lem:dressed:Nelson} Let $U_K^\ve$ be defined as in \eqref{eq:Gross:trafo}.
There exists $K_0>0$ such that for all $K \ge K_0 $ and $\ve > 0$, the renormalized Hamiltonian satisfies
\begin{align} \label{dressing:identity}
H^{\ve} = (U_K^\ve)^*  H_{\mathscr D,K}^{\ve} U_K^\ve ,
\end{align}
where the dressed Hamiltonian is given by
\begin{align}\label{eq:H:ren}
& H_{\mathscr D,K}^{\ve}   = p^2 + \ve \Big( \phi( \bold G_K ) + 2 a^*( k \bold  B_K) \cdot p +  2  p \cdot a ( k  \bold  B_K ) \Big) \notag\\[1mm]
&\quad + \ve^2   \Big( H_f  + 2 a^*( k \bold  B_K ) a( k \bold  B_K ) +  a^*( k \bold  B_K )a^*( k \bold  B_K ) + a( k \bold  B_K )  a( k \bold  B_K )  + 4\pi \ln K \Big) \notag\\[1.5mm]
& \quad  -  \ve^3 \phi(\omega  \bold  B_K) + \ve^4 \| \omega^{\frac12} B_K \|^2_{L^2} 
\end{align}
with quadratic form domain $Q(H_{\mathscr D,K}^{\ve}) = Q(p^2 + H_f )$. 

Moreover, there are constants $C,K_0>0$ such that
\begin{align}\label{eq:energy:bound}
\pm \big( H_{\mathscr D,K}^{\ve}  - \big( p^2 + \ve^2 H_f \big) \big) & \le \tfrac12 \big( p^2 + \ve^2 H_f \big)  + C  K ,\\[1mm]
\label{eq:energy:bound:b}
\pm \big( H_{\mathscr D,K}^{\ve}  - \big( p^2 + \ve^2 H_f \big) \big) & \le \tfrac12 \big( p^2 + \ve^2 H_f \big)  + C\ve^2  ( \ln K ) (N+1) 
\end{align}
in the sense of quadratic forms, for all $K\ge K_0$ and $\ve > 0$.
\end{lemma}

The dressing identity for $H^\varepsilon$ follows the same strategy as in \cite[Lem.~3.1]{Griesemer18}. 
Because of the semiclassical scaling and our slightly modified choice of dressing transformation, 
we give the details in Appendix~\ref{app:renormalization}.  

The two energy bounds follow from Lemmas~\ref{lem:operator:bounds:a} 
and \ref{lem:operator:bounds:b}. For \eqref{eq:energy:bound}, the contribution proportional to $K$ comes from 
the field operator, estimated as
\begin{align}
\big|\langle \Psi , \varepsilon\, a(\boldsymbol{G}_K)\,\Psi \rangle \big|
   \;\le\; C\,K^{1/2}\,\|\Psi\|\,\|\varepsilon\,H_f^{1/2}\Psi\| .
\end{align}
In \eqref{eq:energy:bound:b}, the dependence on $K$ is improved at the expense of involving the number operator. 
Here one uses in particular
\begin{align}\label{eq:LY:bound}
\big|\langle \Psi , \varepsilon\, a(\boldsymbol{G}_K)\Psi \rangle\big|
   \;=\;\big|\langle \langle N^{1/2}\rangle \Psi , \varepsilon\,a(\boldsymbol{G}_K)N^{-1/2}\Psi\rangle\big|
   \;\le\; C\,\varepsilon\,(\ln K)^{1/2}\,\|\langle N^{1/2}\rangle\Psi\|\,\|\langle p\rangle\Psi\| .
\end{align}
The estimates for the remaining terms are straightforward and omitted here.

Next, we compute a Weyl-shifted representation of the dressed Hamiltonian, which will be useful for the proof in Section \ref{sec:proof:norm:approximation}. This transformation isolates the leading-order particle Hamiltonian and expresses the remainder in increasing powers of $\ve$. For $\varphi \in \mathfrak h_{1/2}$, consider
\begin{align}
 \mathscr H^{\ve}_{\mathscr D,K,\varphi}
& 
:=  e^{-i g_{K,\varphi}(x)} W(\ve^{-1} \varphi )^* H^{ \ve }_{\mathscr D , K}   W(\ve^{-1} \varphi)   e^{ig_{K,\varphi}(x)} \label{eq:mathscr:H}
\end{align}
where $g_{K,\varphi }(x) =  2 \Im \langle \varphi , \bold B(\cdot, k) \rangle$, so that $e^{-ig_{K,\varphi}(x)}$ corresponds to the classical analogue of the dressing transformation $U_K^\ve$. A straightforward computation yields
\begin{align}\label{eq:Weyl:transformed:Hamiltonian}
W(\ve^{-1} \varphi )^* H^{ \ve }_{\mathscr D ,K}   W(\ve^{-1} \varphi) & =  \mathfrak h_0 + \ve \mathfrak  h_1 + \ve^2  \mathfrak h_2 - \ve^3 \phi(\omega \bold  B_K ) + \ve^4 \| \omega^{\frac12} B_K \|^2_{L^2} 
\end{align}
with
\begin{align}
 \mathfrak h_0 & =  p^2 + 2 \Re\langle \varphi , \bold  G_K \rangle  +  2  \langle \varphi, k \bold  B_K \rangle\cdot  p +  2 p \cdot \langle k \bold  B_K  , \varphi \rangle  +  4 |\Re \langle \varphi , k \bold  B_K \rangle |^2 + \| \omega^{\frac12} \varphi \|_{L^2}^2 , \\[1.5mm]
 \mathfrak h_1 & = \phi(\bold  G_K + \omega \varphi ) + 2 a^*(k \bold  B_K)\cdot p +2 p \cdot a(k \bold  B_K) + 4\Re \langle \varphi , k \bold  B_K \rangle \cdot  \phi(k \bold  B_K), \\[1.5mm] 
 \mathfrak h_2 & = H_f +  2 a^*(k \bold  B_K) a (k \bold  B_K) + a^*( k \bold B_K )^2+ a( k \bold  B_K )^2  + 4\pi \ln K - 2 \Re \langle \bold  B_K , \omega \varphi \rangle .
\end{align}
Transforming \eqref{eq:Weyl:transformed:Hamiltonian}  with $e^{-i g_{K,\varphi}(x)}$ then yields
\begin{align} \label{eq:Weyl:dressed:Ham}
\mathscr H_{\mathscr D,K,\varphi}^{\ve} &  = h_\varphi + \|\omega^{\frac12} \varphi \|^2_{L^2} + \ve \Big( A_K + \phi(\omega\varphi ) \Big)  + \ve^2 \Big( H_f + D_K - 2 \Re \langle \omega \varphi , \bold  B_K \rangle  \Big) \notag\\[1mm]
&\quad - \ve^3 \phi(\omega \bold  B_K) + \ve^4 \| \omega^{\frac12} B_K \|^2_{L^2} 
\end{align}
where $h_\varphi = p^2 + 2 \Re  \langle \varphi ,  \bold  G \rangle$ as in \eqref{eq:SKG}, and we recall the definitions of $A_K$ and $B_K$ in \eqref{eq:def:A} and \eqref{eq:def:D}, respectively

Let us briefly explain where the different terms come from. Using
\begin{align}
e^{-ig_{K,\varphi }(x)}\, p \,e^{ig_{K,\varphi }(x)} = p - 2\,\Re\langle \varphi, k\bold  B_K\rangle ,
\end{align}
one obtains
\begin{align}
e^{-ig_{K,\varphi}(x)}  \mathfrak h_0 e^{ig_{K,\varphi}(x)} & = \big(  p - 2 \Re \langle \varphi, k  \bold  B_K \rangle \big)^2 + 2 \Re \langle \varphi ,  \bold  G_K \rangle + 2\langle \varphi , k  \bold B_K \rangle \cdot \big( p - 2 \Re \langle \varphi , k \bold  B_K \rangle \big)  \notag\\[1mm]
& \quad  +  \big( p- 2 \Re \langle \varphi , k \bold  B_K \rangle\big) \cdot \langle k  \bold  B_K , \varphi \rangle + 4 |\Re \langle \varphi , k \bold  B_K \rangle |^2 + \| \omega^{\frac12 } \varphi \|^2_{L^2} \notag\\[1mm]
& =  p^2 + 2 \Re \langle \varphi , \bold  G_K \rangle  + 2\Im \langle \varphi , k \bold  B_K \rangle \cdot p + 2 p \cdot \Im \langle k \bold  B_K , \varphi \rangle +  \| \omega^{\frac12} u\|_{L^2}^2\notag\\[1mm]
& = p^2 + 2 \Re \langle \varphi , \bold  G_K \rangle  +  2 \Re \langle \varphi , k^2  \bold  B_K \rangle + \|\omega^{\frac12} \varphi \|_{L^2}^2
\end{align}
where we used $ p \cdot \Im \langle \varphi , k \bold  B_K \rangle  = - \Im \langle k  \bold  B_K , \varphi \rangle \cdot p  + \Re \langle \varphi , k^2  \bold  B_K \rangle $ in the last step. Invoking the identity $ \bold  G_K + k^2  \bold  B_K =  \bold  G $ explains the presence of the operator $h_\varphi$. We further note that
\begin{align}
& e^{-ig_{K,\varphi}(x)}  \mathfrak h_1 e^{ig_{K,\varphi}(x) } 
\notag\\
& \quad = \phi( \bold  G_K + \omega \varphi ) + 2 a^*(k \bold  B_K) \cdot \big( p - 2 \Re \langle \varphi, k \bold  B_K \rangle  \big) + 2 \big( p - 2 \Re \langle \varphi, k \bold  B_K \rangle  \big) \cdot  a(k \bold  B_K) \notag\\[1mm]
& \quad \quad + 4\Re \langle \varphi, k \bold  B_K \rangle \cdot \phi(k \bold  B_K) \notag\\[1mm]
& \quad = \phi( \bold  G_K + \omega \varphi ) + 2 a^*(k  \bold  B_K) \cdot p  +  2  p \cdot a(k \bold  B_K) 
\end{align}
which explains the terms of order $O(\ve) $. The higher-order terms in $\ve$ in \eqref{eq:Weyl:transformed:Hamiltonian} do not involve gradients, so that they commute with $e^{ig_{K,\varphi} (x)}$.

\section{Well-posedness and properties of the aKG equation \label{sec:properties:adiabatic_SKG}}

In the following we prove Proposition \ref{prop:aKG-local}, that is, we establish local well-posedness of the aKG equation. Moreover,  we derive an equation of motion for the ground state $\psi_{\varphi_t}$ and establish some important estimates for $\psi_{\varphi_t}$ and the resolvent $R_{\varphi_t}$.

We recall the notation for the lowest eigenvalue $e(\varphi) = \inf \sigma (h_\varphi)$ of the operator $h_\varphi$ as well as for the spectral gap $\triangle (\varphi) = \inf\{  |e(\varphi)-\lambda| : \lambda\in \sigma(h_{\varphi})\setminus\{e(\varphi)\} \} $. Let us further recall that $\varphi \in C([-T,T];\mathfrak{h}_{1/2})$ is called a solution to \eqref{eq:aKG} with initial condition $\varphi_0\in \mathfrak{h}_{1/2}$, in case that $e(\varphi_t)<0$ for all $t\in [-T,T]$ and
\begin{align}
\label{eq:adiabatic_solution_Integral.x}
\varphi_t =e^{-i\omega t}\varphi_0-i\int_0^t e^{-i\omega (t-s)} \sigma( \psi_{\varphi_s}) \mathrm{d}s,
\end{align}
where $\psi_{\varphi_t}\ge 0$ is the unique ground state of $h_{\varphi_t}$ and $\sigma(\psi_{\varphi_s}) = \omega^{-1/2} \widehat {|  \psi_{\varphi_s} |^2}$. Existence and uniqueness of $\psi_{\varphi_t}$ follow from the condition $e(\varphi_t) < 0$.

We prepare the proof of Proposition~\ref{prop:aKG-local} with two lemmas: the first establishes local stability of the spectral gap, and the second shows that the map $\varphi\mapsto\psi_\varphi$ is locally Lipschitz.

\begin{lemma}\label{lem:spectral:gap}
Let $\varphi\in\mathfrak h_{1/2}$ with $e(\varphi)<0$ and spectral gap $\triangle(\varphi)>0$ and let
\begin{align}
B_r(\varphi) := \big\{ u \in \mathfrak h_{1/2} : \| u - \varphi \|_{\mathfrak h_{1/2}} \le r \big\}.
\end{align}
Then there exist $C>0$ and $r_0>0$ such that for all $0<r\le r_0$ and all $u\in B_r(\varphi)$,
\begin{align}
\triangle(u)\ \ge\ \tfrac12\,\triangle(\varphi )\ >0 .
\end{align}
Moreover, $|e(\varphi) - e(u)| + |\lambda_1(\varphi) - \lambda_1(u) | \le C r$ for all $u \in B_r(\varphi)$, where $\lambda_1(u)$ denotes the first excited eigenvalue of $h_u$.
\end{lemma}

\begin{proof}
By Lemma~\ref{lemma:form:bound:V}, there is $C>0$ so that for all
$\varphi,\xi\in\mathfrak h_{1/2}$ and all $\psi\in H^1$ with $\|\psi\|_{L^2}=1$,
\begin{align}\label{eq:form-diff}
\big|\<\psi,(h_\varphi-h_\xi)\psi\>\big|
\ \le\ C\,\|\varphi-\xi\|_{\mathfrak h_{1/2}}\ \|\psi\|_{H^1}^2 .
\end{align}
Let $e(u)=\inf\sigma(h_u)$ and let $\psi_{\varphi}$ denote the ground state of $h_{\varphi}$. The variational principle combined with the above inequality gives
\begin{align}\label{eq:e(u):upper:bound}
e(u)\ \le\ \<\psi_{\varphi},h_u\psi_{\varphi}\>
\ \le\ e(\varphi)+C\,\|\varphi- u \|_{\mathfrak h_{1/2}} \| \psi_{\varphi}\|_{H^1}^2 \ \le\ e(\varphi) + C r ,
\end{align}
and thus for $r$ sufficiently small, $e(u)<0$. Denoting now by $\psi_u$ the ground state of $h_u$ and exchanging $\varphi , u$ in the above estimate yields
\begin{align}\label{eq:e-Lip}
|e(\varphi)-e(u)| \ \le\ C r
\end{align}
for some constant $C$ depending only on $\varphi$. Note that here we also used the estimate 
\begin{align}
\| \psi_u \|_{H^1}^2 \le C \big( 1+ \| u \|_{\mathfrak h_{1/2}}^2  + e(u) \big) \le C \big( 1 + \| \varphi \|_{\mathfrak h_{1/2}}^2 \big) \le C 
\end{align}
as a consequence of Lemma \ref{lemma:form:bound:V}.

Next, write $\lambda_1(u):=\inf\big(\sigma(h_u)\setminus\{e(u)\}\big)$,  and recall by the max-min principle that
\begin{align}
\lambda_1(u)=\sup_{\mathrm{codim}(V)=1}\ \inf_{\psi\in V} \big\{  \langle\psi,h_u \psi\rangle : \| \psi \|_{L^2}=1\big\}.
\end{align}
Choosing $V=\{ \psi_\varphi\}^\perp$ and using \eqref{eq:form-diff}, we obtain
\begin{align}
\lambda_1(u)
&\ge \inf_{\substack{\psi\in H^1,\ \|\psi\|=1\\ \psi\perp\psi_{\varphi } }}
\<\psi,h_u\psi\>
\ \ge\ \inf_{\substack{\psi\in H^1,\ \|\psi\|=1\\ \psi\perp\psi_{\varphi}}} \left(  \<\psi,h_{\varphi}\psi\>\;-\;C r  \| \psi \|_{H^1}^2 \right) .
\end{align}
Using the bound \eqref{eq:H1:bound:Hamiltonian} and $\| \varphi \|_{\mathfrak h_{1/2}} \le C $, we have $\| \psi \|_{H^1}^2 \le C ( 1 + \< \psi , h_\varphi \psi \>)$, and thus
\begin{align}
\lambda_1(u)
&\ge (1-Cr) \inf_{\substack{\psi\in H^1,\ \|\psi\|=1 \\ \psi\perp\psi_{\varphi}}} \<\psi,h_{\varphi}\psi\> - C r = (1-Cr) \lambda_1(\varphi) - C r .
\end{align}
Since $\triangle (u)=\lambda_1(u)-e(u)$, combining the above estimate with \eqref{eq:e-Lip}, we arrive at
\begin{align}
\triangle(u)\ \ge\ \triangle (\varphi)-C r .
\end{align}
Choosing $r_0>0$ sufficiently small then yields $\triangle(u)\ge \tfrac12\,\triangle(\varphi)>0$ for all $u\in B_r(\varphi)$ with $0<r\le r_0$. Exchanging the roles of $u$ and $\varphi$ in the argument for $\lambda_1$ gives the reverse comparison
\begin{align}
\lambda_1(\varphi)
&\ge (1-Cr)\,\lambda_1(u) - C r 
\end{align}
with a constant $C$ that depends only on $\varphi$, hence $|\lambda_1(u)-\lambda_1(\varphi)| \le C r$. Together with \eqref{eq:e-Lip} this proves the last statement. 
\end{proof}

\begin{lemma} \label{lem:ground:state:difference}Let  $\varphi \in \mathfrak h_{1/2}$ with $e(\varphi )<0$, and define $B_r(\varphi)$ as in Lemma \ref{lem:spectral:gap}. There are constants $C>0$ and $r_0>0$ such that
\begin{align}
\| \psi_u - \psi_\xi \|_{H^1} & \le \frac{C}{\triangle (\varphi)^{1/2}}   \left( \|\psi_u \|_{H^1} + \| \psi_\xi \|_{H^1} \right) \| u - \xi\|_{\mathfrak h_{1/2}} 
\end{align}
for all $u , \xi \in B_r(\varphi)$ with $0<r\le r_0$, where $\psi_u$ and $\psi_\xi$ denote the normalized non-negative ground states of $h_\varphi$ and $h_\xi$, respectively. Moreover
\begin{align}
\| \psi_u - \psi_\xi \|_{L^2} & \le\frac{C}{\triangle (\varphi)^{1/2}}   \left( \|\psi_u \|_{H^2} + \| \psi_\xi \|_{H^2} \right)  \| u - \xi\|_{\mathfrak h_{-1/2}}
\end{align}
for all $ u, \xi \in B_r(\varphi) \cap \mathfrak h_{-1/2}$ with $0< r \le r_0$.
\end{lemma}
\begin{proof} 
We apply \eqref{eq:H1:bound:Hamiltonian} to $\psi_{u} - \psi_{\xi} \in H^1$ and $u-\xi\in \mathfrak h_{1/2}$, thereby obtaining
\begin{align}\label{eq:middle:term}
\| \psi_{u} - \psi_{\xi} \|_{H^1}^2 \le C \left( (1+ \| u - \xi\|^2_{\mathfrak h_{1/2}} ) \| \psi_{u}  - \psi_{\xi} \|_{L^2}^2 + \<  \psi_{u} - \psi_{\xi} , h_{u - \xi} (\psi_{u} - \psi_{\xi} )\>   \right).
\end{align}
For the $L^2$-norm, we use the bound $\| \psi_{u} - \psi_{\xi}\|_{L^2}^2 \le 2 \| Q_{u} \psi_{\xi} \|_{L^2}^2$, where $Q_{u} = \boldsymbol{1} - | \psi_{u} \rangle \langle \psi_{u} |$, as well as
\begin{align}
\label{eq:L_2:for_xi_varphi}
\| Q_{u} \psi_{\xi} \|_{L^2}^2 \le \frac{1}{\triangle(u)}\langle \psi_{\xi},(h_{u}-e(u)) \psi_{\xi}  \rangle  \le \frac{2}{\triangle(\varphi)}\langle \psi_{\xi},(h_{u}-e(u)) \psi_{\xi}  \rangle 
\end{align}
where we we applied Lemma \ref{lem:spectral:gap} in the last step. We continue to estimate
\begin{align}
\< \psi_{\xi},(h_{u}-e(u)) \psi_{\xi}  \> &= \< \psi_{\xi} - \psi_u , h_{u}  \psi_{\xi}  \> - \< \psi_{\xi} - \psi_u , \psi_{\xi}  \>  \< \psi_{u} , h_u \psi_u \> \notag\\
&=  e(\xi) \< \psi_{\xi} - \psi_u ,  \psi_{\xi}  \> +  \< \psi_{\xi} - \psi_u ,  V_{u - \xi} \psi_{\xi}  \>  \notag\\
&\quad - \< \psi_{\xi} - \psi_u , \psi_{\xi}  \>  \left( \< \psi_{u} , h_\xi \psi_u \> +  \< \psi_{u} , V_{u - \xi} \psi_u \>  \right),
\end{align}
where we used $h_u = h_\xi+  V_{u - \xi}$. Since $e(\xi) \le  \< \psi_{u} , h_\xi \psi_u \>$ and $\< \psi_\xi - \psi_u , \psi_\xi \> \ge 0$, it follows that
\begin{align}
\< \psi_{\xi},(h_{u}-e(u)) \psi_{\xi}  \> &\le \| \psi_\xi - \psi_u \|_{L^2} \left( \| V_{u -\xi} \psi_\xi \|_{L^2} +  \|  V_{u -\xi} \psi_u \|_{L^2} \right) \notag\\[1mm]
&\le C  \| \psi_\xi - \psi_u \|_{L^2} \left( \|\psi_\xi \|_{H^1} + \| \psi_u \|_{H^1} \right) \| u - \xi\|_{\mathfrak h_{1/2}} \notag\\
&\le \tfrac{1}{10} \| \psi_\xi - \psi_u \|_{L^2}^2 + C \left( \|\psi_u \|_{H^1} + \| \psi_\xi \|_{H^1} \right)^2 \| u - \xi\|_{\mathfrak h_{1/2}}^2,
\label{eq:mid:term:2}
\end{align}
where the second step follows from \eqref{eq:V:bound:1}. Hence, we get a bound of the form
\begin{align}
\| \psi_u - \psi _\xi \|_{L^2}^2 \le C \triangle(\varphi)^{-1}  \left( \|\psi_u \|_{H^1} + \| \psi_\xi \|_{H^1} \right)^2 \| u - \xi \|_{\mathfrak h_{1/2}}^2.
\end{align}

It remains to estimate the middle term in \eqref{eq:middle:term}. Using $h_{u - \xi} = h_u - V_\xi$:
\begin{align}
\<  \psi_{u} - \psi_{\xi} , h_{u - \xi} (\psi_{u} - \psi_{\xi} )\> &= \<  \psi_{u} - \psi_{\xi} , h_u   \psi_{u} - \psi_{\xi}  \> - \<  \psi_{u} - \psi_{\xi} , V_\xi (\psi_{u} - \psi_{\xi} )\> \notag\\
&\le \<   \psi_{\xi} , ( h_u - e(u))   \psi_{\xi}  \> - \<  \psi_{u} - \psi_{\xi} , V_\xi (\psi_{u} - \psi_{\xi} )\>,
\end{align}
where we also used $e(u)<0$ and $(h_u - e(u) ) \psi_u = 0 $. The first term was estimated in \eqref{eq:mid:term:2}. For the second term we find by Cauchy--Schwarz, \eqref{eq:V:bound:1} and $\| \xi \|_{\mathfrak h_{1/2}}\le C$ that
\begin{align}
| \<  \psi_{u} - \psi_{\xi} , V_\xi (\psi_{u} - \psi_{\xi} )\> | &\le C  \| \psi_{u} - \psi_{\xi} \|_{L^2}  \| \psi_{u } - \psi_{\xi} \|_{H^1}.
\end{align} 
Adding everything together proves the first statement of the lemma. 

For the second one, we proceed as in \eqref{eq:L_2:for_xi_varphi}--\eqref{eq:mid:term:2} to find
\begin{align}
\| \psi_u - \psi_\xi \|_{L^2} \le C \triangle(\varphi)^{-1} \left( \| V_{u -\xi} \psi_\xi \|_{L^2} +  \|  V_{u -\xi} \psi_u \|_{L^2} \right) .
\end{align}
From here, the statement follows from \eqref{eq:V:bound:2} for $s=-\frac12$.
\end{proof}

We can now prove well-posedness of the aKG equation.

\begin{proof}[Proof of Proposition \ref{prop:aKG-local}]
By assumption $e(\varphi_0)<0$, the spectral gap $\triangle(\varphi_0)>0$ is strictly positive. Therefore by Lemma \ref{lem:spectral:gap}, we have $\triangle (u)>0$ for all $u \in B_r(\varphi_0)$ for sufficiently small $r>0$. Thus, for every $u \in  B_r(\varphi_0)$ there is a unique normalized ground state $\psi_u \ge 0$ of $h_u$ with eigenvalue $e(u)<0$, and by Lemma \ref{lemma:form:bound:V} \eqref{eq:H1:bound:Hamiltonian}
\begin{align}
\label{Eq:H_1_Bound_r}
\|\psi_{u}\|_{H^1}\leq C\left(1+\| u \|_{\mathfrak h_{1/2}} \right)\le \lambda
\end{align}
for some $\lambda >0$ depending only on $r$ and $\| \varphi_0 \|_{\mathfrak h_{1/2}}$.

\emph{Fixed-point map.}
On the metric space  $X_T:=C([-T,T];B_r(\varphi_0))$, we define
\begin{align}
\label{Eq:Def_Contraction}
(\mathcal{T}u)(t):=e^{-it \omega}\varphi_0-i\int_0^t e^{i(s-t)\omega}\, \sigma(\psi_{u_s}) \, \mathrm{d}s, \qquad |t|\le T,
\end{align}
which we use to express \eqref{eq:aKG-mild} as a fixed-point equation $\mathcal{T}(\varphi )=\varphi $.

\emph{Contraction.}
To show that $\mathcal{T}$ maps $X_T$ into itself, we first estimate
\begin{align*}
\left\|\mathcal{T}u(t)-\varphi_0\right\|_{\mathfrak{h}_{1/2}}
\leq \left\|(e^{-it\omega}-1)\varphi_0\right\|_{\mathfrak{h}_{1/2}}
+ T \sup_{|s|\le T} \| \sigma(\psi_{u_s}) \|_{\mathfrak h_{1/2}}.
\end{align*}
By dominated convergence  $\lim_{t\rightarrow 0} \|(e^{-it\omega}-1)\varphi_0\|_{\mathfrak{h}_{1/2}}=0$, and using Lemma \ref{lem:sigma:bounds} together with \eqref{Eq:H_1_Bound_r}, we obtain
\begin{align}
\label{eq:adiabatic:proof:Sobolev}
\| \sigma(\psi_{u_s}) \|_{\mathfrak h_{1/2}} \le C \| \psi_{u_s}\|_{H^1}^2 \le C\,\lambda^2.
\end{align}
Choosing $T$ small enough therefore yields $\mathcal{T}(X_T)\subseteq X_T$. 

For $u,\xi\in X_T$, by Lemma~\ref{lem:sigma:bounds} and the Lipschitz continuity of $u\mapsto\psi_u$ (Lemma
\ref{lem:ground:state:difference} and \eqref{Eq:H_1_Bound_r}), we further have
\begin{align}
\|\mathcal T u(t)-\mathcal T \xi(t)\|_{\mathfrak h_{1/2}}
\le \int_0^{|t|} \|\sigma(\psi_{u_s})-\sigma(\psi_{\xi_s})\|_{\mathfrak h_{1/2}}\,\mathrm ds  
&\le C| t|  \lambda\,\sup_{|s|\le T}\|\psi_{u_s}-\psi_{\xi_s}\|_{H^1}  \notag \\
&\le C D |t| \lambda^2 \,\sup_{|s|\le T}\|u_s-\xi_s\|_{\mathfrak h_{1/2}},
\end{align}
hence
\begin{align}
\|\mathcal T u-\mathcal T \xi\|_{X_T}\;\le\; T\,C D \lambda^2\,\|u-\xi\|_{X_T}.
\end{align}
Choose $T$ even smaller, if necessary, so that $T C D \lambda^2 <1$. Then $\mathcal T$ is a strict contraction on $X_T$, and there exists a unique fixed point
$\varphi\in X_T$ with $\mathcal T(\varphi)=\varphi$. By construction, $\varphi$ satisfies \eqref{eq:aKG-mild} and $e(\varphi_t)<0$ for $|t|\le T$, hence
$h_{\varphi_t}$ has a unique ground state for those times. In particular, $\triangle(\varphi_t)\ge \triangle(\varphi_0)/2$ by Lemma \ref{lem:spectral:gap}.

\emph{Uniqueness.}
Let $\xi\in C([-T,T];\mathfrak h_{1/2})$ be any solution of \eqref{eq:aKG-mild} with the same initial datum $\xi_0=\varphi_0$. By Lemma \ref{lemma:form:bound:V} and $e(\xi_{t})<0$, we have $\|\psi_{\xi_s}\|_{H^1}^2\leq C ( 1 + \| \xi_s\|^2_{\mathfrak h_{1/2}})$, and therefore
\begin{align}
\|\xi_t- \varphi _0\|_{\mathfrak h_{1/2}}
\leq \left\|(e^{-it\omega}-1)\varphi_0\right\|_{\mathfrak{h}_{1/2}}
+ C T    \sup_{|s| \le T } C ( 1 + \| \xi_s\|^2_{\mathfrak h_{1/2}}) 
\end{align}
By continuity of $t \mapsto \| \xi_t \|_{\mathfrak h_{1/2}}$, we conclude that $\xi\in C([-T,T];B_r(\varphi_0))$ for $T$ small enough. By the contraction argument above the fixed point is unique in $X_T$, so $\xi=\varphi$ on $[-T,T]$.

\emph{$C^1$-regularity.}
Set $g(t):=\sigma(\psi_{\varphi_t})$. The map $t\mapsto g(t)$ is continuous in $\mathfrak h_{-1/2}$ by Lemma~\ref{lem:sigma:bounds},
\eqref{Eq:H_1_Bound_r}, and Lemma~\ref{lem:ground:state:difference}. Let $v(t):=e^{i\omega t}\varphi_t$. From \eqref{eq:aKG-mild},
\begin{align}
v(t)=\varphi_0 - i\int_0^t e^{i\omega s} \, g(s)\,\mathrm ds,
\end{align}
and since $g\in C([-T,T];\mathfrak h_{-1/2})$, the integral is $C^1$ with $\partial_t v(t)=-i\,e^{i\omega t}\,g(t)$ in $\mathfrak h_{-1/2}$.
Thus $\varphi_t=e^{-i\omega t}v(t)\in C^1([-T,T];\mathfrak h_{-1/2})$ and satisfies the differential aKG in $\mathfrak h_{-1/2}$.

\emph{$L^2$-persistence and continuity.}
Assume $\varphi_0\in L^2$. From \eqref{eq:aKG-mild} and the unitarity of $e^{-i\omega t}$,
\begin{align}
\|\varphi_t\|_{L^2} \;\le\; \|\varphi_0\|_{L^2} \;+\; \int_0^{|t|} \|\sigma(\psi_{\varphi_s})\|_{L^2}\, \d s,
\end{align}
and Lemma~\ref{lem:sigma:bounds} yields $\sup_{|s|\le T}\|\sigma(\psi_{\varphi_s})\|_{L^2}\le C$, so $\sup_{|t|\le T}\|\varphi_t\|_{L^2}\le C(T)$. Moreover, for $t,t_0\in[-T,T]$,
\begin{align}
\|\varphi_t-\varphi_{t_0}\|_{L^2}
\;\le\; \|(e^{-i\omega(t-t_0)}-1)\varphi_{t_0}\|_{L^2} \;+\; \int_{t_0}^{t} \|\sigma(\psi_{\varphi_s})\|_{L^2}\, \d s,
\end{align}
and both terms vanish as $t\to t_0$, hence $\varphi\in C([-T,T];L^2)$.

\emph{Energy conservation.} Using the Hellmann--Feynman identity \eqref{eq:Hellmann-Feynman}, we see that the aKG solution conserves the energy $\mathcal E_{\rm field}$ defined in \eqref{eq:classical:field:energy}:
\begin{align}\label{eq:energy:conservation}
\frac{d}{dt}\mathcal E_{\mathrm{field}}(\varphi_t)
\, =\, \langle \psi_{\varphi_t}, V_{i \omega \varphi_t}\psi_{\varphi_t}\rangle + 2\Re\langle \dot\varphi_t,\omega\varphi_t\rangle \, =\, 0.
\end{align}
Indeed, inserting $\dot\varphi_t=-i(\omega\varphi_t+\sigma(\psi_{\varphi_t}))$ into the second term produces $2\Re\langle \dot\varphi_t,\omega\varphi_t\rangle  = - 2\Im \langle  \sigma(\psi_{\varphi_t}) , \omega \varphi_t\rangle = - \langle \psi_{\varphi_t} , V_{i\omega \varphi_t}  \psi_{\varphi_t} \rangle $, which cancels the first contribution.

\emph{Continuation and alternative.}
Let $I$ be any open existence interval containing $0$. By energy conservation and \eqref{eq:V:bound:4}, there exist $0<\alpha<1$ and $C_0>0$ such that
\begin{align}\label{eq:h12-apriori}
(1-\alpha)\,\|\varphi_t\|_{\mathfrak h_{1/2}}^2
\;\le\; \mathcal E_{\mathrm{field}}(\varphi_0)+C_0,
\qquad \forall t\in I,
\end{align}
hence $\sup_{t\in I}\|\varphi_t\|_{\mathfrak h_{1/2}}<\infty$. On any compact interval $J\subset I$ with $\inf_{t\in J}\triangle(\varphi_t)>0$, Lemma \ref{lem:ground:state:difference} shows that $u\mapsto\psi_u$ is Lipschitz on $\{\varphi_t:t\in J\}$, and hence $u\mapsto\sigma(\psi_u)$ is locally Lipschitz in $\mathfrak h_{1/2}$. Together with the a priori bound \eqref{eq:h12-apriori}, this allows the standard continuation argument to extend the solution beyond $J$. Therefore, if $T_+<\infty$ resp. $T_-<\infty$, the only obstruction to continuation is closing of the gap:
\begin{align}
\liminf_{t\uparrow T_+}\triangle(\varphi_t)=0
\qquad\text{resp.} \qquad \liminf_{t\downarrow -T_-}\triangle(\varphi_t)=0 .
\end{align}
This yields the continuation alternative and shows that the existence interval is maximal.
\end{proof}

In the remainder of this section, we prove estimates on the reduced resolvent 
\begin{align}
R_{\varphi_t} = Q_{\varphi_t} ( h_{\varphi_t} - 
e(\varphi_t)  )^{-1} Q_{\varphi_t},
\end{align}
and derive an equation of motion for $ t\mapsto \psi_{\varphi_t}$. 

\begin{lemma} \label{lem:resolvent:bounds}  Let $\varphi_t$  denote the solution to \eqref{eq:aKG} for initial data $\varphi_0\in \mathfrak h_{1/2}$ with $e(\varphi_0)<0$, and let $\psi_{\varphi_t}\ge 0$ be the non-negative normalized ground state of $h_{\varphi_t}$. Then, there exists $T>0$ such that the following statements hold for all $|t| \le T$.
\begin{itemize}
\item[(i)] There is $C>0$ such that
\begin{align}
\| R_{\varphi_t} \| + \| \< p \> R_{\varphi_t}^{1/2}  \|   + \| \<p^2 \> R_{\varphi_t} \|  + \|  \psi_{\varphi_t}\|_{H^2} & \le  C .
\end{align}
\item[(ii)] The map $t\mapsto \psi_{\varphi_t}$ is in $C((-T,T);H^1) \cap C^1((-T,T);L^2)$ and 
\begin{align}\label{eq:eom:psi:groundstate}
\partial_t \psi_{\varphi_t} =  R_{\varphi_t} V_{i\omega \varphi_t} \psi_{\varphi_t} .
\end{align}
\item[(iii)] There is $C>0$ such that
\begin{align}\label{eq:dot:R:bounds}
\| \dot \psi_{\varphi_t} \|_{L^2} +  \|\dot \sigma ( \psi_{\varphi_t} )  \|_{L^2 } + \|  \dot R_{\varphi_t}  \|  + \|\< p \> \dot R_{\varphi_t} \< p \> \|   & \le C  \ 
\end{align}
where $\| \cdot \|$ denotes the operator norm on $L^2(\mathbb R^3)$, and $\dot f(t) := \partial_t f(t)$,
\end{itemize}
\end{lemma}
\begin{proof} The first bound $\| R_{\varphi_t} \| \lesssim 1$ obviously follows from the spectral gap $\inf_{|t| \le T } \triangle (\varphi_t)>0$. For the second bound, use $p^2 = h_{\varphi_t} - V_{\varphi_t}$ and $h_{\varphi_t} R_{\varphi_t} = Q_{\varphi_t} + e(\varphi_t) R_{\varphi_t}$ so that 
\begin{align}
 R_{\varphi_t}^{1/2} \< p \>^2 R_{\varphi_t}^{1/2}  =  R_{\varphi_t}^{1/2} (1+ p^2 ) R_{\varphi_t}^{1/2}  
 = R_{\varphi_t} (1+e(\varphi_t)) + Q_{\varphi_t} -  R_{\varphi_t}^{1/2} V_{\varphi_t} R_{\varphi_t}^{1/2},
 \end{align}
 which we use to estimate 
 \begin{align}
 \| \< p \> R_{\varphi_t}^{1/2}  \|^2 \lesssim 1 + \| R_{\varphi_t}^{1/2} \| \,  \| \< p \> R_{\varphi_t}^{1/2}  \|  
 \end{align}
 where the last step follows from Lemma \ref{lemma:form:bound:V} \eqref{eq:V:bound:1} as well as $|e(\varphi_t) | \lesssim 1$. This proves the second bound. We proceed similarly in the third bound, namely by writing $
\< p^2\> R_{\varphi_t}  = R_{\varphi_t} (1+e(\varphi_t) )  +Q_{\varphi_t}  - V_{\varphi_t} R_{\varphi_t}  $,
which we now estimate as
\begin{align}
\| \< p^2\> R_{\varphi_t}  \| & 
\lesssim 1   + \| V_{\varphi_t} R_{\varphi_t} \|  \lesssim 1 + \| \< p\>  R_{\varphi_t} \|  \lesssim 1 
\end{align}
where we used again \eqref{eq:V:bound:1}.  For the last bound in (i), we proceed one more time in the same spirit, that is
\begin{align}
\| \psi_{\varphi_t}\|_{H^2}^2 \le 1 + \| p^2 \psi_{\varphi_t} \|^2_{L^2} \le 1 + |e(\varphi_t)| + \| V_{\varphi_t} \psi_{\varphi_t} \|_{L^2}^2 \le 1 + \| \psi_{\varphi_t} \|^2_{H^1}.
\end{align}
That $\| \psi_{\varphi_t}\|_{H^1} \lesssim 1$ follows as in \eqref{Eq:H_1_Bound_r}. This proves the first statement.

Regarding the second statement,  note that continuity and differentiability of $t\mapsto \psi_{\varphi_t} $ follow directly from Lemma  \ref{lem:ground:state:difference} in combination with continuity and differentiability of $t\mapsto \varphi_t$. Moreover, observe that 
$\langle \psi_{\varphi_t},\partial_t\psi_{\varphi_t}\rangle = 2\Re \langle \psi_{\varphi_t},\partial_t\psi_{\varphi_t}\rangle = \partial_t \| \psi_{\varphi_t}\|^2 =0$.

To obtain \eqref{eq:eom:psi:groundstate}, we differentiate the eigenvalue equation $
(h_{\varphi_t}-e(\varphi_t))\psi_{\varphi_t}=0 $ to find
\begin{align}
(\partial_t h_{\varphi_t})\psi_{\varphi_t}+(h_{\varphi_t}-e(\varphi_t))\,\partial_t\psi_{\varphi_t}-\dot e(\varphi_t)\,\psi_{\varphi_t}=0 .
\end{align}
Taking the scalar product with $\psi_{\varphi_t}$ yields the Hellmann--Feynman identity
\begin{align}\label{eq:Hellmann-Feynman}
\dot e(\varphi_t)=\langle \psi_{\varphi_t},(\partial_t h_{\varphi_t})\psi_{\varphi_t}\rangle .
\end{align}
Using $Q_{\varphi_t}\psi_{\varphi_t}=0$, this gives
\begin{align}
Q_{\varphi_t}(h_{\varphi_t}-e(\varphi_t)) Q_{\varphi_t} \partial_t\psi_{\varphi_t}
=-\,Q_{\varphi_t}(\partial_t h_{\varphi_t})\,\psi_{\varphi_t}.
\end{align}
By uniform positivity of the spectral gap, the operator $h_{\varphi_t}-e(\varphi_t)$ is invertible on $\text{Ran}(Q_{\varphi_t})$ with bounded inverse $R_{\varphi_t}$, and thus
\begin{align}
\partial_t\psi_{\varphi_t}=-\,R_{\varphi_t}(\partial_t h_{\varphi_t})\,\psi_{\varphi_t}.
\end{align}
Since $h_{\varphi}=p^2+V_{\varphi}$ depends linearly on $\varphi$, we have $\partial_t h_{\varphi_t}=V_{\partial_t\varphi_t}$. Along the aKG flow, we have $\partial_t\varphi_t=-\,i\omega\varphi_t - i \sigma(\psi_{\varphi_t})$, hence $\partial_t h_{\varphi_t}=-\,V_{i\omega\varphi_t} -V_{ i \sigma(\psi_{\varphi_t})}$. Since $\overline{\sigma(\psi_{\varphi_t})(k)} = \sigma(\psi_{\varphi_t})(-k)$, we find  $V_{ i \sigma(\psi_{\varphi_t})} = 0$, and thus $\partial_t \psi_{\varphi_t} = R_{\varphi_t} V_{i\omega \varphi_t}\psi_{\varphi_t}$. 

For the third statement, note that the bound on $\| \dot \psi_{\varphi_t}\|_{L^2}$ was derived already in the proof of (ii). For the second bound, we compute
\begin{align}
\dot \sigma(\psi_{\varphi_t})(k) =   \langle R_{\varphi_t} V_{i\omega \varphi_t} \psi_{\varphi_t}, \bold G (\cdot, k) \psi_{\varphi_t} \rangle +  \langle \psi_{\varphi_t} , \bold G (\cdot, k) R_{\varphi_t} V_{i\omega \varphi_t}  \psi_{\varphi_t} \rangle.
\end{align}
Both terms are estimated in analogy to the proof of \eqref{eq:sigma:bound:s}, which yields
\begin{align}
|  \langle R_{\varphi_t} V_{i\omega \varphi_t} \psi_{\varphi_t}, \bold G (\cdot, k) \psi_{\varphi_t}\rangle |  \lesssim \| R_{\varphi_t} V_{i\omega \varphi_t} \psi_{\varphi_t} \|_{H^1} \, \| \psi_{\varphi_t}\|_{H^1}  \lesssim 1.
\end{align}
For the last bound, observe that $\dot Q_{\varphi_t} = - \dot P_{\varphi_t} =   P_{\varphi_t} V_{i\omega \varphi_t} R_{\varphi_t} - R_{\varphi_t} V_{i\omega \varphi_t} P_{\varphi_t} $ with $P_{\varphi_t} = | \psi_{\varphi_t} \rangle \langle \psi_{\varphi_t}|$ to compute
\begin{align}
\dot R_{\varphi_t} & =  P_{\varphi_t} V_{i\omega \varphi_t} R_{\varphi_t}^2 - R_{\varphi_t}^2 V_{i\omega \varphi_t} P_{\varphi_t} 
-  R_{\varphi_t}  ( \dot h_{\varphi_t} - \dot e(\varphi_t)) R_{\varphi_t}.
\end{align}
The bound $\| \dot R_{\varphi_t} \| + \| \< p \> \dot R_{\varphi_t} \< p \> \| \le C$ now follows from Lemma \ref{lemma:form:bound:V} \eqref{eq:V:bound:1} in combination with the resolvent bounds from Item (i), as well as the estimate
\begin{align}
\< p\>^{-1}  ( \dot h_{\varphi_t} - \dot e(\varphi_t)) \< p\>^{-1} & = \< p\>^{-1}  \left(  V_{i\omega \varphi_t} - \langle \psi_{\varphi_t} , V_{i\omega \varphi_t } \psi_{\varphi_t} \rangle \right) \< p \>^{-1} \lesssim 1
\end{align}
where the last step is a consequence of the form bound $\pm V_{\varphi} \le C \<p^2\> \| \varphi\|_{\mathfrak h_{1/2}}$ from \eqref{eq:V:bound:1}. 
This completes the proof of the lemma.
\end{proof}

\section{Renormalization of the quantum fluctuations\label{sec:bog:ren}}

In this section, we construct the \emph{renormalized Bogoliubov--Nelson evolution}, which describes quantum field fluctuations around the classical aKG field. Our main goal is to prove Proposition~\ref{prop:Bogo:renormalization}, derive important properties of the unitary propagator $\mathbb U(t)$, and obtain an explicit expression for its quadratic generator. The proof relies on two technical lemmas.

The first lemma provides a representation of the quadratic operator introduced in \eqref{eq:Bog:cutoff},
\begin{align}
\mathbb H_\Lambda (t) = H_f - \< \psi_{\varphi_t} \phi( \bold G_\Lambda ) R_{\varphi_t} \phi( \bold G_\Lambda ) \psi_{\varphi_t} \>,
\end{align}
in terms of a quadratic operator $\mathbb H_{\mathscr D, K,\Lambda}(t)$ with cutoffs $2 \le K\leq\Lambda\le \infty$ and an explicit energy shift. This representation is analogous to \eqref{dressing:identity}, but with the important distinction that no dressing transformation is required to relate $\mathbb H_\Lambda(t)$ and $\mathbb H_{\mathscr D , K ,\Lambda}(t)-4\pi\ln\Lambda$. We nevertheless refer to $\mathbb H_{\mathscr D , K,\Lambda}(t)$ as the {dressed Bogoliubov--Nelson Hamiltonian}, since it corresponds to the Bogoliubov type approximation for the dressed Nelson Hamiltonian $H_{\mathscr D,K}^\ve$. We will show that the operators $\mathbb H_{\mathscr D , K,\Lambda}(t)$ have a limit as $\Lambda\to\infty$.  

The second lemma provides the necessary bounds for $\mathbb H_{\mathscr D,K,\Lambda}(t)$ to prove the existence of the corresponding unitary evolution. These results together constitute the main ingredients to prove Proposition~\ref{prop:Bogo:renormalization}..

For $\Lambda>0$, we define the cutoff version of $\bold B_K$ by
\begin{align}
\bold B_{K,\Lambda}(x,k) & =\bold B_K(x,k) 
 \mathbf 1_{|k|\leq\Lambda} = \frac{e^{-ikx}}{|k|^{3/2}}\,\mathbf 1_{K\leq|k|\leq\Lambda}\,.
\end{align}
The parameter $K$ plays the same role as the cutoff appearing in Lemma~\ref{lem:dressed:Nelson}. Throughout this section, all generators and evolutions are defined relative to $\varphi_t$, a solution to \eqref{eq:aKG} for initial data $\varphi_0 \in \mathfrak h_{1/2}$ with $e(\varphi_t)<0$, as described in the beginning of Section \ref{sec:ren:bog:ev}. 

We then introduce the {dressed Bogoliubov--Nelson Hamiltonian}
\begin{align}
\mathbb H_{\mathscr D,K,\Lambda}(t) 
&= H_f
+ \lsp \psi_{\varphi_t},\,
\big( D_{K,\Lambda} 
+ 2 \Re\langle \sigma(\psi_{\varphi_t}), \bold B_{K,\Lambda}\rangle 
- A_{K,\Lambda} R_{\varphi_t} A_{K,\Lambda} 
\big) 
\psi_{\varphi_t} \rsp_{L^2}\,, \label{eq:def:H:K:Lambda} \\[1.5mm]
A_{K,\Lambda} 
&= \phi(\bold G_K) + 2p\cdot a(k \bold B_{K,\Lambda}) + 2a^*(\bold k B_{K,\Lambda})\cdot p , \label{eq:def:A:Lambda}\\[1.5mm]
D_{K,\Lambda} 
&= 2 a^*(k \bold B_{K,\Lambda})\, a(k \bold B_{K,\Lambda}) 
+ a^*(k \bold B_{K,\Lambda})^2 
+ a(k \bold B_{K,\Lambda})^2 
+ 4\pi \ln K\,. \label{eq:def:D:Lambda}
\end{align}
Observe that $A_{K,\infty}=A_K$ and $D_{K,\infty}=D_K$, as defined earlier in~\eqref{eq:def:A} and \eqref{eq:def:D}.  

The next lemma shows that the form domain of $\mathbb H_{\mathscr D , K,\Lambda}(t)$ always contains the dense subspace $Q(H_f) \subset \mathcal  F$, for all $2 \le K < \Lambda \le \infty $. In particular, setting $\Lambda=\infty$ does not alter the domain. Moreover, it provides the key identity (in statement (ii)) and estimates (in statement (iii)) for the renormalization of the fluctuation dynamics.

\begin{lemma}\label{lem:Bog:1} 
Let $\mathbb H_{\Lambda}(t)$ and $\mathbb H_{\mathscr D ,K,\Lambda}(t)$ be defined by \eqref{eq:Bog:cutoff} and \eqref{eq:def:H:K:Lambda}, respectively, both relative to the aKG solution $\varphi_t$ with initial data $\varphi_0\in \mathfrak h_{1/2}$ satisfying $e(\varphi_0)<0$. There exists $T>0$ such that the following statements hold:
\begin{itemize}
\item[(i)] There exists $C>0$ such that for all $|t|\le T$, in the sense of quadratic forms,
\begin{alignat}{2}
\pm \mathbb H_{\Lambda}(t)  
&\leq C\,\Lambda^2\, H_f  + C 
\qquad &\text{for all }\quad  & 1<\Lambda<\infty,\\[1mm]
\pm \mathbb H_{\mathscr D , K,\Lambda}(t)  
&\leq C\,K^2\, H_f + C\,\ln K 
\qquad &\text{for all } \quad & 2\leq K<\Lambda\leq\infty.
\end{alignat}

\item[(ii)] As quadratic forms on $Q(H_f)$, we have the identity for all $|t|\le T$
\begin{align}\label{eq:Bogoliubov:identity}
\mathbb H_{\mathscr D , K,\Lambda}(t) 
= \mathbb H_\Lambda(t) + 4\pi\ln\Lambda 
\qquad \text{for all } \quad 2\le K<\Lambda<\infty\,.
\end{align}

\item[(iii)] There is $C>0$ such that, for all $2\le K<\Lambda<\infty$, $|t|\le T$ and all normalized $\Psi,\Phi \in Q(H_f)$,
\begin{align}
\big| \lsp \Psi,\,
\big( \mathbb H_{\mathscr D,K,\infty}(t) 
- \mathbb H_{\mathscr D,K,\Lambda}(t) \big)\Phi \rsp \big|
\leq \frac{C\,K^{1/2}}{\Lambda^{1/4}}
\Big( \lsp \Psi, H_f \Psi \rsp 
+ \lsp \Phi,  H_f \Phi \rsp + 1 \Big)\,.
\end{align}
\end{itemize}
\end{lemma}

\noindent We refer to $\mathbb H_{\mathscr D,K}(t):= \mathbb H_{\mathscr D,K,\infty} (t)$ as the \emph{renormalized Bogoliubov--Nelson Hamiltonian}.

 \begin{proof} \underline{Proof of (i):}  The first bound follows from \eqref{eq:bound:phi:R:phi} and $\| R_{\varphi_t} \| \lesssim 1$. For the second bound, we apply Lemma \ref{lem:T:bounds} in combination with $\| \< p \> R_{\varphi_t} \< p \> \| \lesssim 1$, see Lemma \ref{lem:resolvent:bounds}, and $\| \< p \>  P_{\varphi_t} \| \lesssim 1$. Note that Lemma \ref{lem:T:bounds} was stated only for $\Lambda =\infty$, but the proofs generalize trivially to finite values of $\Lambda$.\smallskip

\noindent \underline{Proof of (ii):} For notational simplicity, let us ignore the $t$-dependence and write $h:=h_{\varphi_t} $, $e := e(\varphi_t) $, $R := R_{\varphi_t} $ and $(h-e)R = Q = 1 - P$ with $P=|\psi\rangle \langle \psi|$. In fact, the following argument relies only on these relations together with $(h-e)\psi =0$. Let us also introduce the anti-hermitian field operator $\phi^-(f)  := a^*(f) - a(f)$, and for better readability, let us abbreviate $\bold B := \bold B_{K,\Lambda}$ within this computation. 

We start by noting that $2  a^*(k \bold B ) a( k  \bold B ) + a^*( k \bold B ) a^*( k \bold B ) + a( k \bold B) a( k \bold B ) =  \phi( k \bold B )^2 - \| k \bold B\|^2_{L^2}$ and write, using $\phi(k \bold B) = - [p , \phi^- (\bold B) ]$,
\begin{align*}
 \Big[  [ p^2 , \phi^-(\bold B ) ] , \phi^- (\bold B) \Big]  = 2 ([p , \phi^- (\bold B) ])^2 -   p [ \phi(k \bold B) , \phi^-(\bold B) ]  -  [ \phi (k \bold B) , \phi^-(\bold B) ]  p.
\end{align*}
The last two terms on the right side vanish, since $[ \phi (k \bold B) , \phi^-(\bold B) ]  = 2\Re \langle \bold B, k \bold B \rangle =0$. On the left side we can replace $p^2 $ by $h-e$ in the commutator (since multiplication operators in $x$ commute with $\phi^-(\bold B)$). Because $[p,\phi^-(\bold B)] = \phi(k\bold B)$, this leads to the identity
\begin{align}
\phi(k \bold B )^2 
& = \frac{1}{2}\Big[ \big [ h-e , \phi^- (\bold B) \big ] , \phi^- (\bold B ) \Big] 
\end{align}
and therefore, taking the partial expectation in $ \psi \in L^2(\mathbb R^3_x)$ and using $(h-e)\psi = 0$,
\begin{align}
\lsp \psi,  \phi( k \bold B )^2 \psi \rsp 
& = - \lsp \psi,  \phi^- (\bold B ) (h - e)    \phi^- (\bold B )  \psi \rsp .
\end{align}
On the other hand, we use 
\begin{align}
2 p \cdot a(k \bold B ) + 2 a^*( k \bold B ) \cdot p & = - \big[ h - e  , \phi^- ( \bold B )  \big] - \phi( k^2 \bold B)
\end{align}
together with the identity $
\bold G_K - k^2  \bold B = \bold G_K - k^2  \bold B_{K,\Lambda}   =  \bold G_{\Lambda}  $ and $P (h-e) = 0$, so that 
\begin{align}
P A_{K,\Lambda}  Q & =  P \phi(\bold G_{\Lambda } ) Q  + P \phi^- ( \bold B )  (h-e) Q , \\
Q A_{K,\Lambda}  P & =  Q \phi(\bold G_{\Lambda } ) P -  Q (h-e) \phi^- ( \bold B )  P.
\end{align}
As a result, we have with $R = Q(h-e)^{-1}Q$:
\begin{align}
 \lsp \psi , A_{K,\Lambda} R A_{K,\Lambda}  \psi \rsp 
& =  \langle \psi, \phi(\bold G_{\Lambda } )    R \phi(\bold G_{\Lambda  }   )  \psi \rangle  + \langle \psi,   \phi^-( \bold B )  Q \phi(\bold G_{ \Lambda }   )   \psi \rangle  \notag \\
& \quad  -  \langle \psi, \phi(\bold G_{ \Lambda }  )  Q \phi^- ( \bold B  ) \psi \rangle  -  \langle \psi,    \phi^- ( \bold B )  ( h - e )   \phi^- ( \bold B )  \psi \rangle  .
\end{align}
Next, write the second and third summands as
\begin{align}
&  \langle \psi,   \phi^-( \bold B )  Q \phi(\bold G_{ \Lambda }   )   \psi \rangle    -  \langle \psi, \phi(\bold G_{ \Lambda }  )  Q \phi^- ( \bold B  ) \psi \rangle   \notag\\[1mm]
& \quad = \langle \psi,  [ \phi^-( \bold B) ,  \phi(\bold G_{ \Lambda }   )  ] \psi \rangle  -  \langle \psi,   \phi^-( \bold B)  P \phi(\bold G_{ \Lambda }   )   \psi \rangle   +  \langle \psi, \phi(\bold G_{ \Lambda }  ) P \phi^- ( \bold B  ) \psi \rangle   \notag\\[1mm]
&\quad = -2 \Re \langle \bold B , \bold G_\Lambda \rangle  + 2 \Re \langle \langle \psi , \bold G_\Lambda \psi \rangle  , \langle \psi , \bold B \psi \rangle \rangle \notag\\[1mm]
&\quad = -2 \Re \langle \bold B , \bold G \rangle + 2 \Re \langle \sigma(\psi ) , \langle \psi , \bold B \psi \rangle \rangle  .
\end{align}
Taking all the above into account we arrive at the claimed identity:
\begin{align}
 {\mathbb H}_{\mathscr D,K,\Lambda}(t) & = \mathbb H_\Lambda (t)  +4\pi \ln K   - \| k \bold B \|^2_{L^2}   + 2 \Re \langle \bold B , \bold G \rangle =  \mathbb H_\Lambda (t)  + 4\pi \ln \Lambda,
\end{align} 
where we used $\langle \bold B, 2 \bold G - k^2 \bold B \rangle  = \langle \bold B , \bold G \rangle = \langle \bold B_{K,\Lambda} , \bold G \rangle  = 4\pi \ln \Lambda - 4\pi \ln K$. This completes the proof of part (ii).\medskip

\noindent \underline{Proof of (iii):} Using $ \bold B_K - \bold B_{K,\Lambda} = \bold B_\Lambda$ with $\bold B_\Lambda$ and $\bold B_\Lambda$ defined as in \eqref{eq:Gross:trafo}, we can write
\begin{subequations}
\begin{align}
{ \mathbb H }_{\mathscr  D, K , \infty}(t) - { \mathbb H }_{\mathscr  D, K,\Lambda}(t) & =  \lsp \psi_{\varphi_t},  ( D_{K} - D_{K,\Lambda} + 2 \text{Re} \<  \sigma ( \psi_{\varphi_t})  , \bold B_\Lambda \>   ) \psi_{\varphi_t} \rsp
\label{eq:Bog:difference:Lambda:1} \\[0.5mm]
& \quad + \big( \lsp \psi_{\varphi_t},   ( A_{K,\Lambda}  - A_{K} )  R_{\varphi_t} A_{K }  \psi_{\varphi_t} \rsp + H.c.\big) , \label{eq:Bog:difference:Lambda:2}
\end{align}
\end{subequations}
where the scalar product denotes again the partial expectation in $L^2(\mathbb R^3)$.

To bound the second line, note that $
A_{K} -A_{K,\Lambda}  =  2 a^*( k \bold B_\Lambda ) \cdot p + 2  p\cdot a ( k \bold B_{\Lambda } )$. Thus, we can use Cauchy--Schwarz and estimate via Lemma \ref{lem:T:bounds}: 
\begin{align}
\<  \psi_{\varphi_t}  , A_{K, \Lambda }  R_{\varphi_t} A_{K,\Lambda}  \psi_{\varphi_t} \rsp   \le C K^2 ( H_f +1 )
\end{align}
for all $\Lambda \in [K,\infty]$, as well as via \eqref{eq:bound:a*:R:a} and \eqref{eq:bound:a:R:a*}:
\begin{align}
& \tfrac14 \<  \psi_{ \varphi_t } ,  ( A_{K} -A_{K,\Lambda}  )  R_{\varphi_t} (A_{K} -A_{K,\Lambda} )  \psi_{\varphi_t} \>  \notag\\[1mm]
&\quad  \le \<  \psi_{ \varphi_t } ,  ( a^*( k \bold B_\Lambda ) \cdot p) R_{\varphi_t}  ( p\cdot a ( k \bold B_{\Lambda } ) )  \psi_{\varphi_t} \>  +  \<  \psi_{ \varphi_t } ,   ( p\cdot a ( k \bold B_{\Lambda }  ) )   R_{\varphi_t}  ( a^*( k \bold B_\Lambda )\cdot p) \psi_{\varphi_t} \> \notag\\[1mm]
&\quad  \le C \Lambda^{-1} (H_f+1).
\end{align}
Combining the above estimates, this gives for all $\Psi,\Phi \in D(H_f^{1/2})$
\begin{align}
| \langle \Psi,  \eqref{eq:Bog:difference:Lambda:2} \Phi \rangle | \lesssim C K \Lambda^{-\frac12} ( \< \Psi, \T \Psi \> + \< \Phi, \T \Phi \> + 1 ).
\end{align}
For \eqref{eq:Bog:difference:Lambda:1}, we use again $\bold B_K - \bold B_{K,\Lambda} = \bold B_\Lambda$ to get:
\begin{align}
& D_{K} - D_{K,\Lambda}   = 2  a^*(k \bold B_{K}) a( k  \bold B_{\Lambda})  +  2 a^*(k \bold B_{\Lambda }) a( k  \bold B_{K,\Lambda}  ) \notag \\
& \qquad \quad + \big(  a^*( k \bold B_{K} ) a^*( k \bold B_{\Lambda})  + a^*( k \bold B_{\Lambda } ) a^*( k \bold B_{K,\Lambda})  + H.c. \big) .
\end{align}
It then follows from the bounds in Lemma \ref{lem:operator:bounds:b} that for all $\Psi,\Phi \in D(H_f^{1/2})$
\begin{align}
| \langle \Psi,  ( D_{K} - D_{K,\Lambda} ) \Phi \rangle | \lesssim C (K \Lambda)^{-\frac14} ( \< \Psi, H_f \Psi \> + \< \Phi, H_f \Phi \> + 1 ).
\end{align}
The remaining term in \eqref{eq:Bog:difference:Lambda:1}, is bounded using $\| \sigma ( \psi_{\varphi_t} )  \|_{L^2} \lesssim 1$, see Lemma \ref{lem:sigma:bounds}, and $\| B_\Lambda \|_{L^2} \lesssim \Lambda^{-1}$.  This completes the proof of the lemma.
\end{proof}

To study the unitary evolutions generated by the quadratic operators $ \mathbb H_{\mathscr D, K,\Lambda}( t)$, we need the following estimates.

\begin{lemma} \label{lem:Bog:aux} 
Let $\mathbb H_{\mathscr D ,K,\Lambda}(t)$ be defined by \eqref{eq:def:H:K:Lambda}, relative to the \eqref{eq:aKG} solution $\varphi_t$ with initial data $\varphi_0\in \mathfrak h_{1/2}$ satisfying $e(\varphi_0)<0$. There are constants $T,C,K_0>0$, such that
\begin{subequations}
\begin{align}
\tfrac23 H_f - C  \ln K  (N+1)  & \le {\mathbb H}_{\mathscr D , K,\Lambda}(t)  \le \tfrac32 H_f + C \ln K , \label{lem:LNS:bounds:1}  \\[2mm]
 \pm \partial_t \mathbb H_{\mathscr D , K,\Lambda}(t) & \le C   K^2  \big(  H_f + 1 \big)  ,  \label{lem:LNS:bounds:2} \\[2mm]
\pm i [ \mathbb H_{\mathscr D,K,\Lambda }(t),  N] & \le  CK^2 ( H_f + 1 )   \label{lem:LNS:bounds:3} .
\end{align}
\end{subequations}
for all $|t|\le T $ and $K_0 \le K < \Lambda \le \infty$.
\end{lemma}
\begin{proof} For the upper bound in \eqref{lem:LNS:bounds:1}, use $- A_{K,\Lambda}R_{\varphi_t}A_{K,\Lambda} \le 0$ and apply Lemma \ref{lem:T:bounds} \eqref{eq:D:T:bound} to estimate $\pm D_K \le \frac{1}{4} H_f + C \ln K$. Moreoveer, we estimate $|2 \text{Re}\< \sigma (\psi_{\varphi_t}) , B_{K,\Lambda} \>  | \lesssim \| B_{K}  \|_{L^2} \lesssim K^{-1}$. For the lower bound, we apply the same bound for the part involving $D_{K,\Lambda} + 2 \text{Re}\< \sigma (\psi_{\varphi_t}) , \bold B_{K,\Lambda} \> $ and further estimate  $- A_{K,\Lambda}R_{\varphi_t}A_{K,\Lambda} \gtrsim \ln K (N+1)$ by Lemma \ref{lem:T:bounds} \eqref{eq:ARA:N:bound}. (We note again that Lemma \ref{lem:T:bounds} was stated only for $\Lambda=\infty$, but the estimates extend trivially to all $\Lambda \ge K$.) This proves \eqref{lem:LNS:bounds:1}.

Next, we compute the time derivative 
\begin{align}\label{eq:proof:td:Bog}
\partial_t \mathbb H_{\mathscr D,K,\Lambda}(t) & =  2 \Re \lsp \dot \psi_{\varphi_t} ,  \big( D_{K,\Lambda} +  2 \Re \langle \sigma ( \psi_{\varphi_t ) } , \bold B_{K,\Lambda} \rangle - A_{K,\Lambda} R_{\varphi_t} A_{K,\Lambda}   \big) \psi_{\varphi_t}  \rsp  \notag\\
& \quad + \lsp  \psi_{\varphi_t}, \big(  2 \Re \langle \dot \sigma ( \psi_{\varphi_t} )  ,\bold B_{K,\Lambda} \rangle - A_{K,\Lambda} \dot R_{\varphi_t} A_{K,\Lambda}   \big) \psi_{\varphi_t}  \rsp .
\end{align}
Let us explain the bound for the term involving $A_{K,\Lambda} R_{\varphi_t} A_{K,\Lambda}$. After taking the expectation value in a normalized state $\Phi \in \mathcal F$, one rewrites the inner product as 
\begin{align}
\lsp \Psi , \langle \dot \varphi_{\psi_t} , A_{K,\Lambda} R_{\varphi_t} A_{K,\Lambda} \psi_{\varphi_t} \rsp \Psi \rsp =  \< \dot \psi_{\varphi_t} \otimes \Phi, ( A_{K,\Lambda} R_{\varphi_t} A_{K,\Lambda} ) \psi_{\varphi_t} \otimes \Phi \>
\end{align}
where the innere product on the right side is now in $L^2\otimes \mathcal F$. Then, we use Cauchy--Schwarz and \eqref{eq:ARA:T:bound} to obtain
\begin{align}
& 2 |\< \dot \psi_{\varphi_t} \otimes \Phi, ( A_{K,\Lambda} R_{\varphi_t} A_{K,\Lambda} ) \psi_{\varphi_t} \otimes \Phi \> |  \notag\\[1mm]
&\qquad \le  |\< \dot \psi_{\varphi_t} \otimes \Phi, ( A_{K,\Lambda} R_{\varphi_t} A_{K,\Lambda} ) \dot \psi_{\varphi_t} \otimes \Phi \> |   +  |\<  \psi_{\varphi_t} \otimes \Phi, ( A_{K,\Lambda} R_{\varphi_t} A_{K,\Lambda} ) \psi_{\varphi_t} \otimes \Phi \> |   \notag\\[1mm]
& \qquad \le C K^2 (\| \<  p  \> R_{\varphi_t} \<  p \> \|+ \| R_{\varphi_t} \| )  \left( \| \< p \> \dot \psi_{\varphi_t} \|_{L^2}^2 +  \| \< p \>  \psi_{\varphi_t} \|_{L^2}^2 \right)  \| \< H_f^{1/2} \> \Phi \|^2 \notag\\
& \qquad  \le C  K^2 \| \< H_f^\frac12 \>  \Phi \|^2,
\end{align}
where we used Lemmas \ref{lem:resolvent:bounds} and \ref{lemma:form:bound:V} in the last step, in particular the bound $\| \< p \> \dot \psi_{\varphi_t} \|_{L^2} =  \|  \< p \> R_{\varphi_t} V_{i\omega \varphi_t} \psi_{\varphi_t} \|_{L^2}  \le C$. The other terms in the first line in \eqref{eq:proof:td:Bog} and also the first term in the second line are estimated in the same way. To this end, note that the diagonal bound \eqref{eq:D:T:bound} implies $ \| F^{-1/2} D_K F^{-1/2} \| \lesssim 1 $ in operator norm with $F = K^{-1/2} (1+ H_f) + \ln K$, so that $|\< \Psi , D_K \Phi \> | \lesssim \< \Psi , F \Psi \> + \< \Phi , F \Phi \>$ for all $\Psi,\Phi \in D(H_f^{1/2})$. The remaining term involving $A_{K,\Lambda} \dot R_{\varphi_t} A_{K,\Lambda}$ is estimated via \eqref{eq:ARA:T:bound}  in combination with \eqref{eq:dot:R:bounds}. This proves the bound for the time derivative.

For the commutator, we compute for $\Psi \in \mathcal F$
\begin{align}
\< \Psi , [ { \mathbb H }_{\mathscr  D,\Lambda,K}(t) , N ] \Psi \> & = \< \psi_{\varphi_t} \otimes \Psi ,   [D_{K,\Lambda}   , N  ]  \psi_{\varphi_t} \otimes \Psi \>  \notag\\
& \quad  - 2 \text{Re} \< \psi_{\varphi_t} \otimes  \Psi ,  [A_{K,\Lambda} , N  ] R_{\varphi_t} A_{K,\Lambda}   \psi_{\varphi_t} \otimes   \Psi \>  \label{eq:commutator:N:H}
\end{align}
and then apply the commutator bounds from Lemma \ref{lem:T:bounds}. 

This completes the proof of the lemma.
\end{proof}

The next lemma shows that the unitary evolutions $\mathbb U_{K,\Lambda}(t)$ generated by $ \mathbb H_{\mathscr D , K,\Lambda}(t)$ with initial time $t=0$ exist and that they have a limit $\mathbb U_{K,\Lambda}(t) \xrightarrow{\Lambda \to \infty} \mathbb U_{K,\infty}(t)$. 

\begin{proposition}\label{prop:Bog:ren2} Let $\varphi_t$ denote the solution to \eqref{eq:aKG} with initial data $\varphi_0$ satisfying $e(\varphi_0)<0$, and let $\mathbb H_{\mathscr D,K,\Lambda}(t)$ be defined by \eqref{eq:def:H:K:Lambda} relative to $\varphi_t$. Then, there are $T,K_0>0$ such that the following statements hold:
\begin{itemize}
\item[(i)]
For all $K_0 \le K  < \Lambda \le \infty$ and $\Psi \in D( H_f^{1/2} ) \subset \mathcal F$, there exists a unique solution $\Psi\in C([-T,T],\mathcal F)\cap L^\infty([-T,T],D(H_f^{1/2}))$
to the Cauchy problem
\begin{align}
\begin{cases}
\begin{aligned}
 i \partial_t \Psi(t) & =  \mathbb H_{\mathscr D , K,\Lambda} (t) \Psi(t) \\[1.5mm]
 \Psi(0) & 
= \Psi
\end{aligned}
\end{cases}
\end{align}
with the equation understood in $D((1+H_f)^{-1/2})$. The solution map $\Psi \mapsto \Psi(t)$ defines a unitary $\mathbb U_{K,\Lambda}(t)$ on $\mathcal F$, which is strongly continuous in $t$, and satisfies
\begin{align}\label{eq:T:Bound:Bogoliubov}
\< \mathbb U_{K,\Lambda} (t) \Psi ,H_f \mathbb U_{K,\Lambda} (t) \Psi  \> \le C_K  \<   \Psi , ( H_f + 1 ) \Psi  \>   
\end{align}
for all $\Lambda \in (K,\infty]$, $|t|\le T$, with some $K$-dependent constant $C_K$. 
\item[(ii)] The strong limit 
\begin{align}\label{eq:convergence:Bog:evo}
s-\lim_{\Lambda\to \infty} \mathbb U _{K,\Lambda}(t)  =  \mathbb U _{K,\infty}(t)
\end{align}
exists for all $K\ge K_0$, $|t| \le T$, is unitary, and $t\mapsto \mathbb U _{K,\infty}(t)$ is strongly continuous.
\end{itemize}
\end{proposition}

Before we prove this result, we use it to prove Proposition \ref{prop:Bogo:renormalization}. 

\begin{proof}[Proof of Proposition \ref{prop:Bogo:renormalization}] Observe that via the identity    \eqref{eq:Bogoliubov:identity}, we identity the strongly continuous evolution $\mathbb U_\Lambda(t)$ as
\begin{align}
\mathbb U_\Lambda(t)  = \mathbb U_{K,\Lambda}(t) \exp(it\, 4\pi\, \ln \Lambda).
\end{align}
Moreover, with \eqref{eq:convergence:Bog:evo}, we immediately have $\mathbb U_\Lambda(t) \exp(-it 4\pi  \ln \Lambda) \xrightarrow{ \Lambda \to \infty}  \mathbb U_{K,\infty}(t)  =: \mathbb U(t) $ for any choice of $ K \in [K_0,\infty) $. This proves the statement of Proposition \ref{prop:Bogo:renormalization}.
\end{proof}

\begin{proof}[Proof of Proposition \ref{prop:Bog:ren2}] The existence of the unitary evolution, strong continuity and the bound \eqref{eq:T:Bound:Bogoliubov} follow from the abstract result \cite[Theorem 8]{LNS2015} applied to $H(t) := \mathbb H_{\mathscr D, K,\Lambda}(t)$, $A :=  H_f   + \ln K$ and $B := N+1$. That the assumptions of \cite[Theorem 8]{LNS2015} are satisfied is precisely the statement of the first two bounds from Lemma \ref{lem:Bog:aux}.

To prove the convergence statement, one computes for $\Psi \in D(H_f^{1/2})$:
\begin{align}
& \|( \mathbb U _{K,\Lambda} (t) - \mathbb U_{K,\infty} (t)) \Psi \| ^2 =  2  \Re \int_0^t \langle \mathbb U_{K,\Lambda} (s) \Psi, (\mathbb H_{\mathscr D,K,\Lambda}(s) - \mathbb H_{\mathscr D,K,\infty} (s) )   \mathbb U_{K,\infty}(s) \Psi \rangle \text{d} s   \notag \\
&\qquad  \quad \le \frac{C_K K^\frac12 }{\Lambda^{1/4}} \int_0^t \bigg( \Big\< \mathbb U_{K,\Lambda} (s) \Psi , H_f \mathbb U_{K,\Lambda} (s) \Psi  \Big\> +  \Big\< \mathbb U_{K,\infty} (s) \Psi , H_f \mathbb U_{K,\infty} (s) \Psi  \Big\>  + 1 \bigg)\text{d}s\notag\\[2mm]
& \qquad \quad \le C_K |t| K^\frac12 \Lambda^{-\frac14},
\end{align}
where we applied Lemma \ref{lem:Bog:1} (iii) and \eqref{eq:T:Bound:Bogoliubov}. By density of $ D(H_f^{1/2}) \subseteq \mathcal F$, this proves strong convergence $\mathbb U_{K,\Lambda} (t) \xrightarrow{\Lambda \to \infty}  \mathbb U_{K,\infty}(t)$.
\end{proof}

As a final property of the renormalized Bogoliubov evolution, we need estimates for higher powers of the number operator (and its inverse).

\begin{lemma} \label{lem:Bog:evol:N:bounds} Let $\mathbb U(t)$ be the unitary defined in Proposition \ref{prop:Bogo:renormalization} defined for $t\in [-T,T]$. For every $m\in \mathbb N$, there is a constant $C_m>0$ such that 
for all $|t| \le T $,
\begin{align}
\mathbb U (t)^*  ( N +1 )^m \mathbb U (t) & \le C_m (N+1)^m  \label{eq:Bog:evol:N:bounds},
\end{align}
in the sense of quadratic forms on $\mathcal F$.
\end{lemma}
\begin{proof} As discussed in the proof of Proposition \ref{prop:Bogo:renormalization} above, we have $\mathbb U(t) = \mathbb U_{K,\infty}(t)$ for any $K\in [K_0,\infty)$ for some suitable $K_0>0$. For $m=1$, we can use \eqref{lem:LNS:bounds:3} for some fixed value of $K$ (for instance $K=K_0$) to obtain
\begin{align}
i \partial_t \|  ( N +1 )^{\frac{1}{2}} \mathbb U (t) \Psi \|^2 \le |  \< \mathbb U(t) \Psi , [ \mathbb H_{\mathscr D,K,\infty}(t) ,   N ] \mathbb U (t) \Psi \>| \le C \|  (  N +1 )^{\frac{1}{2}} \mathbb U (t) \Psi \|^2,
\end{align}
and thus via Grönwall's inequality, we obtain the claimed bound. For $m\ge 2$, the bounds are obtained in analogy. 
\end{proof}

\section{Proof of the norm approximation}\label{sec:proof:norm:approximation}

In this section we prove Theorem~\ref{thm:norm:approximation}.  
We fix the initial field \(\varphi_0\in\mathfrak h_{1/2}\) with \(e(\varphi_0)<0\).  
Given \(\varphi_0\), we choose \(T>0\) small enough so that the aKG dynamics and the Bogoliubov--Nelson evolution are well-defined on \([-T,T]\), and so that all previously established estimates hold uniformly for \(|t|\le T\). For notational convenience, we restrict to \(t\in[0,T]\) in the proof; the argument extends verbatim to \(t\in[-T,0]\).

After recalling the precise goal, introducing some notation, and stating a preparatory lemma, the proof proceeds in three steps: a first-order Duhamel expansion in Section~\ref{sec:first-order}, followed by two adiabatic perturbation expansions in Sections~\ref{sec:second-order} and \ref{sec:3rd:order}.

Our aim is to estimate the norm of the difference
\begin{align}
\widetilde \delta(t)
:=\; e^{i\varepsilon^{-2} \mu(t)}\,W(\varepsilon^{-1}\varphi_t)^* \,e^{-it \varepsilon^{-2} H^{\varepsilon}}\,W(\varepsilon^{-1}\varphi_0)\,\psi_{\varphi_0}\otimes\Omega \;-\; \psi_{\varphi_t}\otimes \mathbb U(t)\Omega,
\end{align}
where $e^{-it \varepsilon^{-2} H^\varepsilon}$ and $\mathbb U(t)$ are the renormalized propagators from Propositions~\ref{prop:ren:Nelson} and \ref{prop:Bogo:renormalization}, respectively, and $(\varphi_t,\psi_{\varphi_t})$ is the aKG solution pair from Proposition~\ref{prop:aKG-local} for initial data $\varphi_0\in\mathfrak h_{1/2}$ with $e(\varphi_0)<0$. The phase is
\begin{align}
\mu(t)=\int_0^t\Big(\Im\langle \varphi_s,\partial_s\varphi_s\rangle+\|\omega^{1/2}\varphi_s\|_{L^2}^2+e(\varphi_s)\Big)\, \d s,
\end{align}
and, for convenience, we recall the effective equations
\begin{align}
i \partial_t \mathbb U(t)\Omega &= \mathbb H_{\mathscr D,K,\infty}(t)\,\mathbb U(t)\Omega, \qquad
i \partial_t \varphi_t = \omega\varphi_t + \sigma(\psi_{\varphi_t}), \qquad
\partial_t \psi_{\varphi_t} = R_{\varphi_t} V_{i\omega \varphi_t}\psi_{\varphi_t}.
\end{align}
The generator $\mathbb H_{\mathscr D,K,\infty}(t)$ is defined in~\eqref{eq:def:H:K:Lambda} relative to $\varphi_t$, for $K\ge K_0$ and some $\varepsilon$-independent $K_0>0$. We may choose any $K\ge K_0$. Finally, recall that $(h_{\varphi_t}-e(\varphi_t))\psi_{\varphi_t}=0$.

Instead of estimating $\widetilde \delta(t)$ directly, it is convenient to estimate the modified difference
\begin{align}
\delta(t):=\xi_K(t)\;-\;\psi_{\varphi_t}\otimes\mathbb U(t)\Omega,
\end{align}
with
\begin{align}\label{eq:xi:K:2}
\xi_K(t):=e^{i\varepsilon^{-2}\mu(t)}\,U_K^\varepsilon\,W(\varepsilon^{-1}\varphi_t)^* \,e^{-it\varepsilon^{-2}H^\varepsilon}\,W(\varepsilon^{-1}\varphi_0)\,(U_K^\varepsilon)^*\,\psi_{\varphi_0}\otimes\Omega,
\end{align}
where we inserted the unitary dressing transformation 
\begin{align}
U_K^\varepsilon=\exp\!\Big(\varepsilon\, a^*(\boldsymbol B_K) - \varepsilon\, a(\boldsymbol B_K)\Big),
\qquad
\boldsymbol B_K(x,k)=\frac{\boldsymbol G(x,k)}{k^2}\,\mathbf 1_{\{|k|\ge K\}}.
\end{align}
Using
\begin{align}
\|( U_K^\varepsilon -\mathbf 1)\Psi\|
\;\lesssim\; \varepsilon\,\|B_K\|_{L^2}\,\|(N+1)^{1/2}\Psi\| \quad \text{for all}\quad \Psi \in D(N^{1/2}),
\end{align}
together with $\|B_K\|_{L^2}\lesssim K^{-1}$ and Lemma~\ref{lem:Bog:evol:N:bounds}, one readily verifies
\begin{align}
\|(U_K^\varepsilon)^*\delta(t)-\widetilde\delta(t)\|
&\le \|((U_K^\varepsilon)^*-\mathbf 1)\,\psi_{\varphi_0}\otimes\Omega\|
   + \|((U_K^\varepsilon)^*-\mathbf 1)\,\psi_{\varphi_t}\otimes\mathbb U(t)\Omega\|
\;\lesssim\; \varepsilon K^{-1}
\end{align}
for all $|t|\le T$. Hence $\|\widetilde\delta(t)\| \lesssim \|\delta(t)\|+\varepsilon K^{-1}$. At the end of the proof we shall set $K= \ve^{-4/5}$, so in what follows we focus on estimating $\delta(t)$.
 
Since $U_K^\ve W(\ve^{-1}\varphi_t)^*=W(\ve^{-1}\varphi_t)^* U_K^\ve e^{-ig_{K,\varphi_t}(x)}$ with
\begin{align*}
g_{K,\varphi_t}(x)=2\,\Im\langle\varphi_t,\boldsymbol B_K(x,\cdot)\rangle,
\end{align*}
and $U_K^\ve H^\ve (U_K^\ve)^*=H^\ve_{\mathscr D,K}$ (see Lemma~\ref{lem:dressed:Nelson}), we can write
\begin{align}\label{eq:def:xi(t)}
\xi_K(t)=e^{i\ve^{-2}\mu(t)}\,e^{-ig_{K,\varphi_t}(x)}\,W(\ve^{-1}\varphi_t)^*\,e^{-it\ve^{-2}H^\ve_{\mathscr D,K}}\,
W(\ve^{-1}\varphi_0)\,e^{ig_{K,\varphi_0}(x)}\,\psi_{\varphi_0}\otimes\Omega 
\end{align}
As a final preparation, let us state and prove the following energy propagation bound.
\begin{lemma}\label{lemma:time-dep:energy:bound} Let $\xi_K(t)$ be given by \eqref{eq:def:xi(t)}. There exist $C,T>0$ and $K_0>2$ such that
\begin{align}
\Big\| \Big( p^2 + \ve^2 H_f \Big)^{1/2}  \xi_K(t) \Big\|  \le C K^\frac12 
\end{align}
for all $|t|\in T$, $\ve> 0$ and $K\ge K_0$. 
\end{lemma}
\begin{proof} We estimate
\begin{align}
& \hspace{-0.4cm} \Big\| \Big( p^2 + \ve^2 H_f \Big)^{1/2}  \xi_K(t) \Big\|  \notag\\
& = \Big\| \Big( p^2 + \ve^2 H_f  \Big)^{1/2} e^{ - i g_{K,\varphi_t} (x)} W( \ve^{-1} \varphi_t) ^*   e^{-i t \ve^{-2}  H^{\ve }_{\mathscr D,K}}  W(\ve^{-1} \varphi_0 )  e^{ i  g_{K,\varphi_0} (x)}  \psi_{\varphi_0}  \otimes \Omega \Big\| \notag\\
& \lesssim \Big\| \Big( p^2 + \ve^2H_f + 1  \Big)^{1/2}  e^{-i t \ve^{-2} H^{ \ve }_{\mathscr D,K} } W(\ve^{-1} \varphi_0 )  e^{ i  g_{K,\varphi_0} (x)}  \psi_{\varphi_0}   \otimes \Omega \Big\| 
\end{align}
where we used
\begin{align}\label{eq:commuting:relation}
 e^{i g_{K,\varphi_t}(x)} p^2 e^{-i g_{K,\varphi_t}(x)} & =   p^2  + p \cdot (\nabla g_{K,\varphi_t} )(x) + (\nabla  g_{K,\varphi_t} )(x) \cdot p +  | ( \nabla g_{K,\varphi_t} )(x)|^2 \notag\\[1mm]
 & \le  2 p^2 +  2 | (\nabla g_{K,\varphi_t} )(x) |^2,\notag\\[1mm]
W( \ve^{-1} \varphi_t)  \ve^2 H_f W( \ve^{-1} \varphi_t) ^* & = \ve^2 H_f + \ve \phi(\omega \varphi_t) + \| \varphi_t\|_{\mathfrak h_{1/2}}^2  \le  2 \ve^2 H_f +2 \| \varphi_t\|_{\mathfrak h_{1/2}}^2 ,
\end{align}
and the bounds  
$\| \varphi_t \|_{\mathfrak h_{1/2}}^2 \lesssim 1$ and
\begin{align}
| (\nabla  g_{K,\varphi_t}) (x) | =   2 | \< \varphi_t , k \bold B(x,\cdot) \> | \le \| \varphi_t \|_{\mathfrak h_{1/2}} \| \omega^{1/2} B_K \|_{L^2} \lesssim 1
\end{align}
for $|t| \le T$. Now apply \eqref{eq:energy:bound}, commute the unitary propagator with its generator, and apply \eqref{eq:energy:bound} again:
\begin{align}
\Big\| \Big( p^2 + \ve^2 H_f \Big)^{1/2}  \xi_K(t) \Big\| & \lesssim \Big\| \Big( H^\ve_{\mathscr D , K}  + K \Big)^{1/2}  W(\ve^{-1} \varphi_0 )  e^{ i  g_{K,\varphi_0} (x)}  \psi _{\varphi_0} \otimes \Omega \Big\| \notag\\ 
& \lesssim \Big\| \Big(p^2 + \ve^2 H_f+K \Big) W(\ve^{-1} \varphi_0 )  e^{ i  g_{K,\varphi_0} (x)}  \psi_{\varphi_0}   \otimes \Omega \Big\| \notag\\ 
& \lesssim \Big\| \Big( p^2 + \ve^2 H_f + K  \Big)^{1/2}  \psi_{\varphi_0}   \otimes \Omega \Big\| \notag\\[1mm]
& \lesssim K^\frac12 . \label{eq:lemma7:1:secong:part}
\end{align}
This proves the lemma.
\end{proof}

\subsection{First-order expansion}\label{sec:first-order}

We note the identity
\begin{align}
i\partial_t W(\ve^{-1}\varphi_t)^*=\Big(-\ve^{-1}\phi(i\dot\varphi_t)+\ve^{-2}\Im\langle\varphi_t,\dot\varphi_t\rangle\Big)W(\ve^{-1}\varphi_t)^*
\end{align}
and that $\dot\mu(t)=\Im\langle\varphi_t,\dot\varphi_t\rangle+\|\omega^{1/2}\varphi_t\|_{L^2}^2+e(\varphi_t)$. With this we compute the generator $\mathscr L_{\mathscr D,K}^{\ve}(t)$ of the dressed dynamics defined by $i\partial_t\xi_K(t)=\mathscr L_{\mathscr D,K}^{\ve}(t)\,\xi_K(t)$.

A straightforward computation, using $i\dot\varphi_t=\omega\varphi_t+\sigma(\psi_{\varphi_t})$, yields
\begin{align}
\mathscr L_{\mathscr D,K}^{\ve}(t)
&=e^{-ig_{K,\varphi_t}(x)}\,W(\ve^{-1}\varphi_t)^*\,\ve^{-2}H^{\ve}_{\mathscr D,K}\,W(\ve^{-1}\varphi_t)\,e^{ig_{K,\varphi_t}(x)}\notag\\
&\quad-\ve^{-2}\big(\|\omega^{1/2}\varphi_t\|_{L^2}^2-e(\varphi_t)\big)-\ve^{-1}\phi(\omega\varphi_t+\sigma(\psi_{\varphi_t})) + 2\Re\langle \bold B_K,\omega\varphi_t+\sigma(\psi_{\varphi_t})\rangle\notag\\[1mm]
&=\ve^{-2}\big(h_{\varphi_t}-e(\varphi_t)\big) + \ve^{-1}\big(A_K-\phi(\sigma(\psi_{\varphi_t}))\big)\notag\\
&\quad+\big(H_f+D_K+2\Re\langle \bold B_K,\sigma(\psi_{\varphi_t})\rangle\big) - \ve\,\phi(\omega\bold B_K)+\ve^2\|\omega^{1/2}B_K\|_{L^2}^2,\label{eq:def:fluctuation:generator}
\end{align}
where we used \eqref{eq:mathscr:H} and \eqref{eq:Weyl:dressed:Ham} for $\varphi = \varphi_t$ and, for convenience, recall that
\begin{subequations}\label{def:A:and:D}
\begin{align}
A_K&=\phi(\bold G_K)+2\,a^*(k\bold B_K)\cdot p+2\,p\cdot a(k\bold B_K),\\
D_K&=2\,a^*(k\bold B_K)\,a(k\bold B_K)+a^*(k\bold B_K)^2+a(k\bold B_K)^2+4\pi\ln K,
\end{align}
\end{subequations}

Next we expand $\delta(t) = \xi_K(t)-\zeta(t) $, where we abbreviate $\zeta(t)=\psi_{\varphi_t}\otimes\mathbb U(t)\Omega$. With
\begin{align}\label{eq:recall:eom}
\partial_t \xi_K(t) = -i \mathscr L_{\mathscr D,K}^\ve(t) \xi_K(t), \qquad  \partial_t\zeta(t)= \big(-i \mathbb H_{\mathscr D,K,\infty}(t)+ R_{\varphi_t}V_{i\omega\varphi_t}\big)\zeta(t)
\end{align}
and $\| \xi_K(t) \| = \| \zeta (t) \| = 1$, we compute
\begin{align}
\frac12 \frac{d}{dt} \| \delta(t)\|^2 & = - \frac{d}{dt} \Re \< \xi_K(s) , \zeta(s) \rangle \notag\\
& = \Im \Big\langle \xi_K(s),\big(\mathscr L_{\mathscr D,K}^\ve(s)-\mathbb H_{\mathscr D,K,\infty}(s)-i\,R_{\varphi_s}V_{i\omega\varphi_s}\big)\zeta(s)\Big\rangle ,
\end{align}
where we used that $\mathscr L_{\mathscr D,K}^\ve(s)$ is symmetric. Since $\delta(0)=0$, integration from zero to $t$ gives
\begin{align}
\| \delta(t) \|^2
&=2\Im\int_0^t \d s\,\Big\langle \xi_K(s),\big(\mathscr L_{\mathscr D,K}^\ve(s)-\mathbb H_{\mathscr D,K,\infty}(s)-i\,R_{\varphi_s}V_{i\omega\varphi_s}\big)\zeta(s)\Big\rangle\notag\\
&=2\Im\Big(\mathscr M^{(1)}+\mathscr M^{(2)}+\mathscr R^{(3)}+\mathscr R^{(4)}+\mathscr R^{(5)}\Big).\label{eq:Expansion:M1:M2:R}
\end{align}
Here, we split into two main terms and three remainders:
\begin{subequations}\label{eq:first-order-expansion}
\begin{align}
\mathscr M^{(1)}&=\int_0^t \d s\,\Big\langle \xi_K(s),\,\ve^{-1}\big(A_K-\phi(\sigma_{\psi_{\varphi_s}})\big)\,\zeta(s)\Big\rangle,\\
\mathscr M^{(2)}&=\int_0^t \d s\,\Big\langle \xi_K(s),\,P_{\varphi_s}A_KR_{\varphi_s}A_K\,\zeta(s)\Big\rangle,\\
\mathscr R^{(3)}&=\int_0^t \d s\,\Big\langle \xi_K(s),\,Q_{\varphi_s}\big(D_K+2\Re\langle \bold B_K,\sigma(\psi_{\varphi_s})\rangle\big)\,\zeta(s)\Big\rangle,\label{eq:M:2}\\
\mathscr R^{(4)}&=\int_0^t \d s\,\Big\langle \xi_K(s),\,\big(-i\,R_{\varphi_s}V_{i\omega\varphi_s}\big)\,\zeta(s)\Big\rangle,\label{eq:R:4}\\
\mathscr R^{(5)}&=\int_0^t \d s\,\Big\langle \xi_K(s),\,\big(-\ve\,\phi(\omega\bold B_K)+\ve^2\|\omega^{1/2}B_K\|_{L^2}^2\big)\,\zeta(s)\Big\rangle.\label{eq:R:5xx}
\end{align}
\end{subequations}
We recall the definition of the projections  $P_{\varphi_s}=|\psi_{\varphi_s}\rangle\langle\psi_{\varphi_s}|$ and $Q_{\varphi_s}=\bold 1-P_{\varphi_s}$. We also use repeatedly the identity
\begin{align}
\mathbb H_{\mathscr D,K,\infty}(s)\,\zeta(s)
= P_{\varphi_s}\big(H_f+D_K+2\Re\langle \bold B_K,\sigma(\psi_{\varphi_s})\rangle - A_KR_{\varphi_s}A_K\big)\,\zeta(s).
\end{align}
Before proceeding with the expansion, we bound the error terms from the third and fifth lines.\medskip

\noindent\textbf{Line \eqref{eq:M:2}.} Write $\mathscr R^{(3)}=\mathscr R^{(3.1)}+\mathscr R^{(3.2)}$. For the second term, we estimate
\begin{align}
|\mathscr R^{(3.2)} | & \lesssim \int_0^t \d s \Big| \Big\langle  \xi_K (s)  , Q_{\varphi_s}  \Re \< \bold B_K, \sigma ( \psi_{\varphi_s} ) \> \zeta(s) \Big\rangle   \Big|   \lesssim  \| B_K \|_{L^2} \int_0^t \d s \| \sigma ( \psi_{\varphi_s} ) \|_{L^2} \lesssim K^{-1} 
\end{align}
where we used $\| B_K \|_{L^2} \lesssim K^{-1}$ and $\|\sigma ( \psi_{\varphi_s} ) \|_{L^2} \lesssim 1  $ by Lemma \ref{lem:sigma:bounds}. For the first term, use $Q_{\varphi_s}D_K\zeta(s)=Q_{\varphi_s}(D_K-4\pi\ln K)\zeta(s)$, since $Q_{\varphi_s} \zeta(s) =0$, and then
\begin{align}
|\mathscr R^{(3.1)}|
&\le \int_0^t \d s\,\big|\langle \xi_K(s),\,Q_{\varphi_s}(D_K-4\pi\ln K)\,\zeta(s)\rangle\big|\notag\\
&\lesssim \int_0^t \d s\,\|\langle p\rangle\,\xi_K(s)\|\,\| \langle p\rangle^{-1}(D_K-4\pi\ln K) \< p \>^{-1} \< N^\frac12 \>^{-2} \| \, \| \<p \> \< N^\frac12\>^2 \zeta(s)\|\notag\\[1mm]
& \lesssim K^{-1/4},
\end{align}
where we applied Lemma \ref{lemma:time-dep:energy:bound}, \eqref{eq:bound:D:N:2} and Lemma~\ref{lem:Bog:evol:N:bounds}. Note that the latter implies
\begin{align}\label{eq:zeta:bound}
\|\langle p\rangle\langle N^{\frac12}\rangle^{2}\zeta(s)\|\lesssim \|\psi_{\varphi_s}\|_{H^1}\,\|\langle N^{\frac12}\rangle^{2} \mathbb U(s)\Omega\|\lesssim 1.
\end{align}
Adding both bounds, we get
\begin{align}\label{eq:final:estimate:R:3}
| \mathscr R^{(3)} | \lesssim K^{-1/4}.
\end{align}

\noindent\textbf{Line \eqref{eq:R:5xx}.}  Write $\mathscr R^{(5)}=\mathscr R^{(5.1)}+\mathscr R^{(5.2)}$. For the term proportional to $\ve$, we estimate
\begin{align}
| \mathscr R^{(5.1)} | & \lesssim \ve \int_0^t \d s\, \| \< p \> \xi_K(s) \|\,  \< p \>^{-1}  \phi(\omega \bold B_K) \< N^\frac12 \>^{-1} \<p\>^{-1} \| \, \| \< p \> \< N^\frac12\> \zeta(s) \| \notag\\
& \lesssim \ve K^{-3/4},
\end{align}
where applied \eqref{op3}, \eqref{eq:zeta:bound} and Lemma \ref{lemma:time-dep:energy:bound}.   Using  $\|\omega^{1/2}B_K\|_{L^2}^2\le C K^{-1}$, the term proportional to $\ve^2 $ satisfies $|\mathscr R^{(5.2)}|\lesssim  \ve^2 K^{-1}$.

Since the remaining terms in \eqref{eq:Expansion:M1:M2:R} still contain contributions of order $O(\ve^{-1})$ and $O(1)$, the first-order expansion is not sufficient to obtain the desired result. In the next two sections, we use an adiabatic expansion scheme to expand these terms to higher orders in $\ve$.

\subsection{Second-order expansion}
\label{sec:second-order}

The key ingredient for the higher-order expansion is the following identity
\begin{align}
Q_{\varphi_s} \xi_K (s) &  = R_{\varphi_s} \Big( h_{\varphi_s} - e ( \varphi_s) \Big) \xi_K (s)   \notag\\
& = R_{\varphi_s} \Big( \ve^2 \left(  i\partial_s - \mathscr L_{\mathscr D, K}^\ve (s) \right)  + h_{\varphi_s} -   e ( \varphi_s)  \Big) \xi_K (s)   \notag  \\
& = R_{\varphi_s} \Big(   -  \ve \big( A_K - \phi( \sigma_{\psi_{\varphi_s}} ) \big)   +  \ve^2 \big(  i\partial_s  -  H_f  - D_K -  2\Re \< \bold B_K , \sigma ( \psi_{\varphi_s})  \> \big) \notag\\
&\hspace{5.5cm} + \ve^3 \phi(\omega \bold B_K) - \ve^4 \| \omega^{1/2} B_K\|^2_{L^2}  \Big) \xi_K (s) \label{eq:identity:Qixi},
\end{align}
where we used $(i\partial_s - \mathscr L_{\mathscr D,K}^{\ve} (s) ) \xi_K (s)  =0 $ and \eqref{eq:def:fluctuation:generator}.  

We now systematically replace $Q_{\varphi_s} \xi_K(s)$ via \eqref{eq:identity:Qixi}. The basic idea is that on $\text{Ran}(Q_{\varphi_s})\subset L^2\otimes \mathcal F$, due to the uniform spectral gap of $R_{\varphi_s}$, one can treat the operator $ \ve^2 ( i\partial_s -  \mathscr L_{\mathscr D , K}^{\ve} (s) ) +  ( h_{\varphi_s} - e(\varphi_s))$ as a perturbation of order $O(\ve)$.

To use the above identity in $\mathscr M^{(1)}$, it is crucial to observe that
\begin{align}\label{eq:expval:A}
\mathscr M^{(1)} & = \int_0^t \d s\,\Big\langle \xi_K(s),\,\ve^{-1}\big(A_K-\phi(\sigma_{\psi_{\varphi_s}})\big)\,\zeta(s)\Big\rangle =  \int_0^t \d s \Big\langle  \xi_K(s)  ,  \ve^{-1} Q_{\varphi_s}  A_K  \zeta(s) \Big\rangle
\end{align}
which follows from $\mathbf 1=P_{\varphi_s}  + Q_{\varphi_s} $, $Q_{\varphi_s} \phi( \sigma ( \psi_{\varphi_s} ) ) \zeta(s) =  \phi( \sigma ( \psi_{\varphi_s} ) ) Q_{\varphi_s} \zeta(s) = 0 $ and the identity 
\begin{align}
P_{\varphi_s} ( A_K - \phi(\sigma ( \psi_{\varphi_s} )  ) ) P_{\varphi_s} =0.
\end{align}
To verify the latter, let us compute
\begin{align}
P_{\varphi_s} ( a^*( G_K ) + 2 a^*( k \bold B_K ) \cdot p ) P_{\varphi_s}  & = P_{\varphi_s} \otimes \int \d k\, \< \psi_{\varphi_s}, ( \bold G_K (\cdot, k) +2 k \bold B_K(\cdot,k) \cdot p) \psi_{\varphi_s} \> \, a_k^* \notag\\
& = P_{\varphi_s} \otimes a^*(\sigma ( \psi_{\varphi_s}) ),
\end{align} 
where we used $\< \psi_{\varphi_s}, \bold G_K \psi_{\varphi_s} \> = \sigma ( \psi_{\varphi_s} )  \mathbf 1_{|k|\le K}$ and $2 \< \psi_{\varphi_s}, ( k\bold B_K \cdot p )  \psi_{\varphi_s} \> = \< \psi_{\varphi_s}, G \psi_{\varphi_s} \> \mathbf 1_{|k|\ge K}   = \sigma ( \psi_{\varphi_s} )  \mathbf 1_{|k|\ge K}$ as a consequence of $\psi_{\varphi_s}$ being real-valued and $(p\cdot \overline{k \bold B_K}) = \bold G \mathbf 1_{|k|\ge K}$ with $\bold G$ defined in \eqref{eq:def:G}.

Because of the presence of the \(P_{\varphi_s}\) projection, \eqref{eq:identity:Qixi} is not applicable in \(\mathscr M^{(2)}\). This term is genuinely of order one on the macroscopic time scale. One of the goals of the  expansion is therefore to extract from \(\mathscr M^{(1)}\) a contribution that cancels \(\mathscr M^{(2)}\).

In the remainder of this section, we apply \eqref{eq:identity:Qixi} to the contributions $\mathscr M^{(1)}$ and $\mathscr R^{(3)}$ in \eqref{eq:Expansion:M1:M2:R}. Each step replaces $Q_{\varphi_s}\xi_K(s)$ by the right side of \eqref{eq:identity:Qixi} and uses integration bys part to remove the time-derivative from $\xi_K(s)$. The additional terms generated by this will be shown to be of higher order in $\ve$ (except for the one that will cancel $\mathscr M^{(2)}$).

With \eqref{eq:expval:A}, we can apply \eqref{eq:identity:Qixi} to the first term $\mathscr M^{(1)}$. We split the resulting terms  into $\mathscr M^{(1)} = \mathscr M^{(1.1)} +\mathscr R^{(1.2)} + \mathscr R^{(1.3)} + \mathscr R^{(1.4)}$ with
\begin{subequations}
\begin{align}
\mathscr M^{(1.1)} & =  - \int_0^t \d s  \Big\langle  \xi_K(s)  ,   \big(   A_K  -  \phi(\sigma ( \psi_{\varphi_s} ) ) \big)  R_{\varphi_s} A_K  \zeta(s) \Big\rangle  , \\[1mm]
\mathscr R^{(1.2)} &  = \ve  \int_0^t \d s \Big\langle \Big(  i\partial_s -   H_f \Big)  \xi_K (s)  ,  R_{\varphi_s}   A_K  \zeta(s) \Big\rangle \label{eq:term:R:1:2} , \\[1mm]
\mathscr R^{(1.3)} &  = - \ve  \int_0^t \d s \Big\langle  \xi_K (s)  ,    D_K     R_{\varphi_s}  A_K    \zeta(s) \Big\rangle \label{eq:term:R:1:3} \\[1mm]
\mathscr R^{(1.4)} &  = \int_0^t \d s \Big\langle  \xi_K (s)  ,  \Big( - 2 \ve \Re \< \bold B_K, \sigma ( \psi_{\varphi_s}  )  \> +    \ve^2 \phi(\omega \bold B_K ) - \ve^3 \| \omega^{\frac12 } B_K \|^2_{L^2} \Big)   R_{\varphi_s}   A_K  \zeta(s) \Big\rangle .\label{eq:term:R:1:4}
\end{align}
\end{subequations}
To extract the term of order $O(1)$, we repeat the same steps for $\mathscr M^{(1.1)}$. Inserting $\mathbf 1  = P_{\varphi_s} + Q_{\varphi_s} $ and using $ P_{\varphi_s} \phi( \sigma ( \psi_{\varphi_s} )  ) R_{\varphi_s} = 0$ leads to
\begin{align}\label{eq:cancellation}
 \mathscr M^{(1.1)} = - \mathscr M^{(2)} + \mathscr R^{(1.5)}
\end{align} 
with $\mathscr M^{(2)}$ as in \eqref{eq:M:2} and remainder
\begin{align}\label{eq:extraction:lo:term}
\mathscr R^{(1.5)} = -\int_0^t \d s  \Big\langle  \xi_K (s)  ,   Q_{\varphi_s} \big( A_K - \phi(\sigma( \psi_{\varphi_s} ) ) \big) R_{\varphi_s} A_K  \zeta(s) \Big\rangle  .
\end{align}

The important observation is the cancellation of $\mathscr M^{(2)}$. Before we proceed with this remainder, let us first estimate all other error terms.\medskip

\noindent \textbf{Line \eqref{eq:term:R:1:3}.} After inserting identities $\mathbf 1 = \<p\>\<p\>^{-1}$ and $\mathbf 1 = \< N^\frac12\>^{-1}  \< N^\frac12\>^{1}$, using $[N,R_{\varphi_s}] =0$, and applying the Cauchy--Schwarz inequality, we obtain
\begin{align}
| \mathscr R^{(1.3)}|  & \le \ve \int^t_0  \d s \Big( \| \< p\> \xi_K (s) \| \, |  \<p\>^{-1} D_K \<p\>^{-1}\<N^\frac12 \>^{-2}\| \, \|\<p\>R_{\varphi_s} \< p \>\| \notag\\ 
& \qquad\qquad \times  \, \|  \<N^\frac12 \>^{2} \< p\>^{-1} A_K \<p\>^{-1} \< N^\frac12\>^{-3} \| \, \| \< N^\frac12 \>^3 \< p\> \zeta(s) \| \Big)  \, .
\end{align}
Next we apply $\| \< p \> \xi_K(s) \| \lesssim K^\frac12$, see Lemma \ref{lemma:time-dep:energy:bound}, and $\| \< N^\frac12 \>^3 \< p\> \zeta(s) \|  \lesssim 1 $, see Lemma \ref{lem:Bog:evol:N:bounds}. The bound $\| \<p\> R_{\varphi_s} \< p \>\| \lesssim 1 $ follows from Lemma \ref{lem:resolvent:bounds}. The remaining two operator norms are estimated in Lemma \ref{lem:N:bounds:2}. Putting everything together yields the bound
\begin{align}\label{eq:final:estimate:R:13}
| \mathscr R^{(1.3)}|  & \lesssim  \ve  K^\frac12 ( \ln K)^\frac32 .
\end{align}
\textbf{Line \eqref{eq:term:R:1:4}.} Using in addition \eqref{op3} and $ \|  \< p \> R_{\varphi_s} \< p \> \| \lesssim 1$, the fourth term is estimated similarly as
\begin{align}\label{eq:final:estimate:R:14}
|\mathscr R^{(1.4)}|  \lesssim  (\ln K)^\frac12 \Big(  \ve K^{-1}  +  \ve^2 K^{-\frac12 } \Big) .
\end{align}

\noindent\textbf{Line \eqref{eq:term:R:1:2}.} Here, we need to integrate by parts. Recalling \eqref{eq:recall:eom}, one finds $\mathscr R^{(1.2)} =  \mathscr R^{(1.21)}  + \mathscr R^{(1.22)} + \mathscr R^{(1.23)} +  \mathscr R^{(1.24)} $ with
\begin{subequations}
\begin{align}
\mathscr R^{(1.21)} & =   \Big\langle i \xi_K(s)  ,  \ve R_{\varphi_s} A_K   \zeta(s) \Big\rangle \Big|_{s=0}^{s=t}  \\[1mm]
\mathscr R^{(1.22)} &  = - \ve \int_0^t \d s \Big\langle  \xi_K(s)  ,     R_{\varphi_s}  \Big[H_f ,   A_K   \Big]  \zeta(s) \Big\rangle \label{eq:P:I:commutator} \\[1mm]
\mathscr R^{(1.23)} &  = \ve  \int_0^t \d s \Big\langle  \xi_K (s)  ,     R_{\varphi_s}   A_K P_{\varphi_s}  \Big(  A_K R_{\varphi_s} A_K -  D_K - 2\Re \< \bold B_K , \sigma ( \psi_{\varphi_s} )  \>   \Big)  \zeta(s) \Big\rangle \\[1mm]
\mathscr R^{(1.24)} &  = i  \ve \int_0^t \d s \Big\langle \xi_K (s)  , \Big(    \dot R_{\varphi_s}  A_K - i R_{\varphi_s} A_K  R_{\varphi_s} V_{i\omega \varphi_s}  \Big)  \zeta(s) \Big\rangle.
\end{align}
\end{subequations}
All lines except the second line are estimated using the same arguments as before. In the last line, one also invokes the bounds $\| \<p\> \dot R_{\varphi_s} \< p \> \|\lesssim 1$ (Lemma  \ref{lem:resolvent:bounds}) and $\| \< p\>^{-1} V_{i\omega \varphi_s } \< p \> ^{-1} \|  \lesssim \| \varphi_s \|_{\mathfrak h_{1/2}} \| \psi_{\varphi_s} \|_{H^1}^2 \lesssim 1 $ (Lemma \ref{lemma:form:bound:V}). The result is
\begin{align}
| \mathscr R^{(1.21)}| + | \mathscr R^{(1.23)}| + | \mathscr R^{(1.24)} | \lesssim  \ve  (\ln K)^\frac32.
\end{align}
The commutator term needs to be regularized with $\< p^2 \>^{-1}$, i.e., we estimate
\begin{align}
| \mathscr R^{(1.22)} |&  \le \ve \int_0^t \d s    \| R_{\varphi_s} \< p^2\> \| \|  \<p\>^{-2}  [H_f  ,   A_K  ]\<p\>^{-2}   \< N^\frac12 \>^{-1} \| \| \< p^2 \> \< N^\frac12\> ^2 \zeta(s) \|  \notag\\
& \lesssim \ve K
\end{align}
where we used $\|  R_{\varphi_s} \<p^2 \> \| + \|  \psi_{\varphi_s} \|_{H^2} \lesssim 1$, see Lemma \ref{lem:resolvent:bounds}, in combination with the bound \eqref{eq:bound:A:N:3} and Lemma \ref{lem:Bog:evol:N:bounds}. Thus, we have
\begin{align}\label{eq:final:estimate:R:12}
| \mathscr R^{(1.2)} | \lesssim  \ve K.
\end{align}

\noindent \textbf{Line \eqref{eq:R:4}.} In this term, we need to expand again via \eqref{eq:identity:Qixi}. To do this, we first introduce a cut-off $L>0$ by writing $\varphi_s = \varphi_{s,L}+\varphi_{s,>L}$ with $\varphi_{s,L}  := \mathbf 1_{|k|\le L} \varphi_s$ and $V_{ i\omega \varphi_s} = V_{ i\omega \varphi_{s,L}} +  V_{ i\omega \varphi_{s,>L}}$. The large-momentum part can be estimated immediately as
\begin{align}\label{eq:R:4:bound:co}
& \bigg| \int_0^t \d s\,\Big\langle \xi_K(s),\,\big(-i\,R_{\varphi_s}V_{i\omega\varphi_{s,>L} }\big)\,\zeta(s)\Big\rangle \bigg| \notag\\
& \qquad\qquad \le  \int_0^t \d s \, \| R_{\varphi_s} \|\, \|   V_{i\omega \varphi_{s,>L}} \< p\>^{-2} \| \, \| \< p \>^2 \zeta(s) \| \lesssim L^{-1/4}
\end{align}
where we used \eqref{eq:V:bound:2} to bound 
\begin{align}
\|  V_{i\omega \varphi_{s,>L}} \< p\>^{-2} \| & \lesssim  \| \omega \varphi_{s,>L} \|_{\mathfrak h_{-3/4}} \lesssim L^{-1/4} \| \varphi_s \|_{\mathfrak h_{1/2}} \lesssim L^{-1/4},
\end{align}
For the low-momentum part, we apply \eqref{eq:identity:Qixi}, which gives $\mathscr R^{(4)} = \mathscr R^{(4.1)} + \mathscr R^{(4.2)} + \mathscr R^{(4.3)} + \mathscr R^{(4.4)} + \mathcal O(L^{-1/4} ) $ with 
\begin{align}
\mathscr R^{(4.1)} & =  -i \ve  \int_0^t \d s \Big\langle  \xi_K(s)  ,   \big( - A_K + \phi(\sigma(\psi_{\varphi_s})) \big)  R_{\varphi_s}^2  V_{i \omega \varphi_{s,L}}  \zeta(s) \Big\rangle \\
\mathscr R^{(4.2)} & =  -i \ve^2  \int_0^t \d s \Big\langle  \Big(  i\partial_s - H_f \Big) \xi_K(s)  ,  R_{\varphi_s}^2  V_{i \omega \varphi_{s,L}}  \zeta(s) \Big\rangle \\
\mathscr R^{(4.3)} & = i \ve^2  \int_0^t \d s \Big\langle  \xi_K(s)  ,  \Big(  D_K + 2 \Re \< \bold B_K, \sigma ( \psi_{\varphi_s} )   \> \Big)  R^2_{\varphi_s}  V_{i \omega \varphi_{s,L} }  \zeta(s) \Big\rangle \\
\mathscr R^{(4.4)} & =   i \int_0^t \d s \Big\langle  \xi_K(s)  , \Big(    -   \ve^3 \phi(\omega \bold B_K ) + \ve^4 \| \omega^{\frac12} B_K \|^2_2 \Big)   R^2_{\varphi_s}  V_{i\omega \varphi_{s,L}}  \zeta(s) \Big\rangle  .
\end{align}
All terms except the second are estimated using the same arguments as before, which yields
\begin{align}
|\mathscr R^{(4.1)} | + |\mathscr R^{(4.3)} | + |\mathscr R^{(4.4)} |  \lesssim  \ve  \ln K  .
\end{align}
Note that the cutoff in $\varphi_{s,L}$ is not relevant for these terms, since we can estimate the potential as $\| \< p \>^{-1} V_{i\omega \varphi_{s,L}} \< p\>^{-1} \| \lesssim \| \varphi_s \|_{\mathfrak h_{1/2}} \lesssim 1$ by \ref{eq:V:bound:3}.
  
In the second line, we proceed with integration by parts, so that $\mathscr R^{(4.2)}  = \mathscr R^{(4.21)}  + \mathscr R^{(4.22)} + \mathscr R^{(4.23)}  + \mathscr R^{(4.24)}$ with
\begin{align}
\mathscr R^{(4.21)} & =        \ve^2  \Big\langle  \xi_K(s)  ,  R_{\varphi_s}^2  V_{i \omega \varphi_{s,K} } \zeta(s) \Big\rangle\Big|_{s=0}^{s=t} \notag\\
\mathscr R^{(4.22)} & =  -  i \ve^2  \int_0^t \d s \Big\langle  \xi_K(s)  ,  R_{\varphi_s}^2  V_{i \omega \varphi_{s,L}  } P_{\varphi_s} \Big( A_K R_{\varphi_s} A_K - D_K - 2 \Re\< \bold B_K , \sigma(\psi_{\varphi_s} ) \> \Big)  \zeta(s) \Big\rangle  \notag\\
\mathscr R^{(4.23)} & =   i \ve^2  \int_0^t \d s \Big\langle  \xi_K(s)  ,  \Big( 2 R_{\varphi_{s}} \dot R_{\varphi_{s}} V_{i \omega \varphi_{s,L}}    + R_{\varphi_s}^2  \dot V_{i \omega \varphi_{s,L}} \Big)  \zeta(s) \Big\rangle \notag\\
\mathscr R^{(4.24)} & =   i \ve^2  \int_0^t \d s \Big\langle  \xi_K(s)  ,  \Big(  R_{\varphi_s}^2 V_{i\omega \varphi_{s,L}} R_{\varphi_s} V_{i\omega \varphi_{s} }   \Big)  \zeta(s) \Big\rangle .
\end{align} 
All terms except for the one involving $\dot V_{i\omega \varphi_{s,L}}$ are estimated again as previous terms (the cut-off $L$ is not relevant for those terms), so we omit the details. The remaining term was the reason for introducing the momentum cut-off. We compute $\dot V_{i\omega \varphi_{s,L}} =  V_{i\omega \dot \varphi_{s,L}} =  V_{\omega^2 \varphi_{s,L}} + V_{ \omega \sigma_L}$ with $\sigma_L := \bold 1 _{|k|\le L}  \sigma(\psi_{\varphi_s})$. Since $ \omega \sigma(\psi_{\varphi_s}) =  \omega^{1/2} \widehat { |\psi_{\varphi_s}|^2}$,  we can use \eqref{eq:V:bound:3} to bound
\begin{align}
\| \< p \>^{-1} V_{\omega \sigma_L  }  \< p \>^{-1} \| & \lesssim \| \omega \sigma(\psi_{\varphi_s}) \|_{\mathfrak h_{-1/2}} = \| \psi_{\varphi_s}\|_{L^4}^2 \lesssim 1.
\end{align}
Invoking also \eqref{eq:V:bound:3}, we find that
\begin{align}
\bigg|  \ve^2  \int_0^t \d s \Big\langle  \xi_K(s)  ,    R_{\varphi_s}^2  \dot V_{i \omega \varphi_{s,L}}   \zeta(s) \Big\rangle  \bigg| &  \lesssim  \ve^2 \sup_{|s| \le T}  \Big( \| V_{\omega^2 \varphi_{s,L} } \< p \>^{-2} \| +\| \< p \>^{-1} V_{ \omega  \sigma_L} \< p \>^{-1} \|  \Big)  \notag\\
&  \lesssim  \ve^2 \sup_{|s| \le T}  \Big(  \| \omega^2 \varphi_{s,L}  \|_{\mathfrak h_{-3/4}}  + 1 \big)  \lesssim \ve^2 L^{3/2} .
\end{align}
Adding everything up, and choosing $L=\ve^{-8/7}$ gives
\begin{align}\label{eq:final:estimate:R:4}
| \mathscr R^{(4)} | \lesssim  \ve  \ln K   + \ve^{2/7}.
\end{align}
 
\noindent \textbf{Line \eqref{eq:extraction:lo:term}.} To continue with $ \mathscr R^{(1.5)}$, we insert $A_K = \phi(\bold G_K) + (A_K - \phi(\bold G_K))$, so that 
$\mathscr R^{(1.5)} =   \mathscr R^{(1.5a )} + \mathscr R^{(1.5b)} + \mathscr R^{(1.5c)}$ with
\begin{subequations}
\begin{align}
  \mathscr R^{(1.5a )}  & =  -  \int_0^t \d s  \Big\langle  \xi_K(s)  ,   Q_{\varphi_s}  \phi(\bold G_K - \sigma ( \psi_{\varphi_s})  ) R_{\varphi_s} \phi( \bold G_K) \zeta(s) \Big \rangle \label{eq:R:1:5a}\\
     \mathscr R^{(1.5b )} & =   \int_0^t \d s  \Big\langle  \xi_K(s)  ,   Q_{\varphi_s}  \phi (\bold G_K -  \sigma ( \psi_{\varphi_s} )  ) R_{\varphi_s} (A_K-\phi(\bold G_K))  \zeta(s) \Big \rangle \\
    \mathscr R^{(1.5c )} & =  - \int_0^t \d s  \Big\langle  \xi_K(s)  ,   Q_{\varphi_s} ( A_K - \phi(\bold G_K )  ) R_{\varphi_s} A_K  \zeta(s) \Big \rangle .
\end{align}
\end{subequations}
For $\mathscr R^{(1.5b)}$ we use \eqref{op1} and Lemmas Lemmas \ref{lem:N:bounds:2}, \ref{lemma:time-dep:energy:bound} and \ref{lem:Bog:evol:N:bounds} to estimate
\begin{align}
|  \mathscr R^{(1.5b )}  | & \lesssim  (\ln K)^\frac12 \int_0^t \d s \,  \| \< p \>  \xi_K (s) \| \,  \|   \< p \>^{-1}  \phi(\bold G_K - \sigma (\psi_{\varphi_s} ) ) \< p \>^{-1} \< N^\frac12 \>^{-1} \|  \,  \| \< p \> R_{\varphi_s} \< p \> \|   \notag\\
& \qquad \quad \times    \|\<  N^\frac12 \> \< p \>^{-1} (A_K - \phi(\bold G_K) ) \< p \>^{-1} \< N^\frac12 \>^{-2} \|\,   \| \< N^\frac12\>^2 \< p \> \zeta(s) \| \notag\\[1mm]
& \lesssim (\ln K)^\frac12 K^{-\frac14}.
\end{align}
In close analogy, one verifies that $|\mathscr R^{(1.5c)} | \lesssim  (\ln K)^\frac12 K^{-\frac14}$.

In \eqref{eq:R:1:5a}, there is no additional decay in $K$. Here we need to use the presence of the $Q_{\varphi_s}$ projection in order to expand one more time via \eqref{eq:identity:Qixi}. \medskip
 
\subsection{Third-order expansion}\label{sec:3rd:order}

It remains to estimate the contribution $\mathscr R^{(1.5a)}$ defined in \eqref{eq:R:1:5a}. To this end, we perform a third expansion using \eqref{eq:identity:Qixi}. A direct application of \eqref{eq:identity:Qixi} to \eqref{eq:R:1:5a} would, after integration by parts, produce the term $\phi (\omega \sigma ( \psi_{\varphi_s} )) $ through the commutator $[H_f, \phi(\sigma ( \psi_{\varphi_s} ))  ]$, which we cannot control in terms of the number operator. To avoid this term, we first isolate and estimate the large-momentum tail of $\sigma( \psi_{\varphi_s} )$ and then apply the expansion only to the remaining low-momentum part. For this purpose we use the following elementary lemma.\medskip

\begin{lemma}\label{lem:sigma:tail} For every $L^2$-normalized $\psi \in H^2(\mathbb R^3) $, and alll $L>1$, we have
\begin{align}
\| \sigma ( \psi )  \mathbf 1_{|k|\ge L} \|_{L^2} & \lesssim  L^{-1} \| \psi \|_{H^2} \\[0.5mm]
\| \omega \sigma (\psi)  \mathbf 1_{|k|\le L} \|_{L^2} & \lesssim (\ln L)^\frac12  \| \psi \|_{H^2} 
\end{align}
\end{lemma}
\begin{proof} Writing $e^{-ikx} = |k|^{-4} [ (k\cdot  p) , [ ( k\cdot  p) ,e^{ - i kx} ] ] $, we can estimate 
\begin{align}
\| \sigma (\psi )  \mathbf 1_{|k|\ge L}\|_{L^2}^2 & = \int\limits_{|k|\ge L} \d k \, \omega^{-1}(k) |\< \psi, e^{-ik(\cdot)} \psi \>|^2 \lesssim  \| \psi \|_{H^2} \int\limits_{|k|\ge L} \frac{\d k}{|k|^{5}}  .
\end{align}
For the second one, use  $\| \omega \sigma ( \psi ) \mathbf 1_{|k|\le L} \|_{L^2} \lesssim 1 + \| \omega \sigma ( \psi )   \mathbf 1_{2\le|k|\le K} \|_{L^2}$ and estimate again with the above identity,
\begin{align}
\| \omega \sigma ( \psi )   \mathbf 1_{2\le|k|\le L} \|_{L^2}^2 & = \int\limits_{2\le |k|\le K} \d k\, \omega(k) |\< \psi, e^{-ik(\cdot)} \psi \>|^2 \lesssim   \| \psi \|_{H^2} \int\limits_{2\le |k|\le L} \frac{\d k}{|k|^{3}} .
\end{align}
This proves the lemma.
\end{proof}

Coming back to the estimate for \eqref{eq:R:1:5a}, let us abbreviate $\sigma_K := \sigma  ( \psi_{\varphi_s}) \bold 1_{|k|\le K}  $ and $\sigma_{K} :=  \sigma ( \psi_{\varphi_s}) \bold 1_{|k|> K}  $. We split  $
  \mathscr R^{(1.5a )}  =   \mathscr R^{(1.5a1)} +  \mathscr R^{(1.5a2)} $ into
  \begin{align}
   \mathscr R^{(1.5a1 )}  & =  -  \int_0^t \d s  \Big\langle  \xi_K(s)  ,   Q_{\varphi_s}  \phi(\bold G_K - \sigma_K ) R_{\varphi_s} \phi( \bold G_K) \zeta(s) \Big \rangle,  \label{eq:R:15a1} \\
   \mathscr R^{(1.5a2 )}  & =   \int_0^t \d s  \Big\langle  \xi_K(s)  ,   Q_{\varphi_s}  \phi(  \sigma_{> K} ) R_{\varphi_s} \phi( \bold G_K) \zeta(s) \Big \rangle
   \end{align}
 The large-momentum part is bounded by
\begin{align}
|  \mathscr R^{(1.5a2 )}  | & \lesssim   \int_0^t \d s \Big( \| \sigma_{> K} \|_{L^2}  \| R_{\varphi_s} \< p \> \|  \notag\\
&\qquad  \times  \|  \< p \>^{-1} \phi(\bold G_K)  \< p \>^{-1} \< N^\frac12 \>^{-1} \|\, \| \< p \> \< N^\frac12 \> \zeta(s) \| \Big) \notag\\
&  \lesssim  K^{-1} (\ln K)^{1/2} \label{eq:large:tail:R5}
\end{align}
where we applied \eqref{op1} together with Lemmas \ref{lem:resolvent:bounds}, \ref{lem:Bog:evol:N:bounds} and \ref{lem:sigma:tail}.\medskip

\noindent \textbf{Line \eqref{eq:R:15a1}.} For shorter notation, let us further abbreviate $L_K :=  \phi(\bold G_K - \sigma_K  ) R_{\varphi_s} \phi(\bold G_K) $. Applying identity \eqref{eq:identity:Qixi}, we obtain
$\mathscr R^{(1.5a1)} =   \mathscr R^{(6.1)} + \mathscr R^{(6.2)}  + \mathscr R^{(6.3)} + \mathscr R^{(6.4)} $ with
\begin{align}
\mathscr R^{(6.1)} & =  \int_0^t \d s  \Big\langle \xi_K(s)  ,    \ve  \big(  A_K  - \phi(\sigma_{\psi_s} ) \big)   R_{\varphi_s} L_K  \zeta(s) 
\Big\rangle , \\[1mm]
\mathscr R^{(6.2)} & = -\ve^2  \int_0^t \d s  \Big\langle \Big( i\partial_s -  H_f \Big) \xi_K (s)  ,  R_{\varphi_s} L_K  \zeta(s) \Big\rangle ,  \label{eq:secon:line}\\[1mm]
\mathscr R^{(6.3)} & = \int_0^t \d s  \Big\langle \xi_K(s)  ,    \ve^2  \Big( D_K  + 2 \ve^2 \Re \< \bold B_K , \sigma( \psi_{\varphi_s}) \> \Big) R_{\varphi_s} L_K  \zeta(s) \Big\rangle ,   \\[1mm]
\mathscr R^{(6.4)}  & = \int_0^t \d s  \Big\langle \xi_K (s)  ,   \Big(   - \ve^3 \phi(\omega \bold B_K)  +  \ve^4 \| \omega^{1/2} B_K \|^2_{L^2}  \Big)   R_{\varphi_s} L_K \zeta(s) \Big\rangle .   
\end{align}
All terms except the second one line can be estimated in exact analogy to previous terms. The resulting estimate is
\begin{align}
|\mathscr R^{(6.1)}  | + | \mathscr R^{(6.3)} | + | \mathscr R^{(6.4)}  |  \lesssim  \ve (\ln K)^{\frac32} + \ve^2 (\ln K)^2 \label{eq:estimates:R:61}
\end{align}
In \eqref{eq:secon:line}, we integrate by parts to find $
\mathscr R^{(6.2)} = 
\mathscr R^{(6.21)} + 
\mathscr R^{(6.22)} + 
\mathscr R^{(6.23)} + 
\mathscr R^{(6.24)}$ with
\begin{align}
\mathscr R^{(6.21)} & =  - i    \ve^2  \Big\langle \xi_K (s)  ,  R_{\varphi_s} L_K  \zeta(s) \Big\rangle\Big|_0^t \\[1mm]
\mathscr R^{(6.22)} &  =  -  \ve^2 \int_0^t \d s  \Big\langle \xi_K (s)  ,   R_{\varphi_s} \Big[ H_f , L_K \Big] \zeta(s) \Big\rangle  \\[1mm]
\mathscr R^{(6.23)} &  =   \ve^2 \int_0^t \d s  \Big\langle  \xi_K (s)  ,   R_{\varphi_s} L_K P_{\varphi_s}  \Big(  A_K R_{\varphi_s} A_K -  D_K - 2\Re \< \bold B_K , \sigma ( \psi_{\varphi_s} )  \>   \Big)  \zeta(s) \Big\rangle ,  \\[1mm] 
\mathscr R^{(6.24)} & = -i  \ve^2  \int_0^t \d s  \Big\langle   \xi_K (s)  ,  \Big( \dot R_{\varphi_s} L_K + R_{\varphi_s} \dot L_K \big)   \zeta(s) \Big\rangle .
\end{align}
These terms can be bounded again in the same way as previous error terms. In particular, for the commutator, we have
\begin{align}
[H_f , L_K  ] & =  \phi^-(\omega \bold  G_K - \omega \sigma_K  ) R_{\varphi_s} \phi(\bold G_K)  +  \phi(\bold G_K - \sigma_ K ) R_{\varphi_s} \phi^-(\omega \bold G_K) 
\end{align}
where we recall that $\phi^-(\bold F) = a^*(\bold F) - a(\bold F)$. Using \eqref{op1}, \eqref{op2} and Lemmas \ref{lem:sigma:bounds}, \ref{lem:resolvent:bounds} and \ref{lem:sigma:tail}, one readily verifies for $m\in \mathbb N_0$
\begin{align}
\| \< p \>^{-1} \< N^\frac12\>^{m} \phi^-( \bold G_K - \sigma_K )  \< N^\frac12\>^{-m-1} \< p \>^{-1} \| & \lesssim  (\ln K)^\frac12 + 1 ,  \\
\| \< p \>^{-1} \< N^\frac12\>^{m} \phi^-(\omega \bold G_K - \omega \sigma_K )  \< N^\frac12\>^{-m-1} \< p \>^{-1} \| & \lesssim  K  + \ln K \, .
\end{align}
In total, we arrive  at 
\begin{align}
| \mathscr R^{(6)} | \le \ve^2  K \ln K .
\end{align}
Combining this with \eqref{eq:large:tail:R5} and  \eqref{eq:estimates:R:61}, as well as with the estimates  for $\mathscr R^{(1.5b)}$ and $\mathscr R^{(1.5c)}$, we arrive at
\begin{align}\label{eq:final:estimate:R:15}
 |\mathscr R^{(1.5)} | \lesssim   \ve^2  K \ln K + (\ln K)^\frac12 K^{-1/4}.
\end{align}

\subsection{Conclusion of the proof}

We now gather all estimates, starting from the expansion of $\delta(t)$ in \eqref{eq:Expansion:M1:M2:R}. 
Recalling the cancellation in \eqref{eq:cancellation} and combining the error bounds 
\eqref{eq:final:estimate:R:3}, \eqref{eq:final:estimate:R:13}, 
\eqref{eq:final:estimate:R:14}, \eqref{eq:final:estimate:R:12}, 
\eqref{eq:final:estimate:R:4}, and \eqref{eq:final:estimate:R:15}, we obtain
\begin{align}\label{eq:conclusion}
\| \delta(t) \|^2 \;\lesssim\; (\ln K)^{1/2} K^{-1/4} + \ve K + \ve^2 K \ln K + \ve^{2/7}. 
\end{align}
Optimizing by choosing $K = \ve^{-4/5}$ proves Theorem~\ref{thm:norm:approximation}.  

One could in principle expand all terms containing a $Q_{\varphi_s}$ projection even further, which would lead to a sharper error bound. 
For the sake of clarity and conciseness, we do not pursue this refinement here.

\section{Proof of convergence of reduced densities}
\label{sec:proof:reduced:densities}

We start with the proof of  Corollary~\ref{cor:particle}, which follows directly from Theorem \ref{thm:norm:approximation}.

\begin{proof}[Proof of Corollary \ref{cor:particle}] Let us abbreviate $\Psi_t^\ve = e^{- i  t\ve^{-2} H^\ve  + i \ve^{-2} \mu(t) )} \Psi_0$ with initial state $\Psi_0 =  \psi_{\varphi_0} \otimes  W(\ve^{-1} \varphi_0 ) \Omega  $, and  $ \Phi_t =  \psi_{\varphi_t} \otimes  W(\ve^{-1} \varphi_t ) \mathbb U(t) \Omega  $,  and  observe that $ \gamma_{\rm particle}^\ve(t) = \text{Tr}_{\mathcal F} | \Psi_t^\ve  \> \<  \Psi_t^\ve  | $ and $ |\psi_{\varphi_t} \> \< \psi_{\varphi_t} | = \text{Tr}_{\mathcal F} | \Phi_t  \> \<  \Phi_t  |  $. We then write
\begin{align}
\gamma_{\rm particle}^\ve(t) - |\psi_{\varphi_t} \rangle \langle \psi_{\varphi_t}  |  & = \text{Tr}_{\mathcal F} |  \Psi_t^\ve   \> \<     \Psi_t^\ve   | - \text{Tr}_{\mathcal F} | \Phi_t   \> \<\Phi_t | \notag\\
& \quad = \text{Tr}_{\mathcal  F} | \Psi_t^\ve  \> \<\Psi_t^\ve   -  \Phi_t | +  \text{Tr}_{\mathcal  F} | \Psi_t^\ve  -  \Phi_t \> \<  \Phi_t  | 
\end{align}
to bound $ \text{Tr}_{L^2} | 
\gamma_{\rm el}^\ve(t) - |\psi_t \rangle \langle \psi_t | | \le 2 \| \Psi_t^\ve - \Phi_t    \| \le C \ve^{1/10} \ln \tfrac{1}{\ve}  $, by Theorem \ref{thm:norm:approximation}.
\end{proof}

The proof of Corollary~\ref{cor:field} requires some additional work. 
The main difficulty lies in controlling the unbounded operators appearing in 
\eqref{eq:expectation-N} and \eqref{eq:field-aKG}. 
To prepare for this, we establish two auxiliary lemmas. 
The first provides an improved estimate on the propagated field energy, while the second yields an a priori bound on 
\(\| N_{\le L}\,\xi_K(t)\|\), where the cut-off number operator is given by 
\[
N_{\le L} \;=\;\int_{|k|\le L} \, a_k^* a_k  \, \d k, \qquad 0\le L\le K.
\]

\begin{lemma}\label{lem:improved:field:energy} 
Let $\xi_K(t)$ be defined in \eqref{eq:xi:K:2}. 
There exist constants $C,T >0$, $K_0\ge 2$ such that
\begin{align}
\big\|  H_f^{1/2}   \xi_K(t) \big\|^2  \;\le\; C \,\ve^{-2} \ln K 
\end{align}
for all $|t| \le T$, $\ve>0$ and $K \ge K_0$.
\end{lemma}This estimate improves the field energy bound from Lemma~\ref{lemma:time-dep:energy:bound}.

\begin{proof}[Proof of Lemma \ref{lem:improved:field:energy}]
As in the initial steps of the proof of Lemma~\ref{lemma:time-dep:energy:bound}, we start from
\begin{align}
 \big\| \ve H_f^{1/2} \xi_K(t) \big\|   
   &\le \Big\|\big( \ve^2 H_f + 1 \big)^{1/2}  
   e^{-i t \ve^{-2} H^{ \ve }_{\mathscr D,K} } 
   W(\ve^{-1} \varphi_0 )  e^{ i g_{K,\varphi_0} (x)}  
   \psi_{\varphi_0} \otimes \Omega \Big\| .
\end{align}
By Lemma~\ref{lem:dressed:Nelson}, for $L,K \ge K_0$ we may rewrite
\begin{align}\label{eq:dressed:L:K}
H^{ \ve }_{\mathscr D,K}  
   \;=\; U_K^\ve (U_L^\ve)^* \, H^\ve_{\mathscr D,L}\, U_L^\ve (U_K^\ve)^* .
\end{align}
Moreover, one has the quadratic form bound
\begin{align}
U_L^\ve (U_K^\ve)^* \, \ve^2 H_f \, U_K^\ve (U_L^\ve)^* 
   \;\lesssim\; \ve^2 H_f \,+\, \ve^4\,(K^{-1/2} + L^{-1/2}) 
   \;\lesssim\; \ve^2 H_f + 1 ,
\end{align}
which follows from $\|\omega^{-1/2} B_K\|_{L^2} \lesssim K^{-1/2}$, see \eqref{eq:bounds:B}.
Hence,
\begin{align}
\big\| \ve H_f^{1/2} \xi_K(t) \big\|   
   &\lesssim \Big\| \big( \ve^2 H_f + 1 \big)^{1/2}  
   e^{-i t \ve^{-2} H^{ \ve }_{\mathscr D,L } } 
   U_L^\ve (U_K^\ve)^*   
   W(\ve^{-1} \varphi_0 )  e^{ i g_{K,\varphi_0} (x)}  
   \psi_{\varphi_0} \otimes \Omega \Big\| .
\end{align}
Using $\ve^2 H_f \lesssim H^{ \ve }_{\mathscr D,L } + L$, cf.~\eqref{eq:energy:bound}, we can remove the propagator and obtain
\begin{align} 
\big\| \ve H_f^{1/2} \xi_K(t) \big\|   
   &\lesssim \Big\| \big( H^{ \ve }_{\mathscr D,L } + L \big)^{1/2}  
   U_L^\ve (U_K^\ve)^*   
   W(\ve^{-1} \varphi_0 )  e^{ i g_{K,\varphi_0} (x)}  
   \psi_{\varphi_0} \otimes \Omega \Big\| .
\end{align}
Next, invoking \eqref{eq:dressed:L:K} once more and the bound 
\begin{align}
H^{ \ve }_{\mathscr D,K} \;\lesssim\; p^2 + \ve^2 H_f + \ln K \,(N+1) ,
\end{align}
from \eqref{eq:energy:bound:b}, we deduce
\begin{align} 
\big\| \ve H_f^{1/2} \xi_K(t) \big\|   
   &\lesssim \Big\| \big( p^2 + \ve^{-2} H_f + \ve^2 \ln K (N+1) + L \big)^{1/2}  
   W(\ve^{-1} \varphi_0 )  e^{ i g_{K,\varphi_0} (x)}  
   \psi_{\varphi_0} \otimes \Omega \Big\| .
\end{align}
At this point, the argument proceeds in analogy to \eqref{eq:lemma7:1:secong:part}, yielding
\begin{align}
\| \ve H_f^{1/2} \xi_K(t) \| \;\lesssim\; L^{1/2} + (\ln K)^{1/2}.
\end{align}
Since $L\ge K_0$ was arbitrary, the claim follows.
\end{proof}
\begin{lemma}\label{lem:aux:number:operator:bound} 
Let $\xi_K(t)$ be defined by \eqref{eq:xi:K:2}. There are $C,T>0$ and $K_0>2$, such that
\begin{align}
 \| N_{\le L} \xi_K(t) \| \;\le\; C\, \ve^{-2} L^2
\end{align}
for all $|t|\le T$, $\ve>0$, $K\ge K_0$, and $0\le L \le K$.
\end{lemma}

\begin{proof}
Define
\begin{align}
f(t) \;:=\; \ve^4 \|N_{\le L}\xi_K(t)\|^2 + \ve^2 L^2 \|N_{\le L}^{1/2}\xi_K(t)\|^2.
\end{align}
Since $\xi_K(0) = \psi_{\varphi_0}\otimes\Omega$, we have $f(0)=0$. With the fundamental theorem of calculus and using
$i\partial_t \xi_K(t) = \mathscr L_{\mathscr D,K}^\ve(t)\xi_K(t)$, we obtain $f(t) = (\mathrm{I})+(\mathrm{II})$ with
\begin{align}
(\mathrm{I}) &= i \ve^4 \int_0^t \! \d s \,\langle \xi_K(s), [\mathscr L_{\mathscr D,K}^\ve(s), N_{\le L}^2] \,\xi_K(s)\rangle, \\
(\mathrm{II}) &= i \ve^2 L^2 \int_0^t \! \d s \,\langle \xi_K(s), [\mathscr L_{\mathscr D,K}^\ve(s), N_{\le L}] \,\xi_K(s)\rangle.
\end{align}
Next, note that all terms in $\mathscr L_{\mathscr D,K}^\ve(t)$ except $\ve^{-1}\phi(\bold G_K)$ commute with $N_{\le L}$, since $[N_{\le L},H_f]=0$ and all remaining field operators are supported in $\{|k|\ge K\}\subset\{|k|\ge L\}$. Hence
\begin{align}
[\mathscr L_{\mathscr D,K}^\ve(t), N_{\le L}] &= \ve^{-1}\phi^-(\bold  G_L), \\
[\mathscr L_{\mathscr D,K}^\ve(t), N_{\le L}^2] &= \ve^{-1}\big(N_{\le L}\phi^-(\bold  G_L)+\phi^-(\bold  G_L)N_{\le L}\big),
\end{align}
with $\phi^-(\bold  G_L) = a^*(\bold  G_L)-a(\bold G_L)$. 

Using $\|\phi(\bold  G_L)\xi_K(s)\|\lesssim L\,\|(N_{\le L}+1)^{1/2}\xi_K(s)\|$, we estimate
\begin{align}
|(\mathrm{I})| \;\lesssim\; \ve^3 L \int_0^t \! \d s \,\|N_{\le L}\xi_K(s)\| \, \|(N_{\le L}+1)^{1/2}\xi_K(s)\|
   \;\lesssim\; \int_0^t f(s)\,\d s + \ve^2.
\end{align}
Similarly, since $\|a(\bold  G_L)\xi_K(s) \|\lesssim L\,\|N_{\le L}^{1/2}\xi_K(s) \|$, we find
\begin{align}
|(\mathrm{II})| \;\lesssim\; \ve L^3 \int_0^t \! \d s \,\|N_{\le L}^{1/2}\xi_K(s)\|
   \;\le\; \int_0^t f(s)\,\d s + L^4.
\end{align}
Collecting terms, we obtain
\begin{align}
f(t) \;\lesssim\; \int_0^t f(s)\,\d s + L^4.
\end{align}
By Grönwall’s inequality and $f(0)=0$, this yields $f(t)\lesssim L^4$, as claimed.
\end{proof}

We now prove Corollary~\ref{cor:field}, as a consequence of the preceding lemmas together with the bound established in \eqref{eq:conclusion}.  

\begin{proof}[Proof of Corollary \ref{cor:field}] As the implication from \eqref{eq:expectation-N} to \eqref{eq:field-aKG} is well-known and follows by standard arguments, we omit the proof and concentrate on \eqref{eq:expectation-N}.

We abbreviate
\begin{align}
g(t) := \Big\|    N^{1/2}  W (\ve^{-1}\varphi_t)^*\,e^{-i\,\ve^{-2} t H^\ve}\Psi_{K}^\ve(\varphi_0) \Big\|^2 ,
\end{align}
where $\Psi_{K}^\ve(\varphi_0) = W(\ve^{-1} \varphi_0 ) (U_K^\ve)^* \psi_{\varphi_0} \otimes \Omega$.  
Using the bound $
U_K^\ve  N (U_K^\ve)^* \;\lesssim\; N + \ve^2 \|B_K \|_{L^2}^2$,
we obtain
\begin{align}
g(t) \;\lesssim\; \Big\| N^{1/2} U_K^\ve  W (\ve^{-1}\varphi_t)^*\,e^{-i\,\ve^{-2} t H^\ve}\Psi_{K}^\ve(\varphi_0) \Big\|^2 \;+\; \ve^2 K^{-1}.
\end{align}
With \eqref{eq:xi:K:2} and \eqref{eq:def:xi(t)} we can write
\[
U_K^\ve   W(\ve^{-1} \varphi_t )^* e^{-i t \ve^{-2} H^{\ve}}  \Psi_{K}^\ve (\varphi_0) 
  \;=\; e^{-i \ve^{-2} \mu(t)} \, \xi_K(t) .
\]
We now split the number operator into $N= N_{\le L} + N_{L}$ with $ N_{L} = \int_{|k| >  L}  a_k^* a_k \, \d k$ for some $0 <  L\le K$.  
Since $N_{L} \le  L^{-1} H_f$ for $L>1$, Lemma~\ref{lem:improved:field:energy}  implies
\begin{align}
 \big\langle   \xi_K(t) ,  N_{L}   \xi_K(t)   \big\rangle 
   \;\le\; L^{-1}  \big\langle   \xi_K(t) ,  H_f \xi_K(t)   \big\rangle 
   \;\le\; \ve^{-2} (\ln K) L^{-1}.
\end{align}
Thus
\begin{align}\label{eq:intermediate-polish}
g(t)  \;\lesssim\; \langle  \xi_K( t),  N_{\le L} \xi_K(t) \rangle  
   +  \ve^{-2}  (\ln K)\, L^{-1}  +  \ve^2 K^{-1}.
\end{align}

To estimate the first term, recall $\zeta(t) = \psi_{\varphi_t} \otimes\mathbb U(t) \Omega$, with $\| N \zeta(t) \| \lesssim 1$ by Lemma~\ref{lem:Bog:evol:N:bounds}.  
By Lemma~\ref{lem:aux:number:operator:bound} we obtain
\begin{align}\label{eq:bound-Nle}
 \<  \xi_K (t) ,  N_{\le L}  \xi_K (t) \>  
 &= \langle \xi_K (t) - \zeta(t) ,  N_{\le L} \xi_K(t)  \rangle 
   +  \langle \zeta(t) ,  N_{\le L} \xi_K(t) \rangle \notag\\[1mm]
 &\le \| \xi_K (t)  - \zeta(t)   \| \, \| N_{\le L} \xi_K(t)  \| + \|  N_{\le L}   \zeta(t) \| \notag\\[1mm]
 &\lesssim \ve^{-2} L^2  \| \xi_K (t)  - \zeta(t)   \|  + 1.
\end{align}
Recalling from \eqref{eq:conclusion} that  for $K=\ve^{-4/5}$
\begin{align}
\|\delta(t)\|=\|\xi_K(t)-\zeta(t)\|\;\le\; \ve^{1/10} \, \ln  \tfrac{1}{\ve} ,
\end{align}
we find
\begin{align}
 \<  \xi_K (t) ,  N_{\le L}  \xi_K (t) \> \;\lesssim\; \ve^{-19/10} (\ln  \tfrac{1}{\ve} )\,L^2 + 1 .
\end{align}
Setting $K=\ve^{-4/5}$ in \eqref{eq:intermediate-polish} and combining with the above bound this  yields
\begin{align}
    \<  \xi_K  (t) , N  \xi_K(t) \> 
       \;\lesssim\;   \ve^{-19/10} (\ln \tfrac{1}{\ve} )\, L^2 + 1   
         +  \ve^{-2} (\ln  \tfrac{1}{\ve} )\,L^{-1}  
         +  \ve^{1/5}.
\end{align}
Optimizing over $L$ gives $L=\ve^{-1/30} \le K$, which proves \eqref{eq:expectation-N}.  
\end{proof}

\appendix

\section{Renormalization and dressing identity}

\label{app:renormalization}

To keep the paper self-contained, we briefly present the proofs of Proposition~\ref{prop:ren:Nelson}
and Lemma~\ref{lem:dressed:Nelson} for the Nelson Hamiltonian with semiclassical scaling. Since these
follow well-known arguments \cite{Griesemer18}, we only sketch the main steps.

\begin{proof}[Proof of Proposition~\ref{prop:ren:Nelson}]
We introduce the unitary dressing transformation \eqref{eq:Gross:trafo} with ultraviolet cutoff $\Lambda\ge K$:
\begin{align}
U_{K,\Lambda}^\varepsilon
:=\exp\!\Big(\varepsilon\,a^\ast(\boldsymbol{B}_{K,\Lambda})-\varepsilon\,a(\boldsymbol{B}_{K,\Lambda})\Big),
\qquad
\boldsymbol{B}_{K,\Lambda}(x,k):=\boldsymbol{B}_K(x,k)\,\mathbf{1}_{\{|k|\le\Lambda\}} .
\end{align}
Conjugating each term of the regularized Hamiltonian \eqref{eq:regularized:Nelson} yields
\begin{align}
U_{K,\Lambda}^\varepsilon\,p^2\,(U_{K,\Lambda}^\varepsilon)^\ast
&=p^2+\varepsilon\!\left(2 a^\ast(k\boldsymbol{B}_{K,\Lambda})\cdot  p
+2 p \cdot a(k\boldsymbol{B}_{K,\Lambda})-\phi(k^2\boldsymbol{B}_{K,\Lambda})\right) \notag\\[1mm]
&\hspace{-1.8cm} +\varepsilon^2 \left(2\,a^\ast(k\boldsymbol{B}_{K,\Lambda})a(k\boldsymbol{B}_{K,\Lambda})
+a(k\boldsymbol{B}_{K,\Lambda})^2+a^\ast(k\boldsymbol{B}_{K,\Lambda})^2+\|kB_{K,\Lambda}\|_{L^2}^2\right),\label{rel1}\\[1mm]
U_{K,\Lambda}^\varepsilon\,(\varepsilon^2 H_f)\,(U_{K,\Lambda}^\varepsilon)^\ast
&=\varepsilon^2 H_f-\varepsilon^3\,\phi(\omega\boldsymbol{B}_{K,\Lambda})
+\varepsilon^4\,\|\omega^{1/2}B_{K,\Lambda}\|_{L^2}^2,\label{rel2}\\[1mm]
U_{K,\Lambda}^\varepsilon\,\big(\varepsilon\,\phi(\boldsymbol{G}_\Lambda)\big)\,(U_{K,\Lambda}^\varepsilon)^\ast
&=\varepsilon\,\phi(\boldsymbol{G}_\Lambda)-2\,\varepsilon^2\,\Re\langle \boldsymbol{G}_\Lambda,\boldsymbol{B}_{K,\Lambda}\rangle.\label{rel3}
\end{align}
The scalar $\varepsilon^2$-terms combine to
\begin{align}
\varepsilon^2\|kB_{K,\Lambda}\|_{L^2}^2-2\varepsilon^2\Re\langle \boldsymbol{G}_\Lambda,\boldsymbol{B}_{K,\Lambda}\rangle
=4\pi\varepsilon^2(\ln K-\ln\Lambda),
\end{align}
while the field operator terms linear in $\ve$ add up to $\phi(\boldsymbol{G}_K)$:
\begin{align}
\phi(\boldsymbol{G}_\Lambda)-\phi(k^2\boldsymbol{B}_{K,\Lambda})=\phi(\boldsymbol{G}_K).
\end{align}
Adding $4\pi\varepsilon^2\ln\Lambda$ removes the divergent part, and we obtain the dressed regularized operator
\begin{align}
H_{\mathscr D,K,\Lambda}^\varepsilon
&:=U_{K,\Lambda}^\varepsilon\Big(H_\Lambda^\varepsilon+4\pi\varepsilon^2\ln\Lambda\Big)(U_{K,\Lambda}^\varepsilon)^\ast \notag\\
&=p^2+\varepsilon\Big(\phi(\boldsymbol{G}_K)+2\,a^\ast(k\boldsymbol{B}_{K,\Lambda}) \cdot  p+2\,p \cdot  a(k\boldsymbol{B}_{K,\Lambda})\Big) \notag\\
&\quad+\varepsilon^2\Big(H_f+2\,a^\ast(k\boldsymbol{B}_{K,\Lambda})a(k\boldsymbol{B}_{K,\Lambda})
+a(k\boldsymbol{B}_{K,\Lambda})^2+a^\ast(k\boldsymbol{B}_{K,\Lambda})^2
+4\pi\ln K\Big) \notag\\
&\quad-\varepsilon^3\,\phi(\omega\boldsymbol{B}_{K,\Lambda})
+\varepsilon^4\,\|\omega^{1/2}B_{K,\Lambda}\|_{L^2}^2. \label{HKL}
\end{align}
Using the bounds \eqref{op1}--\eqref{op4} (which apply to $B_{K,\Lambda}$ as well), for any $\delta>0$ there exist $C_\delta,K_0>0$ such that for all $\Lambda\ge K\ge K_0$ and $\varepsilon>0$,
\begin{align}
\pm\big(H_{\mathscr D,K,\Lambda}^\varepsilon-p^2-\varepsilon^2H_f\big)
\;\le\;\delta\,(p^2+\varepsilon^2 H_f)+C_\delta K  
\end{align}
as quadratic forms. By the KLMN theorem, $H_{\mathscr D,K,\Lambda}^\varepsilon$ is thus self-adjoint and semibounded, with form domain $Q(p^2+H_f)$.

Moreover, following \cite[Lem.~3.2]{Griesemer18} and using again \eqref{op1}--\eqref{op4}, 
there exist constants $C>0$ and $\Lambda_0\ge K_0$ such that, for all 
$\Lambda_2\ge\Lambda_1\ge\Lambda_0$ and $K\ge K_0$, 
\begin{align}\label{eq:form:bound:2}
\pm\big(H_{\mathscr D,K,\Lambda_2}^\varepsilon-H_{\mathscr D,K,\Lambda_1}^\varepsilon\big)
\;\le\; C\,\Lambda_1^{-1/4}\,(p^2+\varepsilon^2 H_f).
\end{align}
By \cite[Thm.~A.1]{Griesemer18} there exists a self-adjoint, semibounded operator
$H_{\mathscr D,K}^\varepsilon:=H_{\mathscr D,K,\infty}^\varepsilon$ with form domain $Q(p^2+H_f)$ such that 
\begin{align}
s - \lim_{\Lambda \to \infty} e^{-it H_{\mathscr D,K,\Lambda}^\varepsilon} = e^{-it H_{\mathscr D,K}^\varepsilon}
\qquad (\forall t\in\mathbb R).
\end{align}
Since $U_{K,\Lambda}^\varepsilon\xrightarrow
[\Lambda\to\infty]{}U_{K,\infty}^\varepsilon$ in the strong sense, we define
\begin{align}\label{eq:app:dressing:identity}
H^\varepsilon:=(U_{K,\infty}^\varepsilon)^\ast\,H_{\mathscr D,K}^\varepsilon\,U_{K,\infty}^\varepsilon,
\end{align}
which satisfies
\begin{align}
s-\lim_{\Lambda \to \infty} e^{-it H_\Lambda^\varepsilon}\,e^{-it\,4\pi\varepsilon^2\ln\Lambda} =  e^{-it H^\varepsilon}\qquad (t\in\mathbb R).
\end{align}
This proves Proposition~\ref{prop:ren:Nelson}.
\end{proof}

\begin{proof}[Proof of Lemma \ref{lem:dressed:Nelson}]
For $K\ge K_0$, the dressing identity in Lemma~\ref{lem:dressed:Nelson} is exactly \eqref{eq:app:dressing:identity}. 
The energy bounds \eqref{eq:energy:bound}–\eqref{eq:energy:bound:b} follow directly from \eqref{op1}--\eqref{op4}. 
\end{proof}

\section{Global existence of aKG near minimizers}
\label{app:persistence}
In this appendix we prove Lemma \ref{lem:global} on global well-posedness of the aKG dynamics for initial data close to the manifold of minimizers of the field energy. Recall the field functional
\begin{align}
\mathcal E_{\mathrm{field}}(u) := e(u) + \|u\|_{\mathfrak h_{1/2}}^2,
\end{align}
with minimal value $E_*=\mathcal E_{\mathrm{field}}(\varphi_*)$ attained by minimizers $\varphi_*=-\omega^{-1}\sigma(\psi_*)$, where $\psi_*$ is a minimizer of the Pekar functional (known to be unique up to translations \cite{Lieb1977}),
\begin{align}
\mathcal E_{\rm P}(\psi) = \langle \psi , p^2 \psi \rangle - \| \omega^{-1/2} \sigma(\psi) \|_{L^2}^2.
\end{align}
The manifold of field minimizers is
\begin{align}
\mathcal M := \{\,T_y\varphi_*: y\in\mathbb R^3\,\},\qquad (T_y u)(k):=e^{-ik\cdot y}u(k).
\end{align}
The proof of Lemma \ref{lem:global} is based on the following coercivity estimate for the field functional.
\begin{lemma} \label{lem:field-coercivity}
There exists $c>0$ such that
\begin{align}\label{eq:field-coercivity}
\mathcal E_{\mathrm{field}}(u)-E_* \;\ge\; c\ \mathrm{dist}_{\mathfrak h_{1/2}}(u,\mathcal M)^2
\qquad \forall\,u\in \mathfrak h_{1/2}.
\end{align}
\end{lemma}

This follows from the known coercivity of the Pekar functional (modulo translations and phase, see \cite{Feliciangeli}). We omit the proof.

\begin{proof}[Proof of Lemma \ref{lem:global}] 
We start by noting that for any $u\in \mathfrak h_{1/2}$ and normalized $\psi\in H^1$, the variational principle gives
\begin{align}
\mathcal E_{\mathrm{field}}(u)\;\le\;\mathcal E_{\mathrm P}(\psi)
+ \|\omega^{1/2}u+\omega^{-1/2}\sigma(\psi)\|_{L^2}^2.
\end{align}
Choosing $\psi=T_y\psi_*$ and recalling $\varphi_*=-\omega^{-1}\sigma(\psi_*)$ yields
\begin{align}\label{eq:upper-bound}
\mathcal E_{\mathrm{field}}(u)-E_*
\;\le\; \|u-T_y\varphi_*\|_{\mathfrak h_{1/2}}^2,
\end{align}
and taking the infimum over $y$,
\begin{align}
\mathcal E_{\mathrm{field}}(u)-E_*
\;\le\;\mathrm{dist}_{\mathfrak h_{1/2}}(u,\mathcal M)^2.
\end{align}

Now, let $\varphi_t$ be the aKG solution with initial data $\varphi_0$ on the maximal interval $(-T_-,T_+)$. We recall conservation of the energy from \eqref{eq:energy:conservation},
\begin{align}
\mathcal E_{\mathrm{field}}(\varphi_t) & = \mathcal E_{\mathrm{field}}(\varphi_0) \quad \text{for all} \quad t\in (-T_-,T_+). 
\end{align}
Combining \eqref{eq:field-coercivity} with energy conservation and \eqref{eq:upper-bound}, we thus obtain
\begin{align}
\mathrm{dist}_{\mathfrak h_{1/2}}(\varphi_t,\mathcal M)^2 &
\le\; c^{-1}\big(\mathcal E_{\mathrm{field}}(\varphi_t)-E_*\big) \notag\\
& = \; c^{-1}\big(\mathcal E_{\mathrm{field}}(\varphi_0)-E_*\big) \notag \\
&\le\;  c^{-1}\,\mathrm{dist}_{\mathfrak h_{1/2}}(\varphi_0,\mathcal M)^2,
\end{align}
for all $t\in (-T_-,T_+)$.

By assumption, we have $\mathrm{dist}_{\mathfrak h_{1/2}}(\varphi_0,\mathcal M)\le \eta$.  Thus $\varphi_t$ remains in a fixed neighborhood of $\mathcal M$ for all $t\in (-T_-,T_+)$, meaning that for every such $t$ there exists $\varphi_*\in \mathcal M$ with
\begin{align}
\|\varphi_t - \varphi_*\|_{\mathfrak h_{1/2}} \;\le\; c^{-1}\eta.
\end{align}
Since the spectral gap is the same for all $\varphi_*\in \mathcal M$ and positive, $\triangle (\varphi_*)>0$, we conclude (for sufficiently small $\eta$ and by Lemma \ref{lem:spectral:gap}) that 
\begin{align}
\triangle(\varphi_t)\ \ge\ \tfrac12 \,\triangle(\varphi_*)\ >\ 0
\qquad\forall\,t\in(-T_-,T_+).
\end{align}
If $T_+<\infty$, then $\liminf_{t\to T_+}\triangle(\varphi_t) >0$, 
contradicting the continuation alternative from Proposition \ref{prop:aKG-local}. Hence $T_+=\infty$. The same argument shows 
$T_-=\infty$, and thus the aKG solution with initial datum sufficiently close to $\mathcal M$ is global.
\end{proof}

\section{Convergence of SKG to aKG as $\ve\to 0$}
\label{app:conv:SKG:aKG}

We begin by recalling well-posedness results for the SKG system.  
Global well-posedness of \eqref{eq:SKG} with $\varepsilon=1$ was proved in \cite[Thm.~1.4]{Pecher} for initial data $(\psi_0,\varphi_0)\in H^s\times\mathfrak h_{s-1/2}$ with $s\ge 0$.  
The solutions conserve mass, $\|\psi_t\|_{L^2}=\|\psi_0\|_{L^2}$, and, for $s\ge 1$, also conserve the semiclassical energy $\mathcal E(\psi_t,\varphi_t)=\mathcal E(\psi_0,\varphi_0)$ defined in \eqref{eq:semiclassical:energy}.  
Moreover, if $(\psi_0,\varphi_0)\in H^1\times\mathfrak h_{1/2}$, then
\begin{align}\label{eq:global:H1:bound:SKG}
\sup_{t\in\mathbb R}\big(\|\psi_t\|_{H^1}+\|\varphi_t\|_{\mathfrak h_{1/2}}\big)<\infty.
\end{align}

For general semiclassical parameter $\varepsilon>0$, global well-posedness in $H^1\times\mathfrak h_{1/2}$, together with the same conservation laws and uniform $H^1\times\mathfrak h_{1/2}$ bounds (uniform in $\varepsilon>0$ and $t\in\mathbb R$), follows by adapting the analysis of the Landau--Pekar equations in \cite[Lemma~2.1]{FrankG2017}.

In the next lemma we show that, on the slow time scale $t=\mathcal O(\varepsilon^{-2})$, SKG solutions converge to the aKG solution pair.

\begin{lemma}
\label{lem:comparision:(non)adiabatic}
    Let $\varphi_0\in \mathfrak h_0 \cap \mathfrak h_{1/2}$ with $e(\varphi_0)<0$. Let $\varphi_t$ be the solution to \eqref{eq:aKG} with initial datum $\varphi_0$ and let $\psi_{\varphi_t}\ge 0$ denote the normalized non-negative ground state of $h_{\varphi_t}$. Moreover, let $(\psi_t^\ve, \varphi_t^\ve)$ denote the solutions to \eqref{eq:SKG} for initial data $(\psi_{\varphi_0} , \varphi_0)$.
Then there exist $C,T,\ve_0>0$, such that 
    \begin{align*}
           \Big\|\exp \left( i \ve^{-2} \int_0^t e \left(\varphi^\ve_{\ve^{-2}s} \right)\, \d s \right) \psi^{\ve}_{\ve^{-2}t}  -  \psi_{\varphi_t}\Big\|_{L^2}+\,  \Big\|\varphi^{\ve}_{\ve^{-2}t}-\varphi_t \Big\|_{L^2} \leq  C \ve^{2/7} 
    \end{align*}
    for all $|t| \le T$ and $\ve \in (0,\ve_0]$, where $e(\varphi^\ve_s)= \inf \sigma (h_{\varphi_s}^\ve)$.
\end{lemma}\begin{proof}
We denote the SKG solutions on the slow time scale by  $
u_t^\ve := \psi^\ve_{\ve^{-2}t}$ and $ w_t^\ve := \varphi_{\ve^{-2} t}^\ve$ and recall that they satisfy
\begin{align}
i\partial_t u_t^\ve = \ve^{-2} h_{w_t^\ve} u_t^\ve, \qquad \qquad i\partial_t w_t^\ve = \omega w_t^\ve + \sigma(u_t^\ve).
\end{align}
For the field we estimate
\begin{align}
\| w_t^\ve - \varphi_0 \|_{\mathfrak h_{1/2}}
\le \|(e^{-i\omega t}-1)\varphi_0\|_{\mathfrak h_{1/2}} 
+ \int_0^{|t|} \| \sigma(u_s^\ve) \|_{\mathfrak h_{1/2}}\,\d s\, ,
\end{align}
and use that the first term approaches zero as $|t| \to 0$ and that the second term is bounded by a constant times $T$ (use Lemma \ref{lem:sigma:bounds} together with $\sup_{\ve >0} \sup_{ t\in \mathbb R } \|u_t^\ve \|_{H^1}\le C$). 

Therefore, by Lemma~\ref{lem:spectral:gap}, there exists $T>0$ such that the spectral gap associated with $w_t^\ve$ remains positive, i.e. $\inf_{\ve >0} \inf_{|t|\le T}\triangle(w_t^\ve)>0$. In particular, for every $|t|\le T$ the operator $h_{w_t^\ve}$ has a unique non-negative ground state $\psi_{w_t^\ve}$ with eigenvalue $e(w_t^\ve)<0$.

The ground state approximates the particle component of the SKG solution via
\begin{align}\label{eq:adiabatic:theorem:SKG}
\Big\| e^{ i \ve^{-2} \int_0^t e(w_s^\ve)\,\d s } u_t^\ve - \psi_{w_t^\ve}\Big\|_{L^2} \le C\ve^{2/7}, \qquad |t|\le T.
\end{align}
This can be interpreted as a nonlinear adiabatic theorem for the SKG dynamics, and its proof follows the same strategy as the nonlinear adiabatic theorem for the Landau--Pekar equations in \cite[Thm~II.1]{LeopoldRSS2019}. For completeness, we sketch the proof of \eqref{eq:adiabatic:theorem:SKG} at the end of this section.

Moreover, since $\|w_t^\ve-\varphi_t\|_{\mathfrak h_{1/2}}\le CT$, Lemma~\ref{lem:ground:state:difference} applies to $\psi_{w_t^\ve}\ge0$ and $\psi_{\varphi_t}\ge0$, yielding
\begin{align}
\|\psi_{w_t}^\ve - \psi_{\varphi_t}\|_{H^1} \le C \|w_t^\ve -\varphi_t\|_{\mathfrak h_{1/2}}.
\end{align}
Combining this with \eqref{eq:adiabatic:theorem:SKG} gives
\begin{align}\label{eq:L2:akG:SKG}
\Big\| e^{ i \ve^{-2} \int_0^t e(w_s^\ve )\,\d s } u_t^\ve - \psi_{\varphi_t}\Big\|_{L^2}
\le C\ve + C \|w_t^\ve-\varphi_t\|_{L^2}.
\end{align}

Finally, for $|t|\le T$ we estimate
\begin{align}
\|w_t^\ve -\varphi_t\|_{L^2}
&\le \int_0^{|t|} \|\sigma(u_s^\ve )-\sigma(\psi_{\varphi_s})\|_{L^2}\,\d s \notag\\
&\le C \int_0^{|t|} \Big\|e^{ i \ve^{-2} \int_0^s e(w_r^\ve )\,\d r } u_s - \psi_{\varphi_s}\Big\|_{L^2}\,\d s \notag\\
&\le C \ve^{2/7} + C\int_0^{|t|}\|w_s^\ve -\varphi_s\|_{L^2}\,\d s,
\end{align}
where we used that $\sigma(\psi)$ is quadratic in $\psi$. By Grönwall’s inequality we obtain $\|w_t^\ve -\varphi_t\|_{L^2}\le C\ve^{2/7} $, and inserting this into \eqref{eq:L2:akG:SKG} proves the claimed bounds.
\end{proof}

\begin{proof}[Proof of Inequality \eqref{eq:adiabatic:theorem:SKG}]
We briefly sketch the proof of the adiabatic bound \eqref{eq:adiabatic:theorem:SKG}.  
Since the spectral gap $\triangle (w_t^\ve)>0$ persists uniformly for $|t|\le T$ and $\ve >0$, all estimates from Lemma~\ref{lem:resolvent:bounds} apply to $w_t^\ve$. In particular, there exists $C>0$ such that
\begin{align}
\|R_{w_t^\ve}\| + \|\<p\>R_{w_t^\ve}^{1/2}\| + \|\<p^2\>R_{w_t^\ve}\| + \|\psi_{w_t^\ve}\|_{H^2} &\le C, \\[1mm]
\|\dot \psi_{w_t^\ve}\|_{L^2} + \|\dot\sigma(\psi_{w_t^\ve})\|_{L^2} + \|\dot R_{w_t^\ve}\| + \|\<p\>\dot R_{w_t^\ve}\<p\>\| &\le C ,
\end{align}
for all $|t|\le T$ and $\ve> 0$. Here, $R_{w_t^\ve} = Q_{w_t^\ve} (h_{w_t^\ve} - e(w_t^\ve) )^{-1} Q_{w_t^\ve}$ with $Q_{w_t^\ve} = \bold 1 - | \psi_{w_t^\ve} \>\< \psi_{w_t^\ve}|$. Moreover,
\begin{align}
\partial_t \psi_{w_t^\ve} = R_{w_t^\ve} V_{i\omega w_t^\ve} \psi_{w_t^\ve}.
\end{align}

Define $\widetilde u_t^\ve := e^{i\int_0^t e(w_s^\ve)\,\d s}\,u_t^\ve$. Then
\begin{align}
\|\widetilde u_t^\ve - \psi_{w_t^\ve}\|_{L^2}^2
&=2\,\Im\int_0^t \<\widetilde u_s^\ve ,\;\big(\ve^{-2}(h_{w_s^\ve }-e(w_s^\ve ))-iR_{w_s}V_{i\omega w_s^\ve }\big) \psi_{w_s^\ve}\>\,\d s \notag\\
&=-2\,\Im\int_0^t \<\widetilde u_s^\ve ,\;R_{w_s^\ve }V_{i\omega w_s^\ve }\psi_{w_s^\ve }\>\,\d s ,
\end{align}
where we used $h_{w_s^\ve } \psi_{w_s^\ve } = e(w_s^\ve ) \psi_{w_s^\ve }$. The remaining term is estimated analogously to the contribution in $\mathscr R^{(4)}$ in \eqref{eq:R:4}. Decomposing $w_s^\ve=w_{s, L}^\ve+w_{s,>L}^\ve$ with $w_{s,L}^\ve := w_s^\ve \bold 1_{|k|\le L}$, the momentum tail satisfies
\begin{align}
\bigg| \int_0^t \<\widetilde u_s^\ve,\;R_{w_s^\ve}V_{i\omega w^\ve_{s,>L}}\psi_{w_s^\ve}\>\,\d s\bigg|\;\lesssim\;L^{-1/4}.
\end{align}
For the low-momentum part, we need to expand one time via the following identity
\begin{align}
Q_{w_s^\ve }\widetilde u_s^\ve
=R_{w_s^\ve}(h_{w_s^\ve}-e(w_s^\ve))\widetilde u_s^\ve
=i\ve^2\, R_{w_s^\ve}\, \partial_s\widetilde u^\ve_s,
\end{align}
which can be viewed as a classical version of \eqref{eq:identity:Qixi}. Note that we used $i\partial_t \widetilde u^\ve_t = (h_{w_t^\ve}-e(w_t^\ve)) \widetilde u^\ve_t$.

After integrating by parts we obtain
\begin{align}
\|\widetilde u_t^\ve - & \psi_{w_t^\ve}\|_{L^2}^2
=-2\,\Im\int_0^t \<i\partial_s\widetilde u_s^\ve ,\;R_{w_s^\ve}^2V_{i\omega w_{s,L}^\ve} \psi_{w_s^\ve}\>\,\d s \notag\\[1mm]
&=2\ve^2\,\Im\<\widetilde u_s,\;R_{w_s^\ve}^2 V_{i\omega w_{s,L}^\ve} \psi_{w_s^\ve}\>\Big|_0^t \notag\\[1mm]
&\quad+2\ve^2\,\Im\int_0^t \<\widetilde u_s^\ve ,\;(2\dot R_{w_s^\ve} R_{w_s^\ve} +\dot V_{i\omega w^\ve_{s,L} }+R_{w_s^\ve}^2V_{i\omega w_{s,L}^\ve}R_{w_s^\ve}V_{i\omega w_{s,L}^\ve}) \psi_{w_s^\ve}\>\,\d s .
\end{align}
Since these terms are completely analogous to the contributions $\mathscr R^{(4.32)}$ and $\mathscr R^{(4.24)}$, we obtain for $L=\ve^{-8/7}$ that
\begin{align}
\|\widetilde u_t^\ve - \psi_{w_t^\ve}\|_{L^2}^2 \;\lesssim\; L^{-1/4}+\ve^2L^{3/2} \;\lesssim\;\ve^{2/7}.
\end{align}
This proves \eqref{eq:adiabatic:theorem:SKG}.
\end{proof}

\section*{Acknowledgements}

Morris Brooks gratefully acknowledges support from the ERC Advanced Grant CLaQS, grant agreement No. 834782.
 
\end{spacing}


\begin{thebibliography}{11}


\footnotesize{


\bibitem{Ammari} Z. Ammari, and M. Falconi. Bohr’s correspondence principle for the renormalized Nelson model. \textit{SIAM J. Math. Anal.} 49(6), (2016), 5031--5095.

\vspace{-1mm}

\bibitem{BornOppenheimer}
{M. Born, and  R. Oppenheimer.}
Zur Quantentheorie der Molekeln. 
\textit{Ann. Phys.}, 389 (1927), 457--484.


\vspace{-1mm}

\bibitem{Cardenas} E. C\' ardenas, and D. Mitrouskas. The renormalized Nelson model in the weak coupling limit. \textit{Journal of Physics A}, 58 (2024), 175201.

\vspace{-1mm}

\bibitem{CFO22} M. Correggi, M. Falconi, and M. Olivieri. Quasi-classical dynamics.  \textit{J. Eur. Math. Soc.}, 25, 2 (2023), 731--783.

\vspace{-1mm}

\bibitem{Davies} E.B. Davies. Particle-boson interactions and the weak coupling limit. \textit{J. Math. Phys.} 1, 20 (3), (1979) 345--351.

\vspace{-1mm}

\bibitem{Falconi2} M. Falconi. Classical limit of the Nelson model with cutoff. \textit{J. Math.
Phys.} 54(1), (2013), 012303. 

\vspace{-1mm}


\bibitem{Falconi3} M. Falconi, N. Leopold, D. Mitrouskas and J. Lampart.
Renormalized Bogoliubov theory for the Nelson model. \textit{ Ann. Inst. H. Poincaré C Anal. Non Linéaire} (2025), published online first

\vspace{-1mm}


\bibitem{Falconi1} M. Falconi, N. Leopold, D. Mitrouskas, and S. Petrat. Bogoliubov dynamics and higher-order corrections for the regularized Nelson model. \textit{
Rev. Math. Phys}, 33 (2023), 2350006.

\vspace{-1mm}

\bibitem{Feliciangeli} D. Feliciangeli, S. Rademacher, and R. Seiringer. Persistence of the spectral gap for the Landau--Pekar equations. \textit{Lett Math Phys 111}, 19 (2021).

\vspace{-1mm}

\bibitem{FrankG2017} R.L.\ Frank, and Z.\ Gang. Derivation of an effective evolution equation for a strongly coupled polaron. \textit{Anal. PDE}, 10, (2017), 379--422.


\vspace{-1mm}

\bibitem{FrankS2014} R.L.\ Frank, and\ B.\ Schlein. Dynamics of a strongly coupled polaron. \textit{Lett. Math. Phys.,} 104, (2014). 911--929. 

\vspace{-1mm}

\bibitem{Ginibre06} J. Ginibre, F. Nironi, and G. Velo. Partially Classical Limit of the Nelson Model. \textit{Ann. Henri Poincar\'{e}} 7 (2006), 21--43.

\vspace{-1mm}

\bibitem{Griesemer17} M. Griesemer. On the dynamics of polarons in the strong-coupling limit.\ \textit{
Rev. Math. Phys. 29} (10), (2017), 1750030.
\vspace{-1mm}

\bibitem{Griesemer18} M. Griesemer, and A. W\"unsch. On the domain of the Nelson Hamiltonian. \textit{J. Math. Phys.} 59 (2018), 042111.

\vspace{-1mm}

\bibitem{GubinelliHJ2014} 
{M. Gubinelli, F. Hiroshima,  and J. L\"orinczi.}
Ultraviolet renormalization of the Nelson Hamiltonian through functional integration. \textit{J. Func. Anal.} 267 (2014), 3125--3153. 


\vspace{-1mm}


\bibitem{Hagedorn1} G.A. Hagedorn. A time dependent Born-Oppenheimer approximation. \textit{Commun. Math. Phys.} 77, (1980), 1--19.

\vspace{-1mm}

\bibitem{Hagedorn2} G.A. Hagedorn. High order corrections to the time-dependent Born-Oppenheimer approximation I: smooth potentials. \textit{Ann. Math.} 124, 3 (1986), 571--590.

\vspace{-1mm}

\bibitem{Hagedorn3} G.A. Hagedorn, and A. Joye. A time-dependent Born-Oppenheimer approximation with exponentially small error estimates. \textit{Comm. Math. Phys.}, 223 (2001), 583--626.

\vspace{-1mm}

\bibitem{Hiroshima} F. Hiroshima. Weak coupling limit and removing an ultraviolet cutoff for a Hamiltonian of particles interacting with a quantized scalar field. \textit{J. Math. Phys.} 40(3), (1999) 1215--1236.

\vspace{-1mm}

\bibitem{KleinEtAl} M. Klein, A. Martinez, R. Seiler, and X.P. Wang. On the Born--Oppenheimer expansion for polyatomic molecules. \textit{Commun.
Math. Phys.}, 143 (1992), 607--639.

\vspace{-1mm}

\bibitem{LS2019} J.~Lampart, and J. Schmidt. On Nelson-type Hamiltonians and abstract boundary conditions. \textit{Comm. Math. Phys.}, 376(2) (2019), 629--663. 


\vspace{-1mm}

\bibitem{Leopold1} N. Leopold. Norm approximation for the Fr\"ohlich dynamics in the mean-field regime. \textit{J. Funct. Anal. 285(4)}, 109979. (2023)

\vspace{-1mm}

\bibitem{Leopold5} N. Leopold. 
Derivation of the Maxwell--Schrödinger and Vlasov--Maxwell equations from non-relativistic QED. Preprint 2024, \href{https://arxiv.org/abs/2411.07085}{https://arxiv.org/abs/2411.07085}

\vspace{-1mm}

\bibitem{Leopold2} N. Leopold, D. Mitrouskas, and R. Seiringer. Derivation of the Landau--Pekar equations in a many-body mean-field limit. \textit{Arch. Ration. Mech. Anal. 240}, 383--417. (2021)

\vspace{-1mm}

\bibitem{LeopoldMRSS2020} 
N.\ Leopold,\ D. Mitrouskas, S.\ Rademacher,\ B.\ Schlein, and R.\ Seiringer. Landau--Pekar equations and quantum fluctuations for the dynamics of a strongly coupled polaron. \textit{Pure Appl. Anal.}, 3(4), (2021) 653--676.



\vspace{-1mm}

\bibitem{LeopoldPetrat} N. Leopold, and S. Petrat. Mean-field dynamics for the Nelson model with fermions. \textit{Ann. H. Poincar\'e} 20(10), (2019), 3471--3508.

\vspace{-1mm}


\bibitem{Leopold3} N. Leopold, and P. Pickl. Derivation of the Maxwell--Schro\"odinger equations from the Pauli--Fierz Hamiltonian. \textit{SIAM J. Math. Anal.} 52(5), 4900--4936. (2020)

\vspace{-1mm}

\bibitem{LeopoldRSS2019} 
N.\ Leopold,\ S.\ Rademacher,\ B.\ Schlein, and R.\ Seiringer. The Landau--Pekar equations: Adiabatic theorem and accuracy. \textit{Anal. PDE}, 14, (2021) 2079--2100. 


\vspace{-1mm}

\bibitem{Lieb1977} E.H. Lieb. Existence and uniqueness of the minimizing solution of Choquard’s nonlinear equation. \textit{Studies in Applied Mathematics 57} (1977), 93.

\vspace{-1mm}

\bibitem{LiebLoss} E.H. Lieb, and M. Loss. Analysis. Second edition, American Mathematical Society, 2001.

\vspace{-1mm}


\bibitem{LNS2015} M. Lewin, P.T. Nam, and B. Schlein. Fluctuations around Hartree states in the mean-field regime. \textit{Am. J. Math.}, 137(6), (2015),1613--1650. 


\vspace{-1mm}

\bibitem{Mit20} D. Mitrouskas. A note on the Fr\"ohlich dynamics in the strong coupling limit. \textit{Lett. Math. Phys.}, 111 (2021), 45.


\vspace{-1mm}

\bibitem{Nelson64} E.~Nelson.~Interaction of nonrelativistic particles with a quantized scalar field. \textit{J. Math. Phys.}, 5 (1964), 1190--1197.

\vspace{-1mm}

\bibitem{PST} {G. Panati, S. Teufel, and H. Spohn}. The time-dependent Born-Oppenheimer approximation. \textit{ESAIM: Math. Model. Numer. Anal.} 41, 2 (2007), 297--314.

\vspace{-1mm}

\bibitem{Pecher} H.\ Pecher. Some new well-posedness results for the Klein--Gordon--Schr\"odinger system. \textit{Diff. Int. Equations 25}(1/2), (2012), 117--142.

\vspace{-1mm}

\bibitem{SpohnTeufel} H. Spohn, and S. Teufel. Adiabatic Decoupling and Time-Dependent Born--Oppenheimer Theory. \textit{Commun. Math. Phys.} 224 (2001), 113--132.

\vspace{1mm}

\bibitem{Tenuta} L. Tenuta and S. Teufel. Effective Dynamics for Particles Coupled to a Quantized Scalar Field. \textit{Commun. Math. Phys.} 280 (2008), 751--805,

\vspace{1mm}

\bibitem{Teufel} S. Teufel. Effective N-body dynamics for the massless Nelson model and adiabatic decoupling without spectral gap. \textit{Ann. H. Poincar\'e} 3 (2002), 939--965.

}
\end{thebibliography}
\end{document}